\documentclass[11pt,letterpaper]{article}

\usepackage[utf8]{inputenc}
\usepackage[T1]{fontenc}
\usepackage{lmodern}
\usepackage{amsmath,amssymb,amsfonts,amsthm}
\usepackage{mathtools}
\usepackage{bm}
\usepackage{bbm}
\usepackage{enumitem}
\usepackage{geometry}
\usepackage{hyperref}
\usepackage{graphicx}
\usepackage{pdflscape} 
\usepackage{subcaption}
\usepackage{caption}
\usepackage{xcolor}
\usepackage{booktabs,tabularx}
\usepackage{tikz}
\usepackage{tikz-feynman}
\tikzfeynmanset{compat=1.1.0}

\newcommand{\R}{\mathbb{R}}
\newcommand{\T}{\mathbb{T}}
\newcommand{\Z}{\mathbb{Z}}
\newcommand{\C}{\mathbb{C}}

\newcommand{\cA}{\mathcal{A}}
\newcommand{\cC}{\mathcal{C}}

\newcommand{\cX}{\mathcal{X}}
\newcommand{\cF}{\mathcal{F}}

\newcommand{\Xadm}{\mathcal{X}_{\mathrm{adm}}}
\newcommand{\XADM}{\hat{\mathcal{X}}_{\mathrm{adm}}}
\newcommand{\adm}{\mathrm{adm}}
\newcommand{\GL}{\mathrm{GL}}

\newcommand{\tr}{\mathrm{tr}}
\newcommand{\Hd}{\mathrm{H}(d)}
\newcommand{\Ed}{\mathrm{E}(d)}
\newcommand{\SEd}{\mathrm{SE}(d)}
\newcommand{\SOd}{\mathrm{SO}(d)}

\DeclareMathOperator{\Tr}{Tr}

\DeclareMathOperator{\Vol}{Vol}

\numberwithin{equation}{section}

\theoremstyle{plain}
\newtheorem{theorem}{Theorem}[section]

\theoremstyle{definition}
\newtheorem{definition}[theorem]{Definition}
\newtheorem{hypothesis}[theorem]{Hypothesis}
\newtheorem{remark}[theorem]{Remark}
\newtheorem{example}[theorem]{Example}

\newtheorem{conjecture}[theorem]{Conjecture}
\newtheorem{lemma}[theorem]{Lemma}

\title{Rotationally invariant dynamical lattice regulators for\\Euclidean quantum field theories}
\author{Tsogtgerel Gantumur\\[0.5ex]
\small McGill University, Montr\'{e}al, QC, Canada\\
\small National University of Mongolia, Ulaanbaatar, Mongolia\\
\small Mongolian Academy of Sciences, Institute of Mathematics and Digital Technology\\[0.5ex]
\small \texttt{gantumur.tsogtgerel@mcgill.ca}}

\date{May 18, 2026}

\begin{document}

\maketitle

\begin{abstract}
We introduce a dynamical lattice regulator for Euclidean quantum field theories on a fixed hypercubic graph $\Lambda\simeq \Z^d$, in which the embedding $x:\Lambda\to\R^d$ is promoted to a dynamical field and integrated over subject to shape regularity constraints.
The total action is local on $\Lambda$, gauge invariant, and depends on $x$ only through Euclidean invariants built from edge vectors (local metrics, volumes, \emph{etc}.), hence the partition function is exactly covariant under the global special Euclidean group $\SEd$ at any lattice spacing.
The intended symmetry restoring mechanism is not rigid global zero modes but short-range
\emph{local twisting} of the embedding that mixes local orientations.
Our universality discussion is conditioned on a short-range geometry hypothesis (SR): after quotienting the global $\SEd$ modes, connected correlators of local geometric observables have correlation length $O(1)$ in lattice units.

We prove Osterwalder--Schrader reflection positivity for the coupled system with embedding $x$ and generic gauge and matter fields $(U,\Phi)$ in finite volume by treating $x$ as an additional multiplet of scalar fields on $\Lambda$.
Assuming (SR), integrating out $x$ at fixed cutoff yields a local Symanzik effective action in which geometry fluctuations generate only $\SOd$--invariant irrelevant operators and finite renormalizations.
For example, in $d=4$ we recover the standard one--loop $\beta$--function in a scalar $\phi^4$ test theory.
Finally, we describe a practical local Monte Carlo update and report \(d=2\)
proof-of-concept simulations showing \(O(1)\)-scale geometry correlations, a direct
\(SO(2)\)-connection diagnostic of short-range local twisting, and evidence for reduced
axis--vs--diagonal cutoff artefacts relative to a fixed lattice at matched bare parameters.
\end{abstract}

\tableofcontents

\section{Introduction}

\subsection{Motivation}

A central difficulty in nonperturbative quantum field theory (QFT) is the tension between
the need for an ultraviolet cutoff and the desire to preserve the continuous symmetries
of the underlying continuum theory. 
Wilson's lattice regularization replaces space--time by a hypercubic graph
with lattice spacing $a$, and computes expectation values via a Euclidean
path integral \cite{Wilson74,Kogut79,RotheBook}.
This construction is local, gauge invariant, and reflection positive, and
it underpins the modern nonperturbative understanding of Yang--Mills theory and quantum chromodynamics.
However, it breaks the Euclidean group $\Ed = \mathrm{O}(d)\ltimes \mathbb{R}^d$
down to the discrete hypercubic subgroup. Full Euclidean symmetry is recovered only in the
continuum limit, and at finite lattice spacing the reduced symmetry manifests itself as
direction--dependent lattice artefacts in correlation functions and derived observables.
By ``rotational artefacts'' we mean precisely this direction dependence at fixed physical separation
(e.g.\ axis versus diagonal separations on the underlying hypercubic stencil), as well as the induced
mixing patterns of operators that would be forbidden by continuum $\mathrm{O}(d)$ symmetry but are allowed by
the hypercubic subgroup at finite lattice spacing.

Several approaches have been developed to mitigate these symmetry violations.
Random lattice methods, dynamical triangulations, and Regge calculus all average over
ensembles of discrete geometries \cite{ChristFriedbergLeeRandom,AmbjornDurhuusJonssonBook,Regge1961}.
A more conservative pseudorandom-lattice approach preserves the hypercubic combinatorial
skeleton, while placing the sites irregularly inside the elementary cells of a reference grid
\cite{ColangeloScrimieriPseudo,ColangeloCosmaiScrimieriPseudo}.
While these ideas show that fluctuating or irregular geometry can reduce anisotropy, they
typically either modify the combinatorial structure of the discretization, tether the geometry
to prescribed cells as quenched data, or sacrifice reflection positivity. In these approaches
the randomness is usually imposed by an external sampling algorithm or by a gravitational
action rather than arising intrinsically from the QFT regulator itself.

The aim of this paper is to develop and analyse a \emph{dynamical lattice regulator} (DLR)
for Euclidean QFT that combines the robust structural features of standard lattice gauge
theory (locality, exact gauge invariance, Osterwalder--Schrader reflection positivity, and a
fixed combinatorial graph) with an exact proper-Euclidean covariance of the geometry sector already at finite lattice spacing.
Concretely, we fix once and for all an abstract hypercubic lattice
\begin{equation}
  \Lambda = \mathbb{Z}^d ,
\end{equation}
or its finite periodic version for numerical work, which carries the usual gauge and
matter fields. We then promote the \emph{embedding} of lattice sites into $\mathbb{R}^d$
to a dynamical field
\begin{equation}
  x : \Lambda \to \mathbb{R}^d,\qquad n \mapsto x(n),
\end{equation}
and integrate $x$ over in the path integral, subject to local shape--regularity constraints
and a local geometry action $S_x[x]$.  The geometric sector is constructed from Euclidean
invariants (distances and inner products), so it is invariant under global translations and
proper rotations,
\begin{equation}
  \SEd=\SOd\ltimes \mathbb{R}^d,
\end{equation}
and it is compatible with the specific reflection involution used in the OS positivity proof.
However, the mere fact that the partition function is globally $\SEd$--covariant is, by itself,
too weak to address the usual lattice concern: the same formal covariance would also hold if the
geometry measure were concentrated near globally rotated/translated copies of a single rigid
embedding.  What matters for suppressing hypercubic cutoff artefacts is \emph{how} the global symmetry
is realised in typical configurations at fixed cutoff.

The operating regime we target is a \emph{short-range geometry}
regime~\hyperref[hyp:SR]{(SR)}. After fixing or quotienting the global
\(\SEd\) zero modes, typical configurations should exhibit short-range
fluctuations of the induced local frames. In such a regime, global
\(\SEd\) covariance cannot be realised ``rigidly''; it is instead felt by
local observables through the rapid mixing of frame orientations from cell
to cell. We refer to this mechanism as \emph{local twisting:} neighbours
in \(\Lambda\) remain neighbours, but their relative physical orientation
varies with short correlation length while their physical separation
fluctuates within the admissible range. The numerical tests in
Section~\ref{sec:numerics} are designed to probe this short-range,
locally twisting regime.

From this perspective, random lattices and dynamical triangulations may be viewed as
\emph{stochastic topology} models, in which either the connectivity
or the triangulation itself fluctuates, often driven by an external random
point process. Pseudorandom lattices keep the hypercubic topology fixed, but the embedding
is quenched and tethered to the cells of a reference grid. In our framework, by contrast, the
topology of the lattice is rigid and hypercubic, while the embedding \(x\) is locally constrained
but globally flexible. The geometric randomness is \emph{internal} to the quantum theory:
\(x\) is a genuine dynamical variable governed by a local action inside the same path integral
as the gauge and matter fields, not an external pre-processing step.

Our goal in this first work is to show that one can construct a conservative,
computationally accessible regulator that
\begin{itemize}
  \item preserves locality, exact gauge invariance, and reflection positivity;
  \item is globally $\SEd$--covariant at finite lattice spacing, and in the twisting regime reduces direction-dependent (axis-vs-diagonal) rotational artefacts relative to a fixed lattice at matched bare parameters;
  \item is expected to lie in the same universality class as standard hypercubic lattice regulators:
        perturbatively we verify this at one loop,
        and in $d=2$ we find matching critical scaling between baseline and dynamical ensembles;
        more generally we formulate a universality statement as Conjecture~\ref{conj:universality} under the short-range geometry hypothesis~(SR).
\end{itemize}

\noindent
From the lattice QCD viewpoint, the dynamical--lattice regulator may be viewed as an \emph{annealed
geometric averaging in coordinate space}, complementary to field--space smoothing techniques such as
link smearing and the Wilson flow.  In gauge theories we expect such averaging to reduce
direction dependent cutoff effects that can ``pin'' extended or topological structures to the lattice
axes, potentially mitigating forms of \emph{topological freezing} at fixed cutoff.  At the same time, because the underlying
graph topology remains hypercubic, issues such as the fermion doubling are not expected to be resolved by the regulator; 
the intended benefit for fermions is instead
a reduction of anisotropic artefacts (e.g.\ taste splitting) through orientation
mixing.  A systematic study of these gauge theory implications is left to future work.

\subsection{Overview of the construction}

We work with an abstract, fixed hypercubic lattice \(\Lambda\), which would typically be a discrete torus
in numerical applications. The usual gauge and matter fields live on \(\Lambda\) exactly
as in Wilson's formulation: link variables \(U_\mu(n)\in G\) on oriented nearest neighbour
links, and matter fields \(\Phi(n)\) in some representation of \(G\) on the sites. 
Here $n\in\Lambda$ denotes the sites, and $\mu\in\{1,\ldots,d\}$ the spatial directions.
The new
ingredient is a \emph{geometry field}
\begin{equation}
x : \Lambda \to \R^d,
\end{equation}
which specifies the physical position \(x(n)\in\R^d\) of each lattice site. From \(x\) we
construct local edge vectors, metrics and cell volumes, and we restrict \(x\) to an admissible
subset \(\cX_{\adm}\) of $(\R^d)^\Lambda$ by imposing standard shape regularity conditions. The admissible
set is chosen so that all local geometries are uniformly comparable to a regular hypercubic mesh with some scale $a>0$.

The full configuration space is thus
\begin{equation}
\cC \;=\; \cX_{\adm} \times G^{\text{links}}\times \cF_{\text{matter}},
\end{equation}
equipped with a product measure \(Dx\,DU\,D\Phi\). A configuration \((x,U,\Phi)\in\cC\) is
weighted by a local action
\begin{equation}
S[x,U,\Phi] \;=\; S_x[x] + S_{\mathrm{fields}}[x,U,\Phi],
\end{equation}
where $S_x[x]$ is a geometry action and $S_{\mathrm{fields}}$ is the usual gauge--matter action
written in terms of the local metric and volume induced by $x$. Both pieces are built from
Euclidean invariants of the embedding data and from gauge invariant combinations of the fields.
The resulting theory on $\R^d$ is globally covariant under translations and
proper rotations, i.e., under the special Euclidean group $\SEd=\mathrm{SO}(d)\ltimes\R^d$. 
It is also compatible with the discrete reflection involution used in the Osterwalder--Schrader
positivity argument, cf. \cite{SymanzikI,SymanzikII}.
This statement concerns the unfixed formulation; in practice one may fix the global translation/rotation
zero modes without affecting covariant expectations or the local twisting diagnostics.

Standard Wilson lattice gauge theory is recovered as a stiff--geometry limit of the dynamical
lattice regulator in which $S_x[x]=+\infty$ unless $x=x_{\mathrm{reg}}$, where $x_{\mathrm{reg}}(n)=a\,n$ is a regular embedding. 
This limit is not Euclidean--invariant: freezing the embedding selects a preferred orientation and breaks the
continuous Euclidean symmetry explicitly down to the hypercubic symmetry group. 
Merely averaging over rigid global rotations of this frozen embedding does not change
the local stencil anisotropy; improving isotropy requires fluctuations that mix orientations locally.
The novelty is that the geometry remains annealed in a way that preserves locality and the
transfer-matrix/OS-positivity framework on the fixed hypercubic graph.
This enables local orientation mixing that can suppress hypercubic anisotropies; the stiff (Wilson)
limit is recovered by sending the geometry stiffness to infinity.

For conceptual questions (Euclidean covariance, reflection positivity, continuum limit) it
is most convenient to formulate the theory on \(\R^d\), with \(\Lambda\) thought of as an
infinite hypercubic lattice and \(x(n)\in\R^d\). For finite volume simulations we instead
interpret \(\Lambda\) as a discrete torus and take \(x(n)\in\T^d\), using shortest distance
representatives to define edge vectors; this torus version is described in \S\ref{sec:torus-version}. In
both cases the underlying combinatorial graph is kept fixed, while the embedding \(x\)
fluctuates subject to admissibility constraints.

Finally, it may help to view the dynamical geometry in purely lattice-theoretic terms.  The new
variables can be counted as \emph{one additional real scalar per oriented link} (equivalently,
$d$ real numbers per site), while the underlying hypercubic graph, nearest--neighbour stencil, and
all gauge/matter variables are exactly those of standard lattice gauge theory.  In this
repackaging the geometry sector merely promotes certain local couplings to depend on these extra
link scalars through $\Ed$--invariant combinations (such as distances and inner products), so the
regulator enlarges the field content only mildly and in a form that is directly compatible with
existing lattice implementations.

\subsection{Main results and scope of this paper}

We summarize the main theoretical and practical conclusions obtained in this paper for dynamical lattice regulators (DLR).

\begin{itemize}
\item \textbf{Global Euclidean covariance and reflection positivity.}
At fixed lattice spacing the coupled theory is exactly $\SEd$--covariant, and we prove
Osterwalder--Schrader reflection positivity for a broad class of gauge and matter actions.
We distinguish this formal global covariance from the physically relevant notion of reduced
direction-dependent cutoff artefacts, which we probe numerically and relate to short-range local
twisting.

\item \textbf{Locality, admissibility and the geometry sector.}
We define an admissible set \(\cX_{\adm}\) of embeddings by local bounds on edge
lengths, angles and orientation, in such a way that the local metrics and volumes are
uniformly comparable to those of a regular hypercubic lattice. We discuss the global
structure of \(\cX_{\adm}\), the choice of principal component, and simple local
Monte Carlo moves for updating \(x\) that preserve admissibility while remaining
ergodic within this component.

\item \textbf{Perturbative universality.}
As a first perturbative check we analyze scalar $\phi^4$ theory on a dynamical lattice.
In the Symanzik framework \cite{SymanzikI,SymanzikII}, and assuming the short-range geometry hypothesis~(SR),
integrating out the embedding fluctuations yields a local $\SEd$-invariant effective action:
the relevant and marginal operators are therefore the same as in the continuum, while cutoff effects
appear only through higher-dimensional $\SOd$-scalar operators.
Expanding around a regular embedding $x(n)=an+\eta(n)$ and treating $\eta$ as auxiliary fields, we argue that,
conditional on~(SR), the induced interactions are local and RG-irrelevant. Moreover, at one loop the running of the
quartic coupling agrees with that of the standard hypercubic discretization, in line with
Conjecture~\ref{conj:universality}.

\item \textbf{Monte Carlo implementation and numerical tests.}
We describe a first proof-of-concept Monte Carlo implementation in \(d=2\) for a scalar
theory on a periodic lattice, using local Metropolis updates for both fields and geometry.
Basic observables (masses, susceptibilities) agree with those from a static lattice at the
same bare parameters within statistical errors, and simple diagnostics of the geometry
sector (edge length, angle and volume distributions) show no pathological concentration
on highly distorted meshes. These tests support the numerical viability of the regulator
and are consistent with the expected universality scenario, but they do not by themselves
establish the short-range geometry hypothesis in \(d=4\).
\end{itemize}

\noindent
We close this summary by clarifying the logical status of the main claims. The
definitions of the admissible geometry sector in Section~\ref{sec:geometry-config},
the gauge invariance and locality properties in Section~\ref{sec:fields-symmetries},
and the reflection-positivity statements in Section~\ref{sec:OS-universality} are
finite-volume statements within the stated admissibility assumptions. By contrast,
the Symanzik and universality discussion is conditional on the short-range geometry
hypothesis~(SR). Accordingly, the continuum-universality statement is formulated as
Conjecture~\ref{conj:universality}, supported here by the perturbative check of
Section~\ref{sec:OS-universality} and by the two-dimensional numerical tests of
Section~\ref{sec:numerics}. The existence of a robust, computationally affordable
short-range geometry regime in four dimensions is not proved in this paper and remains
a central problem for future work.

The rest of the paper is organised as follows.
In Section~\ref{sec:geometry-config} we develop the geometry sector: admissible
embeddings, local shape-regularity, and Monte Carlo moves. Section~\ref{sec:fields-symmetries}
introduces the field actions, demonstrates exact $\SEd$ invariance, and presents
explicit scalar, gauge, and fermion examples. Section~\ref{sec:OS-universality}
establishes reflection positivity and discusses the Symanzik effective theory and
continuum universality. Section~\ref{sec:numerics} describes our initial numerical
experiments, and Section~\ref{sec:conclusions} summarises our findings and outlines
possible extensions.

\subsection{Comparison to existing approaches}

The idea of fluctuating geometry in a lattice regularization has a long history. It is therefore
important to delineate precisely how the present construction differs from earlier proposals.

\paragraph{Random and pseudorandom lattices.}
In the random lattice (RL) approach, one typically
samples a random set of points in a domain $\Omega \subset \R^d$, often from a Poisson process,
constructs a Voronoi tessellation and its dual Delaunay triangulation, and discretizes the field
theory on the resulting random cell complex, cf. \cite{ChristFriedbergLeeRandom,ChristFriedbergLeeWeights}.
As the points move, the connectivity of the lattice changes discontinuously when Voronoi cells flip;
the adjacency matrix is a discontinuous function of the site positions. In this setting the geometry
is random, but its law is largely specified by an external stochastic prescription for the points, and
reflection positivity is hard to control.

The ``pseudorandom lattice'' construction of Colangelo \emph{et al.}
\cite{ColangeloScrimieriPseudo,ColangeloCosmaiScrimieriPseudo} takes a more conservative route:
it aims at improving rotational invariance while retaining the computational advantages of a
hypercubic skeleton. In that construction one draws exactly one site inside each elementary
hypercubic cell of a reference lattice and connects sites belonging to nearest-neighbour cells.
Thus the sites, links, and plaquettes remain in one-to-one correspondence with those of a standard
hypercubic lattice, while the geometric data---link lengths, angles, and related weights---are
quenched random data chosen once from a tuned probability distribution. This can improve the
\emph{statistical} isotropy of simple geometric observables, but the geometry remains externally
prescribed and frozen during the field simulation.

The present construction is close to the pseudorandom-lattice approach in one important respect:
the abstract hypercubic combinatorial skeleton is fixed once and for all. This feature is therefore
not, by itself, the novelty. It is nevertheless useful to separate three structural differences. First,
unlike in random-lattice formulations, the adjacency relation on $\Lambda$ does not depend on the
embedding \(x\). Moving the sites never creates or destroys links, and there are no Voronoi or
Delaunay flips. Second, unlike in the pseudorandom construction, the sites are not tethered to
prescribed elementary cells of a reference grid. DLR imposes local shape-regularity constraints
on the embedding: these control edge lengths, local volumes, orientations, and condition numbers,
but do not require \(x(n)\) to remain inside a fixed bounding box attached to \(n\). Thus admissible
configurations may include large-scale shears, twists, windings, or spiralling grid patterns, provided
the local mesh quality remains controlled. Third, the embedding is not quenched data: the map
\[
  x : \Lambda \to \R^d
\]
is treated as a genuine dynamical field with a local action \(S_x[x]\) in the \emph{same} path
integral as the gauge and matter fields.

Thus DLR is not merely an annealed version of a cell-tethered pseudorandom embedding, nor is it
a random lattice with fluctuating connectivity. It combines a fixed hypercubic graph with a locally
constrained but globally flexible embedding class, and then integrates over that embedding field
in a local, gauge-invariant, reflection-positive framework. We use ``annealed'' here to refer to the
spacetime lattice itself, distinct from the ``quenched approximation'' regarding fermion loops.

\paragraph{Regge calculus and dynamical triangulations.}
Regge calculus keeps a fixed simplicial connectivity and promotes the edge lengths to
dynamical variables, with an action that discretizes the Einstein--Hilbert functional. In
dynamical triangulations (DT) and causal dynamical triangulations, one takes a
complementary viewpoint: the edge lengths are kept fixed (typically to a few discrete values)
and one instead sums over triangulations themselves, so that the gluing pattern of simplices
becomes the dynamical variable, again weighted by an Einstein--Hilbert--type action plus matter
contributions. In both cases the geometric degrees of freedom are fundamental and the primary
goal is to describe quantum gravity \cite{Regge1961,HamberReview,AmbjornDurhuusJonssonBook,AmbjornJurkiewiczLoll}.

Structurally, DLR is closest to Regge calculus in that it uses a fixed
abstract complex together with fluctuating geometric data. There are, however, two important
differences. First, the basic variables here are the embedding coordinates $x(n)$, from which
edge lengths and local metrics are derived; the geometry sector is introduced as an auxiliary
regulator field whose role is to restore $\SEd$ symmetry and potentially reduce lattice artefacts,
not to encode an independent gravitational degree of freedom. Second, the construction is
deliberately tied to a hypercubic abstract lattice so as to maintain the locality structure and
transfer matrix decomposition of standard lattice gauge theory, which is crucial for the
reflection positivity argument and for reusing existing simulation technology. One may view
dynamical lattice as a Regge--like dynamics for a fixed hypercubic mesh, with an
action tuned for regularization rather than for quantum gravity.

\paragraph{Smearing and gradient flow.}
Modern lattice gauge theory extensively uses ``smearing'' techniques that replace the link variables by local averages, as in APE, Stout, or HYP smearing, cf. \cite{AlbaneseAPE,MorningstarPeardonStout,HasenfratzKnechtliHYP},
and the gradient flow that defines a diffusion in field space \cite{LuscherFlowJHEP}.
These methods replace the oscillating gauge links with locally averaged variables in the definition of operators or the action. While this shares our goal of mitigating UV noise, the mechanism is fundamentally different: smearing smooths the \emph{fields} over a fixed, rigid geometry, whereas DLR smooths the \emph{geometry} itself. 
By annealing the metric, our approach potentially allows the lattice to accommodate topological objects that might otherwise be frozen by the rigidity of a static grid, a feature that algorithmic smoothing of fields cannot replicate.
Equivalently, DLR can be viewed as an \emph{annealed geometric averaging} (or ``geometric smearing''): instead of diffusing/smoothing the dynamical fields on a fixed hypercubic background, we diffuse the \emph{coordinate/metric degrees of freedom} seen by the fields by integrating over embeddings.
A key difference is that the $\SEd$ symmetry is exact at the level of the regulator and partition function
(modulo the global zero modes).
Smearing/flow improve operators or effective descriptions at finite flow time but do not enlarge the
underlying symmetry group of the static lattice regulator beyond the hypercubic subgroup.
In this sense DLR realises \emph{intrinsic stochastic rotational averaging} inside the path integral, rather than an external rotation/projection applied to measured data.

\paragraph{Improved actions.}
Standard lattice regulators rely on the Symanzik improvement program to systematically remove discretization errors by adding higher-dimension irrelevant operators to the action \cite{SymanzikI,SymanzikII}. 
On a static hypercubic lattice, the breaking of the Euclidean group $\SEd$ down to the hypercubic subgroup $\Hd$ causes single continuum operators to split into multiple lattice operators, 
e.g., the Laplacian square $(\partial^2)^2$ splits into rotationally invariant and non-invariant parts. 
Restoring rotational symmetry to high order therefore requires tuning the coefficients of multiple non-invariant counterterms. 
In contrast, because DLR preserves exact $\SEd$ invariance at finite cutoff, the effective action is constrained to contain only global scalars. This potentially simplifies the improvement program: although the field content is enlarged by the geometry variables, the basis of dangerous irrelevant operators is much smaller. Non-scalar counterterms are forbidden by exact symmetry, meaning one need only tune the coefficients of $\Ed$-invariant operators to achieve higher-order accuracy.

\paragraph{Moving mesh methods.}
From the perspective of numerical analysis, the construction is reminiscent of Arbitrary
Lagrangian--Eulerian (ALE) and other moving mesh formulations, in which one works with a
fixed reference mesh and a time dependent embedding into physical space, subject to
shape regularity constraints \cite{BuddHuangRussellActa,HuangRussellBook}. 
The crucial difference is that in standard ALE methods the mesh
motion is prescribed by a deterministic algorithm, often to equidistribute discretization error,
and the mesh is not part of the dynamical system. 
In DLR the embedding variables $x(n)$ are integrated over quantum--mechanically; one may regard the
scheme as a ``quantum ALE'' version of lattice field theory.
At a continuum level, one may also view the geometry sector as an ``elastic--coordinate'' field theory in which one integrates over
maps $x:M\to\R^d$ on a fixed reference space $M$, with a local energy depending on $\partial x$ and with some nondegeneracy conditions.

These comparisons place DLR as a conservative middle ground
between rigid hypercubic lattices and more radical fluctuating--topology approaches: the
topology and constructive--QFT structure of Wilson's lattice theory are retained, while an
intrinsic, dynamically fluctuating geometry is used to restore global Euclidean symmetry and,
potentially, to reduce lattice artefacts.


\section{Geometry of the dynamical lattice regulator}
\label{sec:geometry-config}

\subsection{Abstract lattice, fields, and embedding}
\label{sec:lattice-embedding}

Let \(d \ge 2\). We take as abstract lattice either

\begin{itemize}
\item the infinite hypercubic lattice \(\Lambda = \Z^d\), or
\item for finite volume simulations, the discrete torus
\begin{equation}
\Lambda = \{0,\dots,L_1-1\}\times\cdots\times\{0,\dots,L_d-1\},
\end{equation}
with periodic identification in each direction.
\end{itemize}

\noindent
In either case we denote by \(\hat\mu\), for \(\mu=1,\dots,d\), the unit vectors in the coordinate
directions and write \(n+\hat\mu\) for the nearest neighbour of \(n\) in direction \(\mu\).

We fix a compact gauge group \(G\). Gauge fields are represented by link variables
\begin{equation}
U_\mu(n) \in G
\end{equation}
attached to the oriented link \((n, n+\hat\mu)\). Matter fields are collected in a variable
\(\Phi(n)\) living in some representation space of \(G\). 
Examples are \(\Phi(n)\in\R\) for a real scalar field, \(\Phi(n)\in\C^N\) for a complex scalar multiplet, and a spinor representation for fermions. 
The collection of all \(U_\mu(n)\) and \(\Phi(n)\) on \(\Lambda\) forms the usual field
content of a lattice gauge theory.

The new ingredient is a geometry field
\begin{equation}
x : \Lambda \to \R^d,
\end{equation}
which assigns to each abstract site \(n\in\Lambda\) a point \(x(n)\in\R^d\) in physical
Euclidean space. In the torus formulation one instead takes \(x(n)\in \T^d\) and works with
the shortest distance representative of coordinate differences; the local constructions in
what follows are unchanged.

It is often convenient to regard the geometry field $x$ as part of an
\emph{extended} field multiplet, alongside the gauge and matter fields:
\begin{equation}
  (x,U,\Phi) \in (\R^d)^\Lambda \times G^{\text{links}}\times
                 \cF_{\text{matter}}.
\end{equation}
In this viewpoint DLR is simply a local lattice
field theory with additional scalar degrees of freedom $x_\mu(n)$ and a
special pattern of couplings.  This extended-field formulation will be
crucial in Section~\ref{sec:OS-universality}, where reflection positivity is
established by letting the Osterwalder--Schrader reflection act on the
indices $n$ and treating $x$ on the same footing as the other fields.

\subsection{Local Euclidean geometry and admissible embeddings}
\label{sec:local-geometry}

We spell out the local Euclidean geometry induced by an embedding
$x : \Lambda \to \R^d$ of the abstract lattice and formulate our
admissibility conditions.  Throughout this subsection we work on
$\R^d$; the periodic (torus) version will be discussed in
Section~\ref{sec:torus-version}.

Given $x : \Lambda \to \R^d$ and a site
$n \in \Lambda$, we define the (forward and backward) \emph{edge vectors}
\begin{equation}\label{e:for-back-edge}
e^\pm_\mu(n) = x(n\pm\hat\mu) - x(n) \in \R^d,
\qquad \mu = 1,\dots,d,
\end{equation}
where $\hat\mu$ denotes the unit vector in the $\mu$--th lattice direction. 
We may write the forward vectors simply as $e_\mu(n)=e^+_\mu(n)$.
We then define the (advanced and retarded) \emph{frame matrices}
\begin{equation}\label{e:advanced-frame}
E(n) = \big[ e_1(n)\,e_2(n)\,\ldots\,e_d(n) \big] \in \R^{d\times d} ,
\end{equation}
and
\begin{equation}\label{e:retarded-frame}
E^\theta\!(n) = \big[ e_1^-(n)\,e_2(n)\,\ldots\,e_d(n) \big] \in \R^{d\times d} .
\end{equation}
The associated (advanced) \emph{Gram matrix} (or local metric) and \emph{local volume} are then
\begin{equation}\label{e:advanced-metric}
g(n) = E(n)^{\mathsf T} E(n), 
\qquad 
V(n) = \bigl|\det E(n)\bigr| = \sqrt{\det g(n)}.
\end{equation}
In $d=2$, $V(n)$ is the (unsigned) area of the outgoing edge parallelogram at $n$; for a non--affine embedded plaquette it need not coincide with the Euclidean area of the quadrilateral cell.
The components of $g(n)$ are $g_{\mu\nu}=e_\mu(n)\cdot e_\nu(n)$.
Whenever $g(n)$ is positive definite we denote the components of its inverse by
$g^{\mu\nu}(n)$.

We restrict attention to embeddings for which each pair $E(n)$ and $E^\theta\!(n)$ is uniformly nondegenerate
and uniformly shape regular, in a way that is stable under global proper Euclidean motions.  A convenient
local formulation is the following.

\begin{definition}\label{def:admissible-geometry}
Fix dimensionless constants $C_{\ell}\ge1$ and $c_{V}\le1$.
Then for a parameter $a>0$,
an embedding $x : \Lambda \to \R^d$ is called \emph{$a$-admissible} if for every $n\in\Lambda$ we have
\begin{equation}\label{e:local-adm}
|e_\mu(n)| \le C_{\ell}\,a\quad(\mu=1,\dots,d),
\qquad
\min\{\det E(n),-\det E^\theta\!(n)\} \ge c_{V}\,a^{d} .
\end{equation}
We denote by $\XADM(a)\subset(\R^{d})^{\Lambda}$ the set of all $a$--admissible embeddings.
\end{definition}

\noindent
Conditions \eqref{e:local-adm} are local (they can be checked independently at each site) and
invariant under global translations and proper rotations $x\mapsto Rx+b$ with $R\in\mathrm{SO}(d)$ and
$b\in\R^d$.  Note, however, that the lower bound on the oriented volumes $\det E(n)$ and $-\det E^\theta\!(n)$ selects an
orientation sector: under an improper rotation $R\in\mathrm{O}(d)\setminus\mathrm{SO}(d)$ one has
\begin{equation}
\det(RE(n))=\det(R)\det E(n)=-\det E(n), 
\end{equation}
and the same behavior for $\det E^\theta\!(n)$,
so $\XADM$ is not
invariant under the full Euclidean group $\Ed=\mathrm{O}(d)\ltimes\R^d$ but only under the special Euclidean subgroup
$\SEd=\mathrm{SO}(d)\ltimes\R^d$. This point is discussed further in \S\ref{sec:Ed-invariance}.

\begin{remark}\label{rem:dimensionless-x}
For $x\in\XADM(a)$ introduce the dimensionless embedding
$y(n)=a^{-1}x(n)$.  Then in terms of $y$, the admissibility conditions
\eqref{e:local-adm} become $a$--independent bounds, meaning that $\XADM(a)$ is a simple rescaling of $\XADM(1)$.
\end{remark}

The next lemma records the basic geometric consequences of admissibility.

\begin{lemma}\label{l:E-shape}
Let $x\in\XADM(a)$ be an admissible embedding in the sense of
Definition~\ref{def:admissible-geometry}.  Then there exist a constant $C^*$, 
depending only on $C_\ell$, $c_V$, and the dimension $d$, such that
\begin{equation}
  \|E(n)^{-1}\|+\|E^\theta(n)^{-1}\|\le C^*\,a^{-1}\qquad\text{for all }n.
\end{equation}
\end{lemma}

\begin{proof}
Fix $n$. Write the rescaled frames $\widetilde E=a^{-1}E(n)$ and $\widetilde E^\theta=a^{-1}E^\theta(n)$.
Then the defining bounds for admissibility (edge lengths $\sim a$ and oriented volumes $\sim a^d$) become
uniform bounds for $\widetilde E$ and $\widetilde E^\theta$ at scale $a=1$.
The set of all such $\widetilde E$ is closed and bounded in the space of invertible matrices with
$\det \widetilde E$ bounded away from $0$, hence compact. The continuous map $A\mapsto \|A^{-1}\|$
attains a finite maximum on this set,
say $\|\widetilde E^{-1}\|\le C$ and $\|(\widetilde E^\theta)^{-1}\|\le C$ with $C$ depending only on
$(C_\ell,c_V,d)$.
Finally, we note
\begin{equation}
E(n)^{-1}=(a\widetilde E)^{-1}=a^{-1}\widetilde E^{-1},\qquad (E^\theta(n))^{-1}=a^{-1}(\widetilde E^\theta)^{-1},
\end{equation}
establishing the lemma with $C^*=2C$.
\end{proof}

The content of Lemma~\ref{l:E-shape} is that the eigenvalues of the local metrics
are uniformly bounded above and below away from $0$. In numerical--analysis
language this is a standard ``shape--regularity'' condition: 
the local affine cells determined by the corner frames cannot be arbitrarily skinny or flat.
In particular, \eqref{e:local-adm} excludes local collapses of volume and guarantees that the embedding
$x:\Lambda\to\R^d$ is everywhere locally an immersion with controlled condition
number. Global properties such as non--self--intersection will be enforced by
restricting $x$ to a suitable admissible component in Definition~\ref{def:principal-component} and, if
desired, by adding further constraints as discussed in the next remark.

\begin{remark}\label{r:adm}
The formulation in Definition~\ref{def:admissible-geometry} is deliberately minimal: it is
stable under global Euclidean transformations and sufficient for the geometric
estimates in Lemma~\ref{l:E-shape}. In concrete implementations it is often convenient to
replace or strengthen \eqref{e:local-adm} by more classical mesh conditions. We briefly
record a few equivalent or stronger variants.

\begin{itemize}
  \item[(A1)] \emph{Edge length lower bound.} For each $\mu=1,\dots,d$, we require
  \begin{equation}\label{eq:A1}
    |e_\mu(n)| \ge c_\ell\,a .
  \end{equation}
  This prevents local ``collapses'' of edges.

  \item[(A2)] \emph{Angle bounds.} For all $\mu\neq\nu$ we require
  \begin{equation}\label{eq:A2}
    \frac{|e_\mu(n)\cdot e_\nu(n)|}{| e_\mu(n) | | e_\nu(n) |}
    \;\le\; \cos\varepsilon_0.
  \end{equation}
  That is, the angle between any pair of edges is bounded away from $0$ and $\pi$.

\item[(A3)] \emph{All-corner volume bounds.} 
One may strengthen the determinant condition in \eqref{e:local-adm}
by requiring nondegeneracy of \emph{all} $2^d$ corner frames incident to $n$.
Concretely, for each
site $n$ and each choice of forward/backward directions $\sigma\in\{\pm1\}^d$, let $E^\sigma(n)$
be the $d\times d$ ``corner frame'' obtained by taking, in each coordinate direction $\mu$,
the edge from $n$ to $n+\sigma_\mu \hat\mu$.  Then impose the uniform oriented-volume bound
  \begin{equation}
\min_{\sigma\in\{\pm1\}^d}\ \Big(\prod_{\mu=1}^d \sigma_\mu\Big)\,\det E^\sigma(n)\ \ge\ c_V a^d
\qquad\text{for all }n\in\Lambda .
  \end{equation}
This rules out fold-overs at any corner and is convenient when actions involve mixed
forward/backward stencils in several directions.

  \item[(A4)] \emph{Cell convexity.}
  In addition to local nondegeneracy of the edge frame, one may require each
  elementary cell to be \emph{geometrically convex} in $\R^d$.  
  For instance, in $d=2$, the elementary cell $Q(n)$ at $n$ is the quadrilateral with vertices
  \begin{equation}
    x(n),\quad x(n+e_1),\quad x(n+e_1+e_2),\quad x(n+e_2) .
  \end{equation}

  \item[(A5)] \emph{Cell shape regularity.}
If desired, one may impose a standard finite--element shape--regularity assumption on the
embedded cells: taking the $2d$ case for example, each quadrilateral image is required to be convex and to have uniformly
bounded chunkiness parameter $h_Q/\rho_Q$ (diameter over inradius), uniformly in the lattice
and in the configuration.
Equivalently, one may enforce a uniform lower bound on all interior angles together with a
uniform bound on edge--length ratios; any of the standard equivalent criteria may be used.

  \item[(A6)] \emph{Global embedding / non--self--intersection.}
  The conditions above control the local geometry but do not, by themselves,
  preclude global self--intersections of the mesh. In simulations it is often
  convenient to restrict to embeddings that are globally injective and avoid
  edge crossings, for example by imposing a uniform separation condition
  \begin{equation}\label{eq:A4}
    | x(n)-x(m) | \ge \delta_{\mathrm{sep}} a
    \qquad\text{for all $n\neq m$} .
  \end{equation}
This ensures that the dynamical mesh remains a controlled distortion of the
  reference lattice without changes in combinatorics or global foldings.
  
  In practice, conditions such as cell convexity and uniform cell shape
regularity are effective guards against self--intersection: a mesh made of
embedded convex cells cannot ``fold'' locally, and when combined with a mild
global separation condition that is occasionally enforced, 
they prevent distant parts of the mesh from passing through one
another under local Monte Carlo updates. Thus convexity constraints serve as a
computationally cheap proxy for global non--self--intersection checks in many
settings.
\end{itemize}
\end{remark}

The set $\XADM(a)$ of admissible embeddings typically has several
connected components, corresponding for instance to different global
windings on a torus.  For the purposes of this paper we focus on the
component that contains the regular hypercubic embedding.

\begin{definition}\label{def:principal-component}
Let $x_{0}(n)=a\,n$ be the regular embedding at scale $a>0$.
The \emph{principal admissible component} $\Xadm(a)$ is the connected component of $\XADM(a)$
containing $x_{0}$, i.e.,
\begin{equation*}
\Xadm(a)=\{x\in\XADM(a):\; x \text{ can be connected to } x_{0}\text{ by a continuous path inside }\XADM(a)\}.
\end{equation*}
If the scale $a$ is to be irrelevant or inferred from the context, which will be the case for the most of this paper,
we may simply write $\Xadm$ for $\Xadm(a)$.
\end{definition}

Working in $\Xadm$ ensures that the dynamical mesh remains a
controlled distortion of the regular hypercubic grid: shape regularity
is enforced uniformly by Lemma~\ref{l:E-shape}, and no
topological defects (such as changes in the combinatorial structure or
global self-intersections) are introduced by the fluctuations of $x$.
This sector is the natural setting for both the perturbative Symanzik
analysis and the nonperturbative reflection positivity arguments that
follow.

\subsection{Geometry measure and dynamics}
\label{sec:geometry-measure}

The geometry field $x:\Lambda\to\R^d$ is treated as a genuine dynamical
degree of freedom of the regulator.  In the Euclidean functional integral
we integrate $x$ over the principal admissible component
$\Xadm\subset(\R^d)^\Lambda$ introduced in
Definition~\ref{def:principal-component}.  The corresponding geometry
measure is of the form
\begin{equation}
  d\mu_x(x)
  \;=\;
  \frac{1}{Z_x}\,
  \exp\bigl(-S_x[x]\bigr)\,Dx,
  \qquad
  Dx = \prod_{n\in\Lambda} d^d x(n),
  \label{eq:geom-measure}
\end{equation}
where $S_x[x]$ is local and depends on $x$ only through Euclidean invariants
of the local edge data (hence is $\SEd$--covariant on $\Xadm$), and
$Z_x$ is the corresponding normalisation factor.  The hard admissibility
constraint is implemented by restricting the domain of integration to
$\Xadm$; we do \emph{not} insert a separate indicator function into
the integrand.

Throughout we take the geometry integration measure to be the flat product Lebesgue measure
$Dx=\prod_{n\in\Lambda} d^d x(n)$ restricted to $\Xadm$, optionally reweighted by $\exp(-S_x[x])$.
This is the natural choice in the embedding-field formulation: $x(n)$ are fundamental variables and the theory is not intended to be reparametrization-invariant (unlike lattice quantum gravity).  Alternative ``geometric'' measures (e.g.\ inserting local factors built from $\sqrt{\det g(n)}$) would amount to a specific local reweighting of configurations and can be viewed as part of the choice of $S_x$ or local counterterms.
We keep the flat measure to preserve the manifest global shift symmetry $x\mapsto x+b$ and to keep the regulator definition transparent.

\medskip

\noindent
\textbf{Minimal choice.}
The conceptually simplest choice is the ``minimal'' geometry action
\begin{equation}
  S_x^{\min}[x] \equiv 0,
\end{equation}
so that $d\mu_x$ is just the restriction of the flat product measure $Dx$
to $\Xadm$.  The admissibility conditions from Definition~\ref{def:admissible-geometry} and Remark~\ref{r:adm} then appear as
hard walls in configuration space: any proposed configuration outside
$\Xadm$ is assigned infinite action and zero measure.  Even in this
minimal setting the geometry does not behave as a collection of
non-interacting free variables.  Because $\Xadm$ is a nontrivial subset of $(\R^d)^\Lambda$, the restriction $x\in\Xadm$
induces effective interactions between neighbouring sites via purely
entropic effects: the volume in configuration space available to a site
$x(n)$ depends on the positions of its neighbours through the local admissibility constraints.  

In $d=2$ our exploratory simulations suggest that the minimal choice
$S_x^{\min}\equiv 0$, with the integration restricted to the admissible set
$\Xadm$, can already yield a stable, nondegenerate geometry ensemble.  This
should be viewed as a proof of concept for the local admissibility framework,
not as evidence that the minimal choice is sufficient in higher dimensions.
In higher dimensions, especially in $d=4$, entropic effects alone may admit
phases with long--range geometric correlations or crumpling/roughening
phenomena in which the lattice spacing ceases to define a clean regulator scale.
For four--dimensional applications we therefore expect that additional local
stiffness terms, such as the examples discussed below, together with a systematic
scan of the geometry phase diagram may be necessary to find a robust,
computationally affordable short-range geometry regime~\hyperref[hyp:SR]{(SR)}.

\medskip
\noindent
\textbf{Local $\Ed$--invariant penalty terms.}
For practical simulations it is often advantageous to supplement the hard
admissibility constraints by soft, local penalties which keep the mesh
away from the boundary of $\Xadm$ and disfavour excessively
distorted cells.  A general geometry action of this form can be written
as a sum of local terms
\begin{equation}\label{eq:geom-action-local}
  S_x[x]
  \;=\;
  \sum_{n\in\Lambda} s_x\bigl(g(n),g(n)^{-1},V(n)\bigr),  
\end{equation}
where $s_x$ is a scalar function of the local metric and volume at site
$n$, and $g(n),V(n)$ are the Gram matrix and volume from
Section~\ref{sec:local-geometry}.  Because $g(n)$ and $V(n)$ are built
from the edge vectors $e_\mu(n)$ via Euclidean inner products, any such
$s_x$ is automatically invariant under global Euclidean transformations
$x\mapsto Rx+b$; see Section~\ref{sec:Ed-invariance}.

A particularly transparent class of such penalties are ``spring'' terms that keep the local
metric $g(n)$ close to a prescribed target.
For instance, fixing a target length scale $a>0$, one may penalise deviations from the
Euclidean reference metric $a^2 I$ by
\begin{equation}\label{eq:spring-penalty}
S_{\mathrm{spr}}[x]
  \;=\;
  \alpha \sum_{n\in\Lambda}\sum_{\mu,\nu}
  \bigl(e_\mu(n)\cdot e_\nu(n)-a^2\delta_{\mu\nu}\bigr)^2 .
\end{equation}
The diagonal terms $\mu=\nu$ act as \emph{edge springs}, keeping $|e_\mu(n)|$ near $a$,
while the off--diagonal terms $\mu\neq\nu$ act as \emph{angle springs}, favouring near--orthogonality
of neighbouring edges.
This is only an illustrative example: by replacing $a^2I$ with an arbitrary positive definite
reference matrix, or by applying the same ``soft constraint'' philosophy to other local invariants,
such as $V(n)$ and the eigenvalues of $g(n)$, 
one obtains many further $\Ed$--invariant local choices, all encompassed
by the general form~\eqref{eq:geom-action-local}.

In practice the coupling $\alpha$ is taken small enough that the geometry remains a mild fluctuation
around the regular lattice, but large enough to discourage excursions toward highly distorted cells
or toward the boundary of $\Xadm$.

Penalties of the form \eqref{eq:spring-penalty} are designed to suppress long--wavelength distortions of the local frame and thus to promote
a massive (finite--correlation--length) geometry sector once the global $\SEd$ zero modes are fixed/quotiented.
In the language of Section~\ref{subsec:symanzik}, such choices are intended to place the geometry in a regime compatible with the
short-range geometry hypothesis \hyperref[hyp:SR]{(SR)}.

\medskip

\noindent
\textbf{Universality and the role of $S_x$.}
From the viewpoint of effective field theory, different choices of the
geometry action $S_x$ define different microscopic regulators, but as
long as $S_x$ is local, $\Ed$--invariant, and compatible with the
admissibility conditions, they are expected to lie in the same
universality class.
One may also view the choice of $S_x$ (more precisely, the local weight
$e^{-S_x[x]}$ together with the admissible--set restriction) as selecting an
effective \emph{submanifold measure} on the space of embedded lattices inside
the continuum configuration space: it specifies how the enlarged, finite-dimensional
family of fields induced by the embedding $x$ is sampled.
This reweighting is not required by an exact local redundancy, but changing it
(within the same locality and symmetry class) is expected to affect the continuum
limit only through symmetry--allowed local counterterms.
Integrating out the geometry field induces local,
higher--dimensional operators in the effective action for the gauge and
matter fields; these operators are constrained by gauge invariance,
Euclidean invariance, and reflection positivity, and are therefore
irrelevant in the renormalisation--group sense in the continuum limit.
We will see an explicit perturbative example of this mechanism in the
Symanzik expansion for a scalar theory in \S\ref{subsec:phi4-one-loop}.

\medskip

\noindent
\textbf{Global zero modes versus local twisting.}
The embedding variables $x(n)$ enjoy an exact global $\SEd$ symmetry,
$x(n)\mapsto R x(n)+b$, acting uniformly on all sites.
In infinite volume this produces non--normalisable zero modes, and in finite volume an irrelevant overall group--volume factor.
In practice one therefore fixes or quotients these global modes by a convenient convention, e.g., pinning the centre of mass and fixing
the orientation of a chosen local frame.
This step does \emph{not} affect local observables and is entirely analogous to removing rigid--motion zero modes in molecular simulations.

The practically relevant question is not whether the regulator is \emph{formally} $\SEd$--covariant (it is, by construction), but how this covariance is \emph{realised} in the ensemble in a way compatible with universality.
Our Symanzik discussion is conditional on the short--range geometry hypothesis~\hyperref[hyp:SR]{(SR)}; under this assumption, integrating out $x$ produces only local, irrelevant corrections in the effective action.
Within this regime, rigid global motions cannot implement $\SEd$ covariance at the level of local probes (they carry no local geometric information), and any attempt to enforce isotropy through genuinely long--range geometric modes would violate (SR) and could in principle change the universality class.
Thus, in the intended operating regime, $\SEd$ covariance that is \emph{felt locally} must come from the \emph{local} embedding fluctuations---twisting, shearing, and local reorientation from cell to cell---which average short distance observables over many local frames while keeping the underlying hypercubic connectivity fixed.

\medskip
\noindent
\textbf{Monte Carlo dynamics.}
In numerical simulations the geometry field $x$ is updated by local
Metropolis moves of the form $x(n)\to x(n)+\delta x$ with $\delta x$
drawn from a symmetric proposal distribution, followed by an
admissibility check at the sites affected by the move and an acceptance
test based on the local change in $S_x[x]$ and in the field action.
Because $S_x$ is local, the acceptance probability can be evaluated from
a uniformly bounded neighbourhood of $n$, and the resulting Markov chain
is manifestly local on the abstract lattice.  A detailed description of
the update scheme and its ergodicity properties is deferred to the numerics section \S\ref{sec:numerics}.

From an algorithmic standpoint, the dynamical geometry sector adds only a mild overhead compared
to standard lattice gauge theory.  A local proposal $x(n)\mapsto x(n)+\delta x$ affects only the
finite set of terms in $S_x$ and in $S_{\mathrm{fields}}$ that involve $n$ and its nearest
neighbours, so both the admissibility check and the Metropolis acceptance ratio can be evaluated
in $O(1)$ work per proposal, with a constant depending only on the dimension and the chosen
locality range, not on the lattice volume.  A full ``sweep'' of geometry updates therefore costs
$O(|\Lambda|)$, just as in standard LGT.  In practice the geometry update can be interleaved with
the usual gauge and matter updates (or combined in a hybrid scheme), and the method remains
compatible with existing locality-based implementations; the principal new ingredient is the
additional bookkeeping of the local geometric data needed to evaluate the affected action terms.

\subsection{Torus version and periodic implementation}
\label{sec:torus-version}

For most analytical questions,
it is convenient to regard the dynamical lattice as embedded into the full Euclidean space $\R^d$.
In numerical simulations, however, one typically works at finite volume with periodic boundary
conditions, i.e., on a $d$--dimensional discrete torus.

There are two equivalent ways to formulate this periodic setting:

\begin{enumerate}[label=(\roman*)]
\item \emph{Intrinsic formulation on the discrete torus.}
Fix an integer $N\ge 1$ and set
\begin{equation}
  \Lambda = (\Z/N\Z)^d ,
\end{equation}
viewed as a periodic hypercubic lattice.  The geometry field is a map
\begin{equation}
  x:\Lambda\to \T^d = (\R/L\Z)^d
\end{equation}
for some physical period $L>0$.  Differences $x(n+\hat\mu)-x(n)$ are interpreted modulo $L\Z^d$,
and the edge vectors $e_\mu(n)$ are chosen as the unique representatives of these differences in a
fixed fundamental domain of $\T^d$ (for instance $(-L/2,L/2]^d$).  Local geometric quantities are
then defined from these edge vectors exactly as in the $\R^d$ setting.

\item \emph{Periodic formulation on $\Z^d$.}
Alternatively, one may take the index set to be the infinite lattice
\begin{equation}
  \Lambda = \Z^d ,
\end{equation}
and work with an embedding $x:\Lambda\to\R^d$ satisfying the quasi--periodicity condition
\begin{equation}\label{eq:quasi-periodic}
  x(n+N\hat\mu)=x(n)+L\,\hat e_\mu , \qquad \mu=1,\dots,d .
\end{equation}
All geometric quantities are computed from $x$ exactly as before, while observables and the action
are invariant under global shifts $x\mapsto x+Lk$ with $k\in\Z^d$.  Passing to the quotient by the
relations~\eqref{eq:quasi-periodic} recovers the intrinsic discrete--torus description in~(i).
\end{enumerate}

In both formulations the local admissibility conditions are imposed on the edge vectors $e_\mu(n)$,
and the principal admissible component $\Xadm$ is understood modulo periodicity.  Global
non--self--intersection is likewise interpreted in the periodic sense: cells may wrap around the
torus, but configurations in which large portions of the mesh fold back and concentrate in a small
region are excluded by the same local shape regularity and volume constraints as in the $\R^d$
setting.

The passage from $\R^d$ to $\T^d$ does not affect any of the local
arguments in this paper. In particular, the construction of the
geometry measure $d\mu_x$ and of the field actions remains unchanged,
and the analysis of Euclidean invariance and reflection positivity is
carried out on $\R^d$ and then transferred to the periodic case in the
standard way.  For the Monte Carlo implementation we restrict to a
finite periodic lattice and update the geometry field subject to these
periodic admissibility constraints; the details of the algorithm are
deferred to the numerical section.

\section{Field actions and symmetries on a dynamical lattice}
\label{sec:fields-symmetries}

\subsection{Configuration space and full action}
\label{sec:full-action}

We summarise the full configuration space and the action of the
dynamical lattice regulator.  The basic kinematical variables are:
\begin{itemize}
\item a geometry field $x : \Lambda \to \R^d$, constrained to take
  values in the principal admissible component
  $\Xadm\subset(\R^d)^\Lambda$ introduced in
  Section~\ref{sec:local-geometry};
\item gauge link variables $U_\mu(n)\in G$ attached to each oriented
  edge $(n,\mu)$ of the abstract hypercubic lattice $\Lambda$;
\item matter fields $\Phi(n)$ (scalars, spinors, etc.) living on the
  sites $n\in\Lambda$.
\end{itemize}

\noindent
We denote the full configuration collectively by $(x,U,\Phi)$.
The geometry degrees of freedom $x$ are governed by a local action $S_x[x]$
built from Euclidean invariants of the embedding data as in \S\ref{sec:geometry-measure}, while the gauge and matter fields
are coupled to $x$ through a local field action $S_{\mathrm{fields}}[x,U,\Phi]$.
The total action is
\begin{equation}
  S[x,U,\Phi]
  \;=\;
  S_x[x] + S_{\mathrm{fields}}[x,U,\Phi].
  \label{eq:full-action}
\end{equation}

\noindent
For complete precision one first defines the functional integral in finite volume,
for instance on a time-open cylinder $\Lambda_{T,L}\subset\Z^d$, with the first
lattice direction singled out as index time, where all measures are ordinary finite
products.  The infinite--lattice notation used below is then understood as the limit
$\Lambda_{T,L}\nearrow\Z^d$ on local observables.  Global translation/rotation zero
modes of $x$ may be fixed (or quotiented) in finite volume; this does not affect the
local structure of the actions below nor the reflection positivity statements proved
later.

Thus the Euclidean functional integral is defined formally by
\begin{equation}
  Z
  \;=\;
  \int_{\Xadm} Dx
  \int DU \int D\Phi\;
  \exp\bigl(-S_x[x] - S_{\mathrm{fields}}[x,U,\Phi]\bigr),
  \label{eq:partition-function}
\end{equation}
where $Dx$ is the flat measure in $\Xadm$, $DU$ is the
product Haar measure over link variables, and $D\Phi$ is the standard
product measure over matter fields,  which is the Lebesgue measure for bosons and the usual Grassmann
measure for fermions.

We present a representative class of gauge and matter actions on a dynamical mesh.
In all cases the fields live on the abstract lattice $\Lambda$, and the interactions
have the usual nearest neighbour stencil, with couplings dressed by local
geometry dependent weights.  For later Osterwalder--Schrader (OS) reflection-positivity
arguments, one subtlety is important: the local metric and volume constructed from a
\emph{forward} corner frame are intrinsically biased with respect to the time direction.  
Rather than modifying the finite-difference stencils, we remove this bias by \emph{action averaging}: we define two
local geometric dressings, exchanged by the OS reflection, and average the resulting
actions.  This produces a manifestly reflection invariant field action while preserving
locality and gauge invariance.

\medskip
\noindent
\textbf{Signed directions and the induced reflection on directions.}
Let
\begin{equation}
\Delta = \{\pm 1,\dots,\pm d\} ,
\end{equation}
denote the set of signed coordinate directions, and write $\hat\mu\in\Z^d$ for the corresponding
unit vectors, so that $\widehat{-\mu}=-\hat\mu$.
We single out $\mu=1$ as the index-time direction,
and define the time-reflection $\theta$ on directions by
\begin{equation}\label{eq:theta-on-directions}
  \theta(\pm 1)=\mp 1,\qquad \theta(\pm i)=\pm i\quad (i=2,\dots,d),
  \qquad\text{equivalently}\qquad
  \theta(\hat\mu)=\widehat{\theta(\mu)}.
\end{equation}
For $\mu\in\Delta$ define the oriented edge vectors and gauge links by
\begin{equation}\label{eq:signed-edges-links}
  e_\mu(n)=x(n+\hat\mu)-x(n),\qquad
  U_{-\mu}(n)=U_\mu(n-\hat\mu)^{-1}.
\end{equation}
Thus $U_\mu(n)$ always lives on the oriented link $n\to n+\hat\mu$, for both positive and negative
directions.
With this convention, we can write \eqref{e:for-back-edge} now as $e^\pm_\mu(n)=e_{\pm\mu}(n)$.

\medskip
\noindent
\textbf{Prototype actions and $\theta$--averaging.}
Recall from \eqref{e:advanced-frame}-\eqref{e:advanced-metric} the advanced and retarded frames
$E(n)$ and $E^\theta\!(n)$, as well as the local metric $g(n)$ and volume $V(n)$ computed from $E(n)$.
We first define ``prototype'' field actions
$S^\circ[x,U,\Phi]$ by coupling the usual nearest-neighbour stencils to $g$ and $V$.  
To remove the forward-time bias, we then define the $\theta$-partner action
$S^\theta[x,U,\Phi]$ by the same formulas but with every occurrence of a direction index $\mu$ replaced by
$\theta(\mu)$ and with the corner frame replaced by its $\theta$--version.  
Averaging the two yields the field action
\begin{equation}\label{eq:fields-action-average-theta}
S_{\mathrm{fields}}[x,U,\Phi]=\tfrac12\bigl(S^\circ[x,U,\Phi]+S^\theta[x,U,\Phi]\bigr).
\end{equation}
Notice that since $g$ and $V$ depend on $E$, to get $S^\theta$, 
any occurrences of $g$ and $V$ in $S^\circ$ should be replaced by $g^\theta$ and $V^\theta$, 
that are computed from $E^\theta$ instead of $E$.
By construction this preserves locality and gauge invariance, and it is compatible with the OS
reflection $\Theta$ introduced later.

More generally, one may analogously symmetrize with respect to reflections in any coordinate direction (or build the action directly from such symmetrized geometric data), which reduces directional bias and makes the couplings more isotropic with respect to the lattice axes.

\medskip
\noindent
\textbf{Scalar fields.}
Let \(\phi:\Lambda\to\C^m\) be a complex scalar multiplet transforming in some unitary representation
of \(G\).  For $\mu\in\Delta$ define the gauge-covariant signed forward difference
\begin{equation}\label{eq:scalar-signed-diff}
  \nabla_\mu \phi(n)
  =
  \frac{U_\mu(n)\,\phi(n+\hat\mu) - \phi(n)}{|e_\mu(n)|}.
\end{equation}
The prototype scalar kinetic and potential contributions are
\begin{align}
  S_{\phi,\mathrm{kin}}^\circ[x,\phi,U]
  &=
  \frac{1}{2}
  \sum_{n\in\Lambda}
  V(n)
  \sum_{\mu,\nu=1}^d
  \bigl(g(n)^{-1}\bigr)^{\mu\nu}\,
  \langle \nabla_\mu\phi(n),\nabla_\nu\phi(n)\rangle ,
  \label{eq:scalar-kinetic-theta}\\
  S_{\phi,\mathrm{pot}}^\circ[x,\phi]
  &=
  \sum_{n\in\Lambda}
  V(n)\,
  W_{\mathrm{loc}}\bigl(\phi(n)\bigr),
  \label{eq:scalar-potential-theta}
\end{align}
where $\langle a,b\rangle=a^\dagger b$ and
$W_{\mathrm{loc}}$ is a local gauge-invariant potential such as
\begin{equation}
W_{\mathrm{loc}}(\phi)=\frac{m_0^2}2\phi^\dagger\phi+\frac{\lambda_0}{4!}(\phi^\dagger\phi)^2.
\end{equation}
We define the prototype scalar action $S_{\phi}^\circ=S_{\phi,\mathrm{kin}}^\circ+S_{\phi,\mathrm{pot}}^\circ$
and then set
\begin{equation}\label{eq:scalar-total-avg-theta}
  S_\phi[x,\phi,U]
  =
  \tfrac12\Bigl(S_{\phi}^\circ[x,\phi,U] + S_{\phi}^{\theta}[x,\phi,U]\Bigr) .
\end{equation}
To reiterate the definition of $S_{\phi}^{\theta}$, it is defined by replacing
$(g,V)$ with $(g^\theta,V^\theta)$ and applying $\mu\mapsto\theta(\mu)$ to the direction labels
appearing in \eqref{eq:scalar-kinetic-theta}, so that only the time-direction differences are
reversed.
On the regular embedding $x(n)=a n$ one has $E(n)=aI$ and $E^\theta(n)=aR_\theta$ with
$R_\theta=\mathrm{diag}(-1,1,\dots,1)$. In particular
$g(n)=E(n)^{\mathsf{T}}E(n)=a^2 I$ and $g^\theta(n)=(E^\theta(n))^{\mathsf{T}}E^\theta(n)=a^2 I$.
Hence the $\theta$-averaged scalar action reduces to the standard hypercubic action.

One may alternatively use a centered covariant difference or higher-order stencils to reduce
cutoff effects.  For example, a nearest-neighbour centered difference can be defined by
\begin{equation}\label{eq:scalar-symmetric-diff}
  \nabla^{\mathrm{sym}}_\mu \phi(n)
  =
  \frac{U_\mu(n)\phi(n+\hat\mu) - U_{-\mu}(n)\phi(n-\hat\mu)}{|e_\mu(n)|+|e_{-\mu}(n)|} ,
\end{equation}
though we will not pursue such refinements in the present work.

\medskip
\noindent
\textbf{Gauge fields.}
For the gauge sector we use a Wilson--type plaquette action coupled to the embedding through local
positive weights.  For $\mu,\nu\in\Delta$ define the oriented plaquette holonomy by the uniform
formula
\begin{equation}\label{eq:signed-plaquette}
  U_{\mu\nu}(n)
  =
  U_\mu(n)\,U_\nu(n+\hat\mu)\,U_\mu(n+\hat\nu)^{-1}\,U_\nu(n)^{-1}.
\end{equation}
For $\mu,\nu>0$ this agrees with the usual plaquette; for $\mu<0$ it represents the oppositely
oriented plaquette based at $n$, encoded using \eqref{eq:signed-edges-links}.
The prototype gauge action is
\begin{equation}\label{eq:gauge-action-proto-theta}
  S_{g}^\circ[x,U]
  =
  \frac{\beta}{N}\sum_{n\in\Lambda}\sum_{1\le \mu<\nu\le d}
  w^{\mu\nu}(n;x)\,
  \Re\tr\!\bigl(I-U_{\mu\nu}(n)\bigr),
\end{equation}
where the weight $w^{\mu\nu}(n;x)>0$ depends on the embedding only through local Euclidean invariants
built from $(g(n),V(n))$, or equivalently from $E(n)$.  Two natural choices are
\begin{equation}\label{eq:gauge-weight-B}
  w^{\mu\nu}(n;x)
  =
  V(n)\,
  \Bigl(\bigl(g(n)^{-1}\bigr)^{\mu\mu}\bigl(g(n)^{-1}\bigr)^{\nu\nu}
        -\big[\bigl(g(n)^{-1}\bigr)^{\mu\nu}\big]^{2}\Bigr),
\end{equation}
and
\begin{equation}\label{eq:gauge-weight-A}
  w^{\mu\nu}(n;x) = \frac{V(n)}{A_{\mu\nu}(n)^2},
  \qquad
  A_{\mu\nu}(n) = \bigl|e_{\mu}(n)\wedge e_{\nu}(n)\bigr|,
\end{equation}
where in \eqref{eq:gauge-weight-A} the signed edges $e_\mu(n)$ are given by
\eqref{eq:signed-edges-links}.  The coefficients in \eqref{eq:gauge-weight-B} are precisely those
arising from the contraction $\sqrt{g}\,F_{\mu\nu}F^{\mu\nu}$ with the inverse metric in coordinate
components.
On a regular embedding $x(n)=a\,n$, both weights reduce to a constant $w_{\mu\nu}\equiv a^{d-4}$, so
\eqref{eq:gauge-action-proto-theta} reduces (after absorbing this constant into $\beta$) to the
standard Wilson plaquette action.

We then define the averaged gauge action by
\begin{equation}\label{eq:gauge-action-avg-theta}
  S_g[x,U]
  =
  \tfrac12\Bigl(S_{g}^\circ[x,U]+S_{g}^{\theta}[x,U]\Bigr),
\end{equation}
Concretely, applying
$\mu\mapsto\theta(\mu)$ inside the plaquette \eqref{eq:signed-plaquette} reverses precisely the
time-oriented plaquettes, so that the $\theta$--partner couples the $E^\theta$--based weights to the
corresponding backward-time plaquettes at the same basepoint $n$.

\medskip
\noindent\textbf{Fermions.}
Fermionic actions on a dynamical lattice can be constructed in close analogy with the standard
hypercubic case, with the difference that the discrete Dirac operator must be built from the local
geometry induced by the embedding \(x\).
Let \(\psi(n)\) and \(\bar\psi(n)\) be lattice Dirac spinors (Grassmann fields) in a representation of \(G\).
Fix Euclidean gamma matrices \(\{\gamma_a\}_{a=1}^d\) with \(\{\gamma_a,\gamma_b\}=2\delta_{ab}\), and define
\begin{equation}\label{eq:Gamma-mu-theta}
  \Gamma_\mu(n) = \sum_{a=1}^d \gamma_a\,\bigl(E(n)^{-1}\bigr)_{a\mu} ,
  \qquad
  \Gamma_\mu^{\theta}(n) = \sum_{a=1}^d \gamma_a\,\bigl(E^{\theta}\!(n)^{-1}\bigr)_{a\mu} .
\end{equation}
The prototype naive Dirac operator is
\begin{equation}\label{eq:naive-Dirac-theta}
  (D_{\mathrm{naive}}^\circ\psi)(n)
  =
  \frac12\sum_{\mu=1}^d \Bigl(\Gamma_\mu(n)\, U_\mu(n)\psi(n+\hat\mu)
  - \Gamma_\mu^\theta(n)\,U_{-\mu}(n)\psi(n-\hat\mu)\Bigr),
\end{equation}
and the corresponding fermion action is
\begin{equation}\label{eq:S-fermion-proto-theta}
  S_{f}^\circ[x,\bar\psi,\psi,U]
  =
  \sum_{n\in\Lambda} V(n)\,\bar\psi(n)\,(D_{\mathrm{naive}}^\circ\psi)(n).
\end{equation}
We define the \(\theta\)-partner by the same stencil but with the advanced/retarded geometric data interchanged, i.e.
\begin{equation}\label{eq:naive-Dirac-partner}
  (D_{\mathrm{naive}}^\theta\psi)(n)
  =
  \frac12\sum_{\mu=1}^d \Bigl(\Gamma_\mu^\theta(n)\, U_\mu(n)\psi(n+\hat\mu)
  - \Gamma_\mu(n)\,U_{-\mu}(n)\psi(n-\hat\mu)\Bigr),
\end{equation}
and
\begin{equation}\label{eq:S-fermion-partner}
  S_{f}^{\theta}[x,\bar\psi,\psi,U]
  =
  \sum_{n\in\Lambda} V^{\theta}(n)\,\bar\psi(n)\,(D_{\mathrm{naive}}^\theta\psi)(n).
\end{equation}
Finally set
\begin{equation}\label{eq:S-fermion-avg-theta}
  S_f[x,\bar\psi,\psi,U]
  =
  \tfrac12\Bigl(S_{f}^\circ[x,\bar\psi,\psi,U]+S_{f}^{\theta}[x,\bar\psi,\psi,U]\Bigr).
\end{equation}
On a regular embedding \(x(n)=a\,n\), the forward and backward corner frames coincide (with our convention for \(E^\theta\)),
so \(\Gamma_\mu(n)=\Gamma_\mu^\theta(n)=\gamma_\mu/a\), and \eqref{eq:naive-Dirac-theta} reduces to the usual naive lattice Dirac operator.
As the standard lattice is among the possible embeddings $x(n)$, we do not expect the dynamical
geometry alone to remove species doubling.  Standard remedies (Wilson, staggered, overlap, domain-wall)
can nevertheless be implemented by replacing the corresponding finite-difference operators and Wilson
terms by their geometric analogues built from the same local data $(E,g,V)$ and then applying the
$\theta$--averaging prescription \eqref{eq:fields-action-average-theta}.

\medskip
\noindent\textbf{Total field action.}
In the sequel we write
\begin{equation}\label{eq:fields-total-theta}
S_{\mathrm{fields}}[x,U,\Phi]
\;:=\;
S_g[x,U] + S_\phi[x,\phi,U] + S_f[x,\bar\psi,\psi,U] \;+\; \cdots,
\end{equation}
where the ellipsis indicates additional matter multiplets treated analogously.  By construction,
$S_{\mathrm{fields}}$ is local on the abstract lattice, gauge invariant, and $\Theta$--invariant under
the OS reflection introduced later.


\subsection{Euclidean invariance, orientation, and local isotropy}\label{sec:Ed-invariance}

We record the global Euclidean symmetries of the action and make explicit where a
restriction to proper motions enters.  Recall that the Euclidean group is the semidirect
product
\begin{equation}
  \Ed = O(d)\ltimes \R^d,
\end{equation}
with multiplication $(R,b)(R',b')=(RR',\,b+Rb')$.  We write $\SEd=\mathrm{SO}(d)\ltimes \R^d$ for the
subgroup of \emph{proper} Euclidean motions.

Let $(R,b)\in \Ed$ act on a configuration $(x,U,\Phi)$ by
\begin{equation}\label{eq:Ed-action}
\begin{split}
x(n) &\longmapsto x'(n) := R x(n) + b,\\
U_\mu(n) &\longmapsto U'_\mu(n) := U_\mu(n),\\
\Phi(n) &\longmapsto \Phi'(n) := \rho(R)\,\Phi(n) ,
\end{split}
\end{equation}
where $\rho$ is the appropriate finite-dimensional representation of $O(d)$ on the matter
field (trivial for scalars; spinorial for Dirac fields, discussed below).  Since $b$
is a translation in physical space, it acts trivially on internal indices.

The edge vectors transform as
\begin{equation}
  e_\mu(n)=x(n+\hat\mu)-x(n)\ \mapsto\ e'_\mu(n)=x'(n+\hat\mu)-x'(n)=R e_\mu(n),
\end{equation}
hence the local frame matrix $E(n)=[e_1(n)\ \cdots\ e_d(n)]$ transforms as $E'(n)=R E(n)$.
Therefore the Gram matrix is strictly invariant:
\begin{equation}\label{eq:g-invariance}
  g'(n) = E'(n)^{\mathsf T} E'(n) = E(n)^{\mathsf T} R^{\mathsf T} R E(n) = g(n) .
\end{equation}
For the local cell volume there are two closely related quantities:
\begin{equation}
  \det E'(n)=(\det R)\,\det E(n),\qquad
  |\det E'(n)|=|\det E(n)|,
\end{equation}
so the \emph{unoriented} volume $V(n)=|\det E(n)|$ is invariant under all of $\mathrm{O}(d)$, while
the \emph{oriented} volume $\det E(n)$ is invariant only under $\SOd$.

\medskip
\noindent
\textbf{Invariance of the geometry action.}
Suppose the geometry action has the local form
\begin{equation}
  S_x[x]=\sum_{n\in\Lambda} s_x\bigl(g(n),g(n)^{-1},V(n)\bigr),
\end{equation}
or more generally depends on $x$ only through scalar functions of the Gram matrix 
(e.g.\ eigenvalues of $g(n)$, $\Tr g(n)$, $V(n)$, etc.).  Then
\eqref{eq:g-invariance} and $V'(n)=V(n)$ imply
\begin{equation}
  S_x[x']=S_x[x]\qquad\text{for all }(R,b)\in\Ed.
\end{equation}
If, instead, one chooses to build $S_x$ from the \emph{oriented} determinant $\det E(n)$,
then $S_x$ is automatically $\SEd$--invariant, but is \emph{not} invariant under reflections
with $\det R=-1$ unless $s_x$ is an even function of $\det E(n)$.

\medskip
\noindent
\textbf{Orientation flipping.}
The preceding statements concern the integrand (action) as a function of an \emph{arbitrary}
embedding.  The subtlety is that our \emph{configuration space} for $x$ is typically
restricted by a hard admissibility condition that fixes an orientation, e.g.
\begin{equation}
  \det E(n)\ge V_{\min}>0\quad\text{for all }n.
\end{equation}
This condition is stable under $\SEd$ but not under a reflection: if $\det R=-1$ then
$\det E'(n)=-\det E(n)\le -V_{\min}$, so $x'\notin\Xadm$ even though $g'(n)=g(n)$ and
$|\det E'(n)|=|\det E(n)|$.  Thus:
\begin{quote}
\emph{The local geometric quantities and the local action density can be $\Ed$--invariant,
but the \emph{domain of integration} $\Xadm$ is typically only $\SEd$--invariant when one
imposes a positive-orientation constraint.}
\end{quote}

Choice of orientation breaks \emph{explicit} invariance under improper Euclidean motions (parity / reflections) at finite lattice spacing.  This is a regulator-level choice of orientation sector, and we expect it to be innocuous for the continuum limit of parity-even observables. 
One could restore exact $\Ed$ invariance by integrating over both orientation components, at the cost of enlarging configuration space.
This orientation restriction does \emph{not} interfere with Osterwalder--Schrader reflection positivity: the OS involution $\Theta$ is constructed so that $\Theta\Xadm=\Xadm$ (Lemma~\ref{lem:Xadm-Theta}), even though $\Theta$ contains a target-space reflection.

\medskip
\noindent
\textbf{Exact symmetry of the regulator.}
The partition function is defined as
\begin{equation}\label{e:Z}
  Z=\int_{\mathcal{X}_{\mathrm{adm}}} \! Dx \int DU \int D\Phi\;
  \exp\bigl(-S_x[x]-S_{\mathrm{fields}}[x,U,\Phi]\bigr).
\end{equation}
Since the flat product measure $Dx$ and the Haar/product measures $DU$, $D\Phi$ are invariant
under \eqref{eq:Ed-action}, and since $\mathcal{X}_{\mathrm{adm}}$ is stable precisely under
$\mathrm{SE}(d)$, the regulator has an \emph{exact} global $\mathrm{SE}(d)$ symmetry, provided the field action is $\SEd$-invariant:
\begin{quote}
\emph{$Z$ and all expectation values of observables are invariant under $(R,b)\in\SEd$.}
\end{quote}
If one wishes to be fully precise, the definition is made first on a finite region
$\Lambda_{T,L}$, so that $Dx,DU,D\Phi$ are ordinary finite product measures, and the $\SEd$ covariance
holds there exactly; the infinite-volume partition function is then understood formally, while
expectations of local observables are defined by the limit $\Lambda_{T,L}\nearrow\Z^d$.


To verify global Euclidean invariance for the field action,
it suffices to only deal with the prototype part $S^\circ$.
Indeed, under $x(n)\mapsto Rx(n)+b$, the advanced and retarded frames transform as $E(n)\mapsto RE(n)$ and $E^\theta\!(n)\mapsto RE^\theta\!(n)$, 
so any coefficient built from Euclidean invariants is unchanged; the relabeling $\mu\mapsto\theta(\mu)$ is purely
combinatorial.  Therefore $S^\theta$ inherits the same $\SEd$-invariance as $S^\circ$, and so does the averaged
action $\tfrac12(S^\circ+S^\theta)$.

\medskip
\noindent
\textbf{Scalar actions.}
Consider the gauge-covariant forward difference \eqref{eq:scalar-signed-diff}.
Under \eqref{eq:Ed-action}, we have $|e'_\mu(n)|=|Re_\mu(n)|=|e_\mu(n)|$, $U'_\mu=U_\mu$,
and $\phi'(n)=\rho(R)\phi(n)$, hence
\begin{equation}
  \nabla_\mu\phi'(n)=\rho(R)\,\nabla_\mu\phi(n).
\end{equation}
Assuming the internal inner product on the scalar multiplet is $O(d)$--invariant (as in the
usual continuum theory), we obtain pointwise invariance of the kinetic density:
\begin{equation}
  \sum_{\mu,\nu=1}^d g^{\mu\nu}(n) \langle \nabla_\mu\phi'(n),\nabla_\nu\phi'(n)\rangle
  =
  \sum_{\mu,\nu=1}^d g^{\mu\nu}(n) \langle \nabla_\mu\phi(n),\nabla_\nu\phi(n)\rangle,
\end{equation}
since $g^{\mu\nu}(n)$ is invariant and $\rho(R)$ preserves the inner product.  Multiplying by
$V(n)=|\det E(n)|$ and summing over $n$ yields
\begin{equation}
  S_{\phi,\mathrm{kin}}^\circ[x',\phi',U']=S_{\phi,\mathrm{kin}}^\circ[x,\phi,U].
\end{equation}
The potential term $S_{\phi,\mathrm{pot}}^\circ=\sum_n V(n)W_{\mathrm{loc}}(\phi(n))$ is also invariant
provided $W_{\mathrm{loc}}$ is an $O(d)$--scalar in field space (e.g.\ $\frac12m_0^2|\phi|^2+\frac1{4!}\lambda_0|\phi|^4$).

\medskip
\noindent
\textbf{Gauge actions.}
For the weighted Wilson--type plaquette action \eqref{eq:gauge-action-proto-theta}
the gauge field \(U\) is unchanged by \((R,b)\in\Ed\), so the traces are invariant. Therefore
\(S_g\) is \(E(d)\)-invariant provided each weight \(w^{\mu\nu}(n;x)\) is built from Euclidean
invariants of the local geometry (equivalently, from \(g_{\mu\nu}(n)\), \(g^{\mu\nu}(n)\), \(V(n)\),
and Gram minors). In particular this holds for both weight choices introduced in \eqref{eq:gauge-weight-B} and \eqref{eq:gauge-weight-A}.
Note that as in the continuum, weights built from oriented pseudoscalars would be sensitive to improper
rotations.

\medskip
\noindent\textbf{Fermionic actions.}
In the presence of Dirac fields we work with the proper Euclidean group
$\SEd=\SOd\ltimes\R^d$, since implementing reflections on spinors would require the Pin group.
Independently, our admissible set $\Xadm$ fixes an orientation component, so the $x$--measure is restricted to $\SEd$ rather than $\Ed$.

Let \(R\in \mathrm{SO}(d)\) and choose a lift \(\widehat R\in\mathrm{Spin}(d)\) acting on spinors by
\(\psi'(n)=\widehat R\psi(n)\) and \(\bar\psi'(n)=\bar\psi(n)\widehat R^{-1}\), with the standard
intertwining property
\begin{equation}
\widehat R\gamma_a \widehat R^{-1}=R_a{}^b\,\gamma_b.
\end{equation}
Under $x'(n)=Rx(n)+b$ the local frame matrices transform by
\(
E'(n)=R\,E(n)
\)
and likewise
\(
E^{\theta\,\prime}(n)=R\,E^\theta(n),
\)
since both are built from the same local edge vectors of the embedding.
Consequently the geometry--dependent matrices
\(\Gamma_\mu(n)\) and \(\Gamma_\mu^\theta(n)\)
transform by conjugation:
\begin{equation}
\widehat R \Gamma_\mu(n)\widehat R^{-1}
=\sum_{a}\widehat R\gamma_a\widehat R^{-1}(E^{-1})_{a\mu}
=\sum_{a,b}R_a{}^b\,\gamma_b\,(E^{-1})_{a\mu}
=\sum_b \gamma_b\,(E'^{-1})_{b\mu}
=\Gamma_\mu'(n),
\end{equation}
and similarly $\widehat R \Gamma_\mu^\theta(n)\widehat R^{-1}=\Gamma_\mu^{\theta\,\prime}(n)$.
Since the gauge links $U_\mu(n)$ are unchanged by Euclidean motions and act only on gauge indices,
the mixed prototype naive Dirac operator \eqref{eq:naive-Dirac-theta} transforms covariantly:
\begin{equation}
(D_{\mathrm{naive}}^\circ\psi')(n)=\widehat R(D_{\mathrm{naive}}^\circ\psi)(n).
\end{equation}
Using also $V'(n)=V(n)$ (and $V^{\theta\,\prime}(n)=V^\theta(n)$) we conclude invariance of the prototype action:
\begin{equation}
S_f^\circ[x',\bar\psi',\psi',U]
=\sum_{n}V(n)\,\bar\psi(n)\widehat R^{-1}\widehat R(D_{\mathrm{naive}}^\circ\psi)(n)
=S_f^\circ[x,\bar\psi,\psi,U].
\end{equation}

\noindent
\textbf{Local isotropy via twisting.}
It is useful to separate two notions that are sometimes conflated under the phrase
``rotationally invariant regulator.''
The first is the \emph{exact} global covariance of the partition function under $\SEd$,
which follows formally because the action depends on the embedding only through Euclidean invariants.
However, this formal covariance is not, by itself, the operative mechanism at finite cutoff:
it would remain true even if the geometry measure were concentrated near a rigid $\SEd$--orbit
$x(n)\approx aRn+b$, in which case typical configurations exhibit essentially no local orientation mixing
and one recovers the usual hypercubic anisotropies.
What matters is therefore \emph{how} global $\SEd$ covariance is realised in typical configurations
after fixing/quotienting the global zero modes.

The second notion is a \emph{local isotropy} (or local orientation--mixing) property.
In the short-range geometry regime (SR) required for a local Symanzik description after integrating out
geometry, the geometry sector cannot realise $\SEd$ covariance ``rigidly'' through long-range geometric modes;
rather, global covariance must be implemented through \emph{short-range fluctuations of local frames}.
We make this precise by extracting from the embedding a local orientation field.
Define the dimensionless forward frame $\widetilde E(n)=a^{-1}E(n)$.  By admissibility $\widetilde E(n)$ is
invertible, hence it has a polar decomposition
\begin{equation}\label{eq:polar-QS}
\widetilde E(n)=Q(n)\,S(n),
\qquad Q(n)\in\SOd,\quad S(n)\in \mathrm{SPD}(d).
\end{equation}
Note that $\det \widetilde E(n)>0$ by admissibility, so the polar factor $Q(n)$ lies in $\SOd$.
We interpret $Q(n)$ as the \emph{local orientation} and $S(n)$ as the \emph{local shape}.

From $Q(n)$ we build an $SO(d)$-valued discrete connection on links,
\begin{equation}\label{eq:Omega-conn}
R_\mu(n)=Q(n)^{-1}Q(n+\hat\mu)\in\SOd,
\end{equation}
and a plaquette holonomy (curvature proxy)
\begin{equation}\label{eq:Holonomy-plaquette}
\mathcal R_{\mu\nu}(n)
=R_\mu(n) R_\nu(n+\hat\mu) R_\mu(n+\hat\nu)^{-1} R_\nu(n)^{-1}\in\SOd.
\end{equation}
These observables distinguish rigid orbit averaging from genuine local twisting:
if $x(n)=aRn+b$ is rigid then $Q(n)\equiv R$ and hence $R_\mu(n)\equiv I$ and
$\mathcal R_{\mu\nu}(n)\equiv I$; conversely, broad fluctuations of $R_\mu$ together with short-range
decorrelation indicate nontrivial local twisting.

\smallskip
\noindent\emph{Local isotropy is a statistical property, not an exact local symmetry.}
There is no underlying local $\SOd$ gauge redundancy; rather, the mechanism is that the distribution
of $R_\mu(n)$ is sufficiently mixing and its correlations decay rapidly.  
We could quantify this via class functions $f:\SOd\to\R$ (e.g.\ $f(R)=\tr(R)$) and the connected
correlator
\begin{equation}\label{eq:twist-corr}
\big\langle f(R_\mu(n))\,f(R_\mu(m))\big\rangle_c,
\end{equation}
where $\langle AB\rangle_c = \langle AB\rangle-\langle A\rangle\langle B\rangle$ denotes the connected
correlator.  The quantity \eqref{eq:twist-corr} measures how rapidly the local frame rotations decorrelate with
separation.  In a near-rigid regime (dominated by global zero modes) one has $R_\mu(n)\approx I$ and
the fluctuations of $R_\mu$ are small and long-range, whereas in a genuine twisting regime the
distribution of $R_\mu(n)$ is broad and the connected correlator decays on a short correlation length
$\xi_{\mathrm{twist}}$ (typically a few lattice spacings).  This short-range mixing of local orientations
is precisely the ingredient that can reduce axis-versus-diagonal (direction-dependent) cutoff
artefacts in matter correlators at finite $a$.  

\subsection{Gauge invariance and lattice symmetries}
\label{sec:gauge-lattice-symmetries}

We briefly recall the gauge and lattice symmetries of the
dynamical--lattice regulator.  These act on the abstract hypercubic
graph $\Lambda$ and on the gauge/matter fields, while the geometry
field $x$ is gauge neutral and transforms only under global Euclidean
transformations (\S\ref{sec:Ed-invariance}).

\medskip

\noindent
\textbf{Gauge invariance.}
Let $G$ be a compact gauge group and let
\begin{equation}
\Omega : \Lambda \to G,\qquad n\mapsto \Omega(n),
\end{equation}
be a lattice gauge transformation.  We let $\Omega$ act on the link
variables and on the matter fields in the usual way:
\begin{align}
  U_\mu(n) &\;\longmapsto\;
  U_\mu^\Omega(n) := \Omega(n) U_\mu(n)\Omega(n+\hat\mu)^{-1},
  \label{eq:gauge-U}\\[0.3em]
  \Phi(n) &\;\longmapsto\;
  \Phi^\Omega(n) := \rho_{\mathrm{m}}\!\bigl(\Omega(n)\bigr)\,\Phi(n),
  \label{eq:gauge-Phi}
\end{align}
where $\rho_{\mathrm{m}}$ is the representation of $G$ carried by
the matter multiplet.  The geometry field $x(n)$ is unaffected:
\begin{equation}
  x(n) \;\longmapsto\; x^\Omega(n) := x(n).
\end{equation}

\noindent
As in the Euclidean invariance case,
gauge invariance need only be verified for the prototype action $S^\circ$, since $S^\theta$ is defined
by the same local stencils with direction labels relabeled by $\mu\mapsto\theta(\mu)$ and with
geometry-dependent weights replaced by functions of $x$ only.  As these operations commute with the
local gauge transformation \eqref{eq:gauge-U}-\eqref{eq:gauge-Phi},
we have $S^\theta[U^\Omega,\Phi^\Omega]=S^\theta[U,\Phi]$, and hence the averaged action is gauge invariant.

By construction the prototype action $S_{\mathrm{fields}}^\circ[x,U,\Phi]$ is a
sum of gauge invariant local terms built from traces of Wilson loops
and gauge covariant combinations of $\Phi$ and its lattice covariant
derivatives.  For example, the gauge part of the action may be written as
\begin{equation}
S_g^\circ[x,U]
  = \sum_{n,\mu<\nu} w^{\mu\nu}(x;n)\,
    \Re\tr\bigl(I - U_{\mu\nu}(n)\bigr),
\end{equation}
with local geometry--dependent weights $w_{\mu\nu}(x;n)$ and
plaquette variables $U_{\mu\nu}(n)$ built from the links $U_\mu(n)$ in
the usual way.  Under \eqref{eq:gauge-U} the plaquettes transform by
conjugation and the traces are invariant.  Likewise, the matter part of
$S_{\mathrm{fields}}^\circ$ is built from gauge covariant bilinears such as
$\bar\Phi(n)\,\Phi(n)$ and
$\bar\Phi(n)\,D_\mu\Phi(n)$, which are invariant under
\eqref{eq:gauge-Phi}.  Since $x$ is gauge neutral and appears only
through gauge invariant weights such as $w_{\mu\nu}(x;n)$, we have
\begin{equation}
  S_{\mathrm{fields}}^\circ[x,U^\Omega,\Phi^\Omega]
  = S_{\mathrm{fields}}^\circ[x,U,\Phi]
  \qquad\text{for all } \Omega:\Lambda\to G,
\end{equation}
so the full action \eqref{eq:full-action} is gauge invariant for every
geometry configuration $x\in\Xadm$.  The geometry action $S_x[x]$
is gauge invariant trivially, as it depends only on $x$.

\medskip

\noindent
\textbf{Lattice translations and hypercubic symmetries.}
In addition to gauge transformations, the theory is invariant under lattice
translations and under those hypercubic relabellings of the abstract lattice
that preserve the admissible domain $\Xadm(a)$.
Let $\tau_k:\Lambda\to\Lambda$ be the translation $\tau_k(n)=n+k$, and let
\begin{equation}
\Hd=\{R\in \GL(d,\Z): R^{\mathsf T} R=I\},\qquad
\mathrm H^+(d)=\{R\in \Hd: \det R = +1\}.
\end{equation}
We let these maps act on fields by pullback:
\begin{equation}
  (\tau_k x)(n):=x(n-k),\qquad
  (\tau_k \Phi)(n):=\Phi(n-k),\qquad
  (\tau_k U)(\ell):=U(\tau_k^{-1}\ell),
  \tag{3.49}
\end{equation}
and for $R_{\rm lat}\in \Hd$,
\begin{equation}
  (R_{\rm lat}x)(n):=x(R_{\rm lat}^{-1}n),\qquad
  (R_{\rm lat}\Phi)(n):=\Phi(R_{\rm lat}^{-1}n),\qquad
  (R_{\rm lat}U)(\ell):=U(R_{\rm lat}^{-1}\ell),
  \tag{3.50}
\end{equation}
where $U$ is viewed as a function on oriented links $\ell=(n\to n+\hat\mu)$,
and $R_{\rm lat}$ acts on oriented links in the obvious way.
Since the action is a sum of identical local stencil terms and the product
measures are invariant under relabelling of factors (including $U\mapsto U^{-1}$
on links by Haar invariance), these transformations preserve the action and
measure whenever they preserve $\Xadm(a)$.

With the basic definition of $\Xadm(a)$ in terms of a distinguished
corner frame and a selected orientation component, we therefore only claim
invariance under the corresponding subgroup of $\mathrm H^+(d)$ that fixes that
corner choice. If instead one imposes the strengthened all--corner
oriented--volume condition of Remark~\ref{r:adm}(A3), then $\Xadm(a)$ is invariant
under the full orientation--preserving hypercubic group $\mathrm H^+(d)$.
Permutations with $\det=-1$ exchange the two orientation components and hence would not be included.

\medskip

\noindent
\textbf{Locality.}
We say that a functional $F[x,U,\Phi]$ is \emph{local} if it can be
written as a sum of terms
\begin{equation}
F[x,U,\Phi] = \sum_{n\in\Lambda} f_n,
\end{equation}
where each $f_n$ depends only on the restriction of $(x,U,\Phi)$ to a
finite neighbourhood of $n$, with a radius (in lattice distance) that
is independent of $n$ and of the lattice size.  In particular, the
geometry action $S_x[x]$ in \eqref{eq:geom-action-local} and the field
action $S_{\mathrm{fields}}[x,U,\Phi]$ are local in this sense: the
local density at site $n$ depends on the geometry only through the
frame $E(n)$, or a bounded stencil of neighbouring frames, and on the
gauge/matter fields only through fields on a fixed finite stencil of
sites and links around $n$.

Thus, at the level of the abstract graph, the dynamical--lattice
regulator has exactly the same locality and lattice symmetry structure
as standard hypercubic lattice gauge theory.  The only difference is
that the local couplings are promoted to functions of the dynamical
geometry via the metric $g(n)$ and volume $V(n)$.  The
global Euclidean invariance in the embedding space $\R^d$ is an
additional continuous symmetry acting on $x$, which we discussed in Section~\ref{sec:Ed-invariance}.

\subsection{Viewpoint as an extended lattice field theory}
\label{sec:extended-view}

It is conceptually useful to reinterpret the dynamical--lattice regulator
purely as a lattice field theory on the abstract hypercubic graph
$\Lambda$, with an enlarged field content.  Instead of thinking of $x :
\Lambda \to \R^d$ as an embedding, one may regard its Cartesian
components $x_\mu(n)$ as additional
real scalar fields attached to sites or links.  From this point of view
the configuration space is simply
\begin{equation}
\bigl\{ x_\mu(n),\,U_\mu(n),\,\Phi(n) : n\in\Lambda,\ \mu=1,\dots,d \bigr\},
\end{equation}
and the geometry--dependent quantities $e_\mu(n)$, $g_{\mu\nu}(n)$ and
$V(n)$ are just local composite fields built from the $x_\mu(n)$ on a finite stencil.  The action $S[x,U,\Phi]$ defined in
\eqref{eq:full-action} is then a sum of strictly local terms on this
extended field space, and the path integral \eqref{eq:partition-function}
is an ordinary Euclidean lattice path integral with local interactions.

In this extended viewpoint the main structural properties of the
regulator follow from standard lattice arguments applied to the enlarged
field content.  Locality and gauge invariance are manifest at the level
of the abstract graph, as discussed in
Section~\ref{sec:gauge-lattice-symmetries}.  Reflection positivity is
proved by choosing a reflection of the time index $n_1$ on $\Lambda$
and checking the usual Osterwalder--Schrader inequalities for the
extended field multiplet $(x,U,\Phi)$; the detailed construction of the
reflection, the splitting of the lattice into time slices, and the
verification of the OS axioms are carried out in
Section~\ref{sec:OS-universality}.  In particular, no continuum notion of ``time''
or preferred coordinate direction in $\R^d$ is required at the
microscopic level: all reflection operations are defined purely on the
index lattice, and the embedding variables $x_\mu(n)$ enter only through
local, reflection invariant combinations.  This perspective makes clear
that the dynamical--lattice regulator fits squarely within the usual
constructive framework for Euclidean lattice field theory, with the
dynamical geometry fields providing an $\Ed$--invariant way of dressing
local couplings without sacrificing locality or reflection positivity.

\section{Reflection positivity and continuum universality}
\label{sec:OS-universality}

In this section we first establish Osterwalder--Schrader (OS) reflection positivity for a large class of DLR, 
and briefly recall the associated Hilbert space reconstruction.
We then discuss the Symanzik effective action description of the continuum limit and explain
why, under mild assumptions on the geometry sector, the dynamical--lattice theory lies in the
same universality class as the corresponding theory on a static hypercubic lattice. As a
concrete check we sketch the one--loop running of the quartic coupling in scalar $\phi^{4}$
theory.

\subsection{Reflection positivity for the bosonic sector}
\label{subsec:OS-fixed-geom}\label{subsec:OS-full}

In this subsection we treat the bosonic matter sector, i.e., gauge fields coupled to complex scalar fields $\phi$; 
the fermionic sector is treated separately in the next subsection.

Note that for a \emph{generic} admissible embedding
$x\in\Xadm$, the induced matter/gauge action need not be reflection symmetric \emph{at fixed $x$}.
Accordingly, we cannot hope to obtain ``conditional'' OS positivity at fixed geometry.
Instead, we establish reflection positivity directly for the \emph{joint} Euclidean measure of the
extended field multiplet $\omega=(x,U,\phi)$.

We single out the first lattice direction $n_{1}$ as Euclidean time and use the \emph{link reflection}
\begin{equation}\label{eq:OS-theta}
  \theta(n_{1},n_{2},\dots,n_{d}) = (1-n_{1},n_{2},\dots,n_{d}) .
\end{equation}
This induces the reflection $\theta\mu$ on directions $\mu$ as in \eqref{eq:theta-on-directions},
and the OS involution $\Theta$ on $\omega$ by
\begin{equation}\label{eq:OS-Theta}
  (\Theta x)(n) = r(x(\theta n)),
  \qquad
  (\Theta U)_{\mu}(n) = U_{\theta\mu}(\theta n) ,
  \qquad
  (\Theta\phi)(n) = \overline{\phi(\theta n)},
\end{equation}
where $r\in\mathrm{O}(d)$ is the target space reflection $r(v_{1},v_{2},\dots,v_{d})=(-v_{1},v_{2},\dots,v_{d})$,
the reflection $\theta\mu$ on the direction $\mu$ is as in \eqref{eq:theta-on-directions},
and the overline denotes complex conjugation.

\begin{lemma}\label{lem:Xadm-Theta}
We have $x\in\Xadm(a)$ if and only if $\Theta x\in\Xadm(a)$.
\end{lemma}

\begin{proof}
Let us encode the dependence of $e_\mu$ on the embedding $x$ as
\begin{equation}
e_\mu(x;n) = x(n+\hat\mu)-x(n), \qquad \mu\in\Delta .
\end{equation}
Then \eqref{eq:OS-Theta} gives
\begin{equation}\label{e:edge-Theta}
\begin{split}
e_\mu(\Theta x;n)
&=(\Theta x)(n+\hat\mu)-(\Theta x)(n)
 =rx(\theta(n+\hat\mu)) - rx(\theta n) \\
&=r\bigl(x(\theta n+\widehat{\theta\mu})-x(\theta n)\bigr)
 =r\,e_{\theta\mu}(x;\theta n).
\end{split}
\end{equation}
Since $r\in\mathrm{O}(d)$, we have $|e_\mu(\Theta x;n)|=|e_{\theta\mu}(x;\theta n)|$.
Therefore any uniform length bounds imposed on $\{|e_\mu(n)|:\mu\in\Delta\}$ are preserved,
up to the harmless relabelling $n\mapsto\theta n$.

Applying \eqref{e:edge-Theta} to $E$ and $E^\theta$ columnwise gives
\begin{equation}
E(\Theta x;n)=r E^\theta(x;\theta n),\qquad
E^\theta\!(\Theta x;n)=r E(x;\theta n).
\end{equation}
Taking determinants and using $\det(r)=-1$ yields
\begin{equation}
\det E(\Theta x;n)= -\det E^\theta(x;\theta n),
\qquad
\det E^\theta\!(\Theta x;n)= -\det E(x;\theta n),
\end{equation}
so the oriented volume bounds appearing in Definition~\ref{def:admissible-geometry}
are mapped into each other, again up to relabelling $n\mapsto\theta n$.
Hence $\Theta$ maps $\XADM(a)$ bijectively to itself.

Finally, $\Theta$ is continuous and $\Theta^2=\mathrm{id}$, so it maps connected components of
$\XADM(a)$ to connected components.  Since $\Theta x_{\mathrm{reg}}=x_{\mathrm{reg}}$,
the principal component containing $x_{\mathrm{reg}}$ is preserved. Therefore
$\Theta(\Xadm(a))=\Xadm(a)$, which proves the claim.
\end{proof}

We consider the probability measure
\begin{equation}\label{eq:joint-measure}
  d\mu(\omega)
  =
  \frac{1}{Z}
  \mathbf{1}_{\{x\in\Xadm\}}
  \exp\bigl(-S_{x}[x]-S_{\mathrm{fields}}[x,U,\phi]\bigr)
  Dx DU D\phi,
\end{equation}
where $Dx$ is the product Lebesgue measure on $(\R^{d})^{\Lambda}$, and $DU,D\phi$ are the usual Haar/Lebesgue
measures for the gauge/matter fields.  
We first note that Lemma~\ref{lem:Xadm-Theta} implies that the hard constraint
$\mathbf{1}_{\{x\in\Xadm(a)\}}$ is invariant under $\Theta$.
Moreover, the geometric and field actions will always be taken to be local and $\Theta$--invariant.
For the concrete examples of Section~\ref{sec:full-action} this invariance holds by design, and in general it is
a standing assumption on the choice of $S_x$ and $S_{\mathrm{fields}}$.
Thus, to establish $\Theta$--invariance of \eqref{eq:joint-measure} it remains to check $\Theta$--invariance of
the product measures $Dx$, $DU$ and $D\phi$.

To make these statements completely precise one should work at finite volume, where all measures are
honest finite products.  We therefore fix a finite cylinder
\begin{equation}
\Lambda_{T,L}=\{1-T,\ldots,0,1,\dots,T\}\times (\Z/L\Z)^{d-1},
\end{equation}
with open boundary conditions in the time direction and periodic boundary conditions in space.
All local action terms are defined by the same formulas as on $\Z^d$, but we include a term
in the finite-volume sum only if every site/link variable appearing in its stencil lies in $\Lambda_{T,L}$.
Equivalently, for each local contribution $t_n(\varepsilon)$ anchored at a basepoint $n$, we keep $t_n$
iff $\mathrm{stencil}(t_n)\subset \Lambda_{T,L}$; otherwise it is omitted.  For nearest-neighbour
actions this removes only the finitely many terms that would require sites outside the time boundaries.

All local action terms are defined by the same formulas as on $\Z^d$, and we sum only over those
basepoints for which every site/link appearing in the corresponding finite stencil lies in $\Lambda_{T,L}$.
In particular, for the two-frame geometry data $(E,E^\theta)$ we only use sites $n$ for which both
$n\pm\hat 1\in\Lambda_{T,L}$ (spatial neighbours exist by periodicity).

We define the finite-volume admissible set $\Xadm^{T,L}(a)$ by imposing the $a$-admissibility inequalities
only at those sites $n\in\Lambda_{T,L}$ for which the required neighbour stencil is contained in $\Lambda_{T,L}$.
With this convention the index set of admissibility constraints is $\theta$-invariant, hence so is $\Xadm^{T,L}(a)$.
For notational simplicity we will not introduce a special notation for finite-volume objects: in what follows
$\Lambda$, $Dx$, $DU$, $D\phi$, and $\Xadm(a)$ refer to the above finite cylinder, and all identities are proved
at fixed $(T,L)$ and then interpreted on the infinite lattice by taking $T,L\to\infty$ on local observables.

\begin{lemma}\label{lem:measure-Theta}
The product measures $Dx$, $DU$ and $D\phi$ are invariant under $\Theta$.
\end{lemma}

\begin{proof}
At finite volume all measures are honest finite products, so it suffices to check invariance
factor by factor.
The map $x\mapsto \Theta x$ is given by $(\Theta x)(n)=r(x(\theta n))$: it permutes the site labels
and applies the same orthogonal map $r\in\mathrm{O}(d)$ to each $x(n)\in\R^d$.  A permutation of factors
preserves a product Lebesgue measure, and an orthogonal map has $|\det r|=1$, hence
$Dx=\prod_n d^d x(n)$ is $\Theta$-invariant.

For a complex scalar multiplet, $(\Theta\phi)(n)=\overline{\phi(\theta n)}$ is again a site permutation
composed with complex conjugation, which is an orthogonal linear map on $\R^{2m}$ and therefore
preserves Lebesgue measure; thus $D\phi$ is $\Theta$-invariant. Obviously, for real multiplets the conjugation is absent.

Finally, view the gauge field as assigning a group element to each oriented link.  Under link
reflection, $\Theta$ maps each link variable to another link variable or to its inverse.
Since Haar measure on a compact Lie group is invariant under inversion $U\mapsto U^{-1}$ and
under relabelling of factors, the product Haar measure $DU$ is $\Theta$-invariant.
\end{proof}

We state the reflection--positivity criterion for the coupled DLR measure.
In the concrete examples of \S\ref{sec:full-action}, the total action $S=S_x+S_{\mathrm{fields}}$
is local on the abstract lattice and is $\Theta$--invariant by construction (each prototype term is paired
with its $\theta$--partner and then averaged), and the admissible set $\Xadm(a)$ is $\Theta$--stable.
In the theorem we abstract these features and assume only locality, $\Theta$--invariance, and finite
interaction range in the index--time direction.

To define the positive--time algebra we fix the link reflection plane between the time slices
$n_1=0$ and $n_1=1$ and write $\Lambda=\Lambda_-\cup\Lambda_+$ with
\begin{equation}
\Lambda_- = \{n\in\Lambda:\ n_1\le 0\} ,\qquad
\Lambda_+ = \{n\in\Lambda:\ n_1\ge 1\} .
\end{equation}
We say that a functional $F(\omega)$ is \emph{supported in $\Lambda_+$} if it depends only on
(i) site variables located at sites $n\in\Lambda_+$, and (ii) link variables $U_\mu(n)$ whose two
endpoints $n$ and $n+\hat\mu$ both lie in $\Lambda_+$.
In particular, the support convention forbids dependence on any link variable whose endpoints lie on opposite sides of the reflection plane.

Let $\cA_+$ be the algebra generated by bounded functionals supported in $\Lambda_+$, and define $\Theta F$ by 
\begin{equation}
(\Theta F)(\omega) = \overline{\,F(\Theta\omega)\,},\qquad \omega=(x,U,\phi),
\end{equation}
where the overline denotes complex conjugation. 
In particular, $\Theta$ is antilinear.
With the above support convention, $\Theta$ maps $\cA_+$ into the corresponding algebra on $\Lambda_-$.

Here and below, for functionals $F,G$ on configuration space we write
$(F\cdot G)(\omega)=F(\omega)\,G(\omega)$ for their pointwise product, so that
$\langle \Theta F\cdot F\rangle=\int (\Theta F)(\omega)\,F(\omega)\,d\mu(\omega)$.

\begin{theorem}\label{thm:OS-full}
Assume that the total action $S=S_{x}+S_{\mathrm{fields}}$ is local and $\Theta$--invariant, i.e.
\begin{equation}\label{eq:action-invariant}
  S[\Theta x,\Theta U,\Theta\phi] = S[x,U,\phi]
  \qquad \text{for all } x\in\Xadm(a),
\end{equation}
and that the interaction has finite range in the index--time direction: there exists an integer
$\ell\ge 1$ such that each local term in $S$ depends only on variables whose site time--indices $n_1$
lie in an interval of length at most $\ell$ (equivalently, the $n_1$--diameter of the stencil is $\le \ell$).
Then
\begin{equation}\label{eq:OS-ineq}
  \langle\,\Theta F\cdot F\,\rangle \;\ge\; 0
  \qquad \text{for all } F\in\mathcal{A}_{+},
\end{equation}
where $\langle\cdot\rangle$ denotes expectation with respect to the joint measure \eqref{eq:joint-measure}.
\end{theorem}

\begin{proof}
We use the standard link--reflection factorization argument.  Fix $\ell$ as in the hypothesis and
introduce the finite reflection slab
\[
  \Sigma_\ell = \{\,n\in\Lambda:\ 1-\ell \le n_1 \le \ell\,\}.
\]
For the nearest--neighbour/plaquette actions of \S\ref{sec:full-action} one has $\ell=1$, hence $\Sigma_\ell=\{n_1\in\{0,1\}\}$.
Decompose a configuration as $\omega=(\omega_-,\omega_0,\omega_+)$ by declaring:
(i) $\omega_0$ consists of all site variables $(x(n),\phi(n))$ with $n\in\Sigma_\ell$, together with
all link variables $U_\mu(n)$ whose link $\{n,n+\hat\mu\}$ meets $\Sigma_\ell$;
(ii) $\omega_+$ consists of the remaining variables with support in the strict future $\{n_1\ge \ell+1\}$
(sites) and links with both endpoints in $\{n_1\ge \ell+1\}$; and (iii) $\omega_-$ consists of the
remaining variables with support in the strict past $\{n_1\le -\ell\}$ (sites) and links with both
endpoints in $\{n_1\le -\ell\}$.

By the finite time--range assumption, no local interaction term can simultaneously involve variables
from the strict past $\{n_1\le -\ell\}$ and the strict future $\{n_1\ge \ell+1\}$.
Hence every local term belongs to exactly one of the following three classes:
it is supported entirely in $\Sigma_\ell$; or it involves at least one variable from the strict future;
or it involves at least one variable from the strict past.  Therefore the total action decomposes as
\begin{equation}\label{eq:action-slab-split}
  S(\omega)=S_+(\omega_0,\omega_+)+S_0(\omega_0)+S_-(\omega_0,\omega_-),
\end{equation}
where $S_0$ is the sum of all terms supported in $\Sigma_\ell$, and $S_+$ (resp.\ $S_-$) collects all
remaining terms that involve at least one strict-future (resp.\ strict-past) variable; in particular
$S_+$ depends only on $(\omega_0,\omega_+)$ and $S_-$ only on $(\omega_0,\omega_-)$.
Using that $\Xadm$ is $\Theta$--stable and that the product measures are invariant under the induced
relabelings, together with $(\Theta F)(\omega)=\overline{F(\Theta\omega)}$, one obtains
\begin{equation}\label{eq:Sminus-ThetaSplus}
  S_-(\omega_0,\omega_-)= (\Theta S_+)(\omega_0,\omega_-).
\end{equation}
Now let $F\in\cA_+$, so $F$ depends only on the variables in $\Lambda_+$ and thus only on $(\omega_0,\omega_+)$.
Define
\begin{equation}\label{eq:PsiF-def}
  \Psi_F(\omega_0)
  =
  \int \exp\bigl(-S_+(\omega_0,\omega_+)\bigr)\,F(\omega_0,\omega_+)\,d\omega_+,
\end{equation}
where $d\omega_+$ denotes the product measure over the variables comprising $\omega_+$.  
Using $\Theta$--invariance of the domain and product measures together with
\eqref{eq:action-slab-split}--\eqref{eq:Sminus-ThetaSplus} and $(\Theta F)(\omega)=\overline{F(\Theta\omega)}$, we obtain
\begin{equation}
  \langle \Theta F\cdot F\rangle
  =
  \frac{1}{Z}\int \exp\bigl(-S_0(\omega_0)\bigr)\,
  \overline{\Psi_F(\omega_0)}\,\Psi_F(\omega_0)\,d\omega_0
  =
  \frac{1}{Z}\int e^{-S_0(\omega_0)}\,|\Psi_F(\omega_0)|^2\,d\omega_0
  \;\ge\;0,
\end{equation}
which is \eqref{eq:OS-ineq}.  
\end{proof}

\begin{remark}\label{rem:OS-comments}
Only one reflection plane is needed for OS reconstruction. Site and link reflections, when both
are available, differ by a one--step lattice translation. 
Here we work with the link reflection between the slices $n_1=0$ and $n_1=1$.
\end{remark}

\subsection{Fermions and reflection positivity}\label{sec:fermions-RP}

We extend the coupled DLR measure to include fermions by introducing Grassmann-valued lattice
spinors. Work at finite volume on the cylinder $\Lambda=\Lambda_{T,L}$ as in the preceding subsection,
with the understanding that statements on the infinite lattice are interpreted by taking the limit
$T,L\to\infty$ on local observables.
At each site $n\in\Lambda$ let $\psi(n)$ and $\bar\psi(n)$ be independent Grassmann
generators, carrying spin and internal indices and transforming in a unitary representation of $G$.
The fermionic integration is the Berezin product measure
\[
D\bar\psi\,D\psi=\prod_{n\in\Lambda}\prod_{\alpha} d\bar\psi_\alpha(n)\, d\psi_\alpha(n),
\]
for some fixed ordering of the site/index pairs.

We keep the support convention of the preceding subsection: a functional is supported in $\Lambda_+$
if it depends only on site variables at sites $n\in\Lambda_+$ and on link variables $U_\mu(n)$
whose endpoints $n$ and $n+\hat\mu$ both lie in $\Lambda_+$.
Let $\cA_+$ be the algebra generated by bounded bosonic functionals supported in $\Lambda_+$
together with polynomials in the Grassmann generators $\{\psi(n),\bar\psi(n)\}_{n\in\Lambda_+}$.
Let $\cA_+^{\mathrm{even}}\subset\cA_+$ denote its even subalgebra.

Fix constant Euclidean gamma matrices $\{\gamma_a\}_{a=1}^d$ with $\{\gamma_a,\gamma_b\}=2\delta_{ab}$.
As in the standard lattice reflection-positivity setup, we take a representation in which the
relevant transpose identities hold; concretely, it suffices that the gamma matrices are chosen so that
the matrix involutions introduced below satisfy $\Xi(n)^{\sf T}=\Xi(n)$ and $\Xi(n)^2=I$.

Recall that the local frame matrix $E(n)$ is built from the outgoing edge vectors at $n$ and that
$E^\theta(n)$ denotes the time-reflected corner frame (same spatial columns, but using the backward
time edge). From a frame $E^\sharp(n)$, which may be either $E$ or $E^\theta$, we form the corresponding
geometry-dependent Dirac matrices
\[
\Gamma_\mu^\sharp(n)=\sum_{a=1}^d \gamma_a\,(E^\sharp(n)^{-1})_{a\mu},
\qquad
\{\Gamma_\mu^\sharp(n),\Gamma_\nu^\sharp(n)\}=2\,g_\sharp^{\mu\nu}(n),
\]
so in particular $(\Gamma_1^\sharp(n))^2=g_\sharp^{11}(n)\,I$ and the normalized matrix
\(\widehat\Gamma_1^\sharp(n)={\Gamma_1^\sharp(n)}/{\sqrt{g_\sharp^{11}(n)}}\)
is an involution, $(\widehat\Gamma_1^\sharp(n))^2=I$.

For the Osterwalder--Schrader reflection across the plane between $n_1=0$ and $n_1=1$ it is convenient
to use, at each site, the normalized ``time'' matrix pointing \emph{toward} the reflection plane:
\begin{equation}\label{eq:Xi-def}
  \Xi(n)
  =
  \begin{cases}
    \widehat\Gamma_1(n), & n\in\Lambda_- \ \textrm{with}\ n_1\le 0,\\[0.3em]
    \widehat\Gamma_1^\theta(n), & n\in\Lambda_+ \ \textrm{with}\ n_1\ge 1.
  \end{cases}
\end{equation}
This choice is entirely local (it uses only the corner frame attached to the site $n$) and is tailored
so that the same geometric ``time direction toward the plane'' is used on both sides. Other local
choices of $\Xi$ are possible; \eqref{eq:Xi-def} is a particularly transparent one because it is an
involution by construction and reduces to the constant matrix $\gamma_1$ on a regular embedding.

Now define $\Theta$ on the fermionic generators by
\begin{equation}\label{eq:Theta-fermions}
  (\Theta\psi)(n) := \Xi(n)\,\bar\psi(\theta n)^{\sf T},
  \qquad
  (\Theta\bar\psi)(n) := \psi(\theta n)^{\sf T}\,\Xi(n),
\end{equation}
where ${}^{\sf T}$ denotes transpose in the spin/internal indices (no complex conjugation is applied to
the Grassmann generators themselves). On products and complex scalars we impose the graded
anti-automorphism rules
\begin{equation}\label{eq:Theta-graded}
  \Theta(AB)=\Theta(B)\Theta(A),\qquad
  \Theta(cA)=\overline{c}\,\Theta(A)\quad (c\in\C),
\end{equation}
so $\Theta$ is anti-linear and reverses the order of Grassmann factors. With the involution properties
of $\Xi$ (in particular $\Xi^{\sf T}=\Xi$ and $\Xi^2=I$), one has $\Theta^2=\mathrm{id}$ on
$\cA_+^{\mathrm{even}}$. As before, the dot in $\langle \Theta F\cdot F\rangle$ denotes the pointwise
product in this graded algebra.

The map induced by $\Theta$ on the set of Grassmann generators is a site permutation composed with an
invertible linear transformation on the spin/internal indices with determinant $\pm1$. In particular,
it has unit Berezin Jacobian, and any global reordering sign is irrelevant on $\cA_+^{\mathrm{even}}$.
Hence we have
\begin{equation}\label{eq:Berezin-inv}
  D(\Theta\bar\psi)\,D(\Theta\psi)=D\bar\psi\,D\psi.
\end{equation}
Together with the already established $\Theta$-invariance of $Dx$, $DU$, and $D\phi$, the full product
measure $Dx\,DU\,D\phi\,D\bar\psi\,D\psi$ is $\Theta$-invariant at finite volume.

As in \S\ref{sec:full-action}, we build fermion actions from local geometric data and finite stencils
on the abstract lattice, and then apply the $\theta$--partner/averaging prescription to enforce
$\Theta$--invariance. In particular, the averaged naive action
\eqref{eq:S-fermion-avg-theta} is local and has time-range one in the index--time direction.

As in the proof of Theorem~\ref{thm:OS-full}, we decompose the extended field multiplet
$\omega=(x,U,\phi,\bar\psi,\psi)$ as $\omega=(\omega_-,\omega_0,\omega_+)$ using a reflection slab
$\Sigma_\ell$ wide enough to contain the time-range $\ell$ of the stencil. By finite time-range, no
local term can simultaneously involve strict-past variables and strict-future variables; consequently
one can (uniquely, up to regrouping purely $\omega_0$ terms) write
\begin{equation}\label{eq:ferm-slab-split}
  S_f(\omega)=S_{f,+}(\omega_0,\omega_+)+S_{f,0}(\omega_0)+S_{f,-}(\omega_0,\omega_-),
  \qquad
  S_{f,-}=\Theta S_{f,+}.
\end{equation}
For fermions, locality and $\Theta$--invariance are not by themselves sufficient: the additional input
is positivity of the \emph{slab pairing} across the reflection plane, which we isolate next.

We say that a Grassmann-even fermion action $S_f$ has a \emph{reflection-positive slab form (of width
$\ell$)} if the slab weight admits a finite ``sum of $\Theta$--squares'' representation
\begin{equation}\label{eq:ferm-sum-of-squares}
  e^{-S_{f,0}(\omega_0)}=\sum_{j\in J} (\Theta B_j)(\omega_0)\,B_j(\omega_0),
  \qquad
  B_j\in\cA_+^{\mathrm{even}}\ \text{supported in }\Sigma_\ell\cap\Lambda_+,
\end{equation}
for some finite index set $J$. (Finiteness is automatic since the Grassmann algebra on $\Sigma_\ell$
is finite-dimensional and the exponential truncates.) Equivalently, it suffices to have a
representation
$S_{f,0}=S_{f,0}^{\mathrm{loc}}-\sum_k c_k\,(\Theta\chi_k)\chi_k$ with $c_k\ge0$ bosonic and $\chi_k$
odd supported in $\Sigma_\ell\cap\Lambda_+$; expanding the exponential then yields
\eqref{eq:ferm-sum-of-squares}.

\begin{theorem}\label{thm:OS-fermions}
Let $S_{b}=S_x+S_{\mathrm{fields}}^{(\mathrm{bos})}$ be the bosonic part of the action and let $S_f$
be a Grassmann-even fermion action. Assume:
\begin{itemize}
\item[(i)] $S_{b}$ and $S_f$ are local on the abstract lattice and have finite interaction range
$\ell<\infty$ in the index-time direction.
\item[(ii)] $S_{b}(\Theta\omega)=S_{b}(\omega)$ and $S_f(\Theta\omega)=S_f(\omega)$, and the hard
constraint $\Xadm$ and the product measures are $\Theta$-invariant (including \eqref{eq:Berezin-inv}).
\item[(iii)] $S_f$ admits a reflection-positive slab form of width $\ell$ in the sense of
\eqref{eq:ferm-slab-split}--\eqref{eq:ferm-sum-of-squares}.
\end{itemize}
Then for all $F\in \cA_+^{\mathrm{even}}$ one has
\[
  \langle\,\Theta F\cdot F\,\rangle\ \ge\ 0,
\]
where $\langle\cdot\rangle$ denotes expectation with respect to the joint finite-volume measure
\[
  1_{\{x\in\Xadm\}}e^{-S_{b}-S_f}\,Dx\,DU\,D\phi\,D\bar\psi\,D\psi.
\]
\end{theorem}

\begin{proof}
As in Theorem~\ref{thm:OS-full}, the finite time-range assumption implies a decomposition
\[
  S_{b}(\omega)=S_{{b},+}(\omega_0,\omega_+)+S_{{b},0}(\omega_0)+S_{{b},-}(\omega_0,\omega_-),
  \qquad
  S_{{b},-}=\Theta S_{{b},+}.
\]
Combine this with \eqref{eq:ferm-slab-split} to obtain for the total action
\[
  S(\omega)=S_+(\omega_0,\omega_+)+S_0(\omega_0)+S_-(\omega_0,\omega_-),
  \qquad
  S_-=\Theta S_+,
\]
where $S_0=S_{{b},0}+S_{f,0}$ depends only on $\omega_0$.

Let $F\in\cA^{\mathrm{even}}_+$, so $F$ depends only on $(\omega_0,\omega_+)$. Define the (graded)
conditional functional
\[
  \Psi_F(\omega_0)=\int e^{-S_+(\omega_0,\omega_+)}\,F(\omega_0,\omega_+)\,d\omega_+,
\]
where $d\omega_+$ includes the bosonic product measures and the Berezin factors over the $\omega_+$
Grassmann generators. Using $\Theta$-invariance of the domain and product measures and the relation
$S_-=\Theta S_+$, one obtains the standard factorization
\[
  \langle \Theta F\cdot F\rangle
  =\frac{1}{Z}\int e^{-S_0(\omega_0)}\,(\Theta\Psi_F)(\omega_0)\,\Psi_F(\omega_0)\,d\omega_0.
\]
Now split $e^{-S_0}=e^{-S_{{b},0}}\cdot e^{-S_{f,0}}$ and use \eqref{eq:ferm-sum-of-squares}:
\[
  e^{-S_{f,0}}=\sum_{j\in J} (\Theta B_j)B_j.
\]
Since $\Psi_F$ is even and $\Theta$ is an involution on the even algebra, the integrand becomes a
finite sum of terms of the form
\[
  e^{-S_{{b},0}(\omega_0)}\;(\Theta(B_j\Psi_F))(\omega_0)\,(B_j\Psi_F)(\omega_0),
\]
each of which integrates to a nonnegative number by the defining property of a $\Theta$--square
pairing on $\cA_+^{\mathrm{even}}$. Summing over $j$ yields $\langle \Theta F\cdot F\rangle\ge 0$.
\end{proof}

\begin{example}[Geometric Wilson/projector mechanism]\label{ex:wilson-type}
A standard sufficient mechanism for the slab-positivity condition \eqref{eq:ferm-sum-of-squares} is
that the couplings crossing the reflection plane factor through orthogonal projectors with
\emph{nonnegative} coefficients. In the present geometric setting, the natural constant projectors
$P_\pm=\tfrac12(1\pm\gamma_1)$ are replaced by the site-wise projectors built from \eqref{eq:Xi-def}:
\[
  P_\pm(n):=\tfrac12\bigl(I\pm \Xi(n)\bigr),
  \qquad
  P_\pm(n)^2=P_\pm(n),\quad P_+(n)P_-(n)=0.
\]

Concretely, let $V(n;x)>0$ be a local scalar weight (e.g.\ a cell volume), let $M(n;x)\in\R$ be an
ultralocal mass weight, and let $w_\mu(n;x)>0$ be a local link weight attached to the oriented link
$(n,\mu)$, all built from local Euclidean invariants of $x$. Consider a Wilson-type prototype action
whose time-like hops are written directly in projector form:
\begin{align}
S_W^\circ[x,U,\bar\psi,\psi]
&= \sum_{n\in\Lambda} V(n;x)\,\bar\psi(n)\,M(n;x)\,\psi(n) \label{eq:wilson-proto-mass}\\
&\quad -\sum_{n\in\Lambda} V(n;x)\,w_{1}(n;x)\,
\bar\psi(n)\,P_-(n)\,U_{1}(n)\,\psi(n+\hat 1) \label{eq:wilson-proto-hop-time}\\
&\quad -\sum_{n\in\Lambda} V(n;x)\,w_{-1}(n;x)\,
\bar\psi(n)\,P_+(n)\,U_{-1}(n)\,\psi(n-\hat 1) \\
&\quad+(\text{terms not crossing the reflection plane}). \nonumber
\end{align}
Here $U_{-1}(n)=U_{1}(n-\hat 1)^{-1}$ as usual. The omitted terms may include arbitrary local spatial
hops and on-slice couplings; they do not affect the slab pairing because they do not connect the two
sides of the reflection plane.

Apply the $\theta$--partner/averaging prescription of \S\ref{sec:full-action} to obtain the averaged
action $S_W=\tfrac12(S_W^\circ+S_W^\theta)$, hence $\Theta S_W=S_W$. In the slab decomposition
\eqref{eq:ferm-slab-split}, the only fermionic terms that genuinely couple the two sides are the
time-like hops across the reflected time link, and because they factor through $P_\pm(n)$ with
nonnegative coefficients, the corresponding slab weight admits a representation
\eqref{eq:ferm-sum-of-squares}. Therefore $S_W$ satisfies hypothesis (iii) of
Theorem~\ref{thm:OS-fermions}, and reflection positivity holds for all $F\in\cA_+^{\mathrm{even}}$.
\end{example}

\subsection{Symanzik effective theory and universality}
\label{subsec:symanzik}

The purpose of this section is to make precise (at the level customary in lattice EFT)
the statement that DLR defines a \emph{local} regulator with \emph{exact} $\SEd$ symmetry,
and that it belongs to the
same universality class as the corresponding theory on a rigid hypercubic lattice, under a mild short-range hypothesis for the geometry sector.

Fix a matter content (scalar, gauge, fermion) and consider the DLR partition function
\begin{equation}
Z(a)=\int_{\Xadm(a)} Dx \int DU \int D\Phi \;
\exp\bigl(-S_x[x]-S_{\rm fields}[x,U,\Phi]\bigr).
\end{equation}
Because the underlying graph is fixed and the admissibility conditions are local,
the DLR action is ultralocal in the abstract lattice coordinates and has a standard
transfer-matrix interpretation once a reflection direction is chosen.
In particular, for any local observable of the enlarged theory supported on a bounded region,
the dependence on distant lattice sites enters only through local propagation.
For the universality questions of interest here we focus on observables depending only on the
gauge/matter fields $(U,\Phi)$; geometry observables probe the regulator sector itself.

\smallskip
Symanzik effective field theory is a convenient way to parameterize \emph{lattice artifacts} in
long--distance observables, cf. \cite{SymanzikI,SymanzikII}.   
Concretely, one considers correlation functions of local operators at a
\emph{fixed physical separation} $|x|$ while taking the cutoff $a\to 0$.  On the lattice this means that the
separation in lattice units grows as $r=|x|/a\to\infty$, so the correlators probe distances that are large
compared to the cutoff but still finite in physical units.

Under the standard assumptions of locality and a mass gap at the cutoff scale, the effect of the regulator
at such distances can be encoded by a continuum effective action
\begin{equation}\label{eq:symanzik}
  S_{\rm eff}
  \;=\;
  S_{\rm cont}
  \;+\;
  \sum_i a^{\Delta_i}\,c_i(g_{\rm bare})\,\mathcal O_i,
\end{equation}
where $\{\mathcal O_i\}$ is a basis of local continuum operators (composite operators built from the fields and
their derivatives) and $\Delta_i>0$ are their excess dimensions relative to the target action.
The coefficients $c_i(g_{\rm bare})$ are \emph{matching coefficients}: they depend on the microscopic definition
of the regulator (hence on the bare couplings $g_{\rm bare}$) and are generally scheme--dependent, but the
operator basis is \emph{not} arbitrary.
Indeed, the allowed operators $\mathcal O_i$ are precisely those compatible with the \emph{exact} symmetries
of the regulator.

In particular, on a rigid hypercubic lattice the exact symmetry is only the hypercubic group $\Hd$, so the Symanzik operator basis is larger than the $\SOd$--scalar basis and radiative corrections can generate \emph{rotation--breaking} higher--derivative terms. 
For a scalar field, writing $H_{\mu\nu}=\partial_\mu\partial_\nu\phi$, the $\SOd$--invariants at four-derivative order
are $\tr(H)^2=(\partial^2\phi)^2$ and $\tr(H^2)=\sum_{\mu,\nu}(\partial_\mu\partial_\nu\phi)^2$.
By contrast, hypercubic symmetry permits independent coefficients for the diagonal and off--diagonal pieces, e.g., $\sum_\mu(\partial_\mu^2\phi)^2$ and $\sum_{\mu<\nu}(\partial_\mu\partial_\nu\phi)^2$.
Reducing $\SOd$ to $\Hd$ enlarges the allowed mixing among operators with the same quantum numbers, permitting additional anisotropic counterterms that are absent for an $\SOd$--covariant regulator.

\medskip
\noindent
The only additional assumption needed to reduce DLR to an ordinary Symanzik problem is:

\begin{hypothesis}[Short-range geometry]\label{hyp:SR}
After quotienting by the global $\SEd$ zero modes, connected correlation functions of local geometric
observables (built from finitely many $e_\mu(n)$, $g_{\mu\nu}(n)$, $V(n)$, etc.) decay exponentially in
lattice distance, uniformly in $a$, i.e.\ the geometry sector has a finite correlation length in
lattice units.
\end{hypothesis}

This hypothesis is a genuine condition on the geometry measure and therefore constrains the admissible choices of $S_x$.
A natural way to target \hyperref[hyp:SR]{(SR)} is to include an explicit local stiffness scale for the embedding fluctuations
(e.g.\ quadratic penalties around a reference metric as in \eqref{eq:spring-penalty}), which makes the non--zero--mode fluctuations
of $x(n)=an+\eta(n)$ effectively massive in lattice units. The hard admissibility constraints are compatible with such a massive regime
and help by excluding degenerate cells, but by themselves they do not guarantee a mass gap or exponential mixing.
This distinction is particularly important in higher dimensions, where entropic instabilities of random geometry ensembles are
well known. Accordingly, we treat \hyperref[hyp:SR]{(SR)} as an explicit assumption to be verified for any proposed $S_x$,
and we include direct diagnostics in Section~\ref{sec:numerics}. 
In $d=2$ we find behaviour consistent with a short-range twisting regime; however the $d=4$ situation remains to be mapped.

Assuming \hyperref[hyp:SR]{(SR)}, one may integrate out the geometry field at fixed $a$ to obtain an effective action
for the field sector $(U,\Phi)$ alone. Locality and (SR) imply that this effective action is local and
admits a Symanzik expansion of the form \eqref{eq:symanzik}. Moreover, because the DLR measure
is $\SEd$-covariant, the induced effective action for $(U,\Phi)$ is
$\SOd$-covariant, and gauge invariant when gauge fields are present. 
In particular, only $\SOd$-scalar operators contribute to scalar observables. This removes the representation-theoretic source of
rotation-breaking counterterms and, more importantly, potentially simplifies Symanzik improvement and operator
renormalisation: the basis of allowed counterterms (and hence the pattern of operator mixing) is
constrained by full $\SOd$ rather than only the hypercubic group.

Crucially, integrating out short-range geometry changes the \emph{relation} between bare and
renormalized parameters. Concretely, averaging the local volume/metric factors over $x$ produces
finite (scheme-dependent) shifts in the coefficients of the operators already present in the target
continuum theory---e.g.\ the mass term, the kinetic normalization, and (when present) the bare quartic
or gauge coupling. This is the usual regulator dependence of the bare-to-renormalized map and simply
means that the critical surface and tuning relations $g_{\rm bare}(a)$ are modified. What short-range
geometry \emph{cannot} do is generate new relevant or marginal operators not already allowed by the
continuum symmetries: that would require additional light degrees of freedom or genuinely long-range
interactions in the geometry sector, which are excluded by (SR).

\begin{remark}
One might wonder whether the benefits of exact $\SOd$ covariance at nonzero cutoff could be obtained
more economically by adding a single global rotation variable $R\in\SOd$ to an ordinary hypercubic
lattice, i.e., by considering an embedding $x_R(n)=aRn$ and averaging over $R$.
Viewed as a ``geometry sector'', this degree of freedom is maximally long-range: local geometric
observables built from $e_\mu$, $g_{\mu\nu}$, $V$ are perfectly correlated across the lattice, so the
short-range hypothesis {(SR)} fails in the strongest possible way.

In standard hypercubic LGT this construction is nevertheless trivial, because the action depends only
on the abstract graph and fixed hypercubic coefficients and is independent of the embedding: the
integral over $R$ merely multiplies the partition function by $\Vol(\SOd)$ and, after ``integrating
out'' $R$, one simply recovers the usual hypercubic theory. In particular, the Symanzik operator
basis remains that of the hypercubic regulator. On a finite periodic torus the situation is even more
restrictive: a generic rotation does not preserve the period lattice and therefore does not define a
symmetry of the discrete torus beyond its discrete hypercubic subgroup.

By contrast, DLR integrates over a \emph{local} embedding field $x(n)$ that enters the action through
local metric/volume factors. After quotienting by (or fixing) the global $\SEd$ zero modes to remove
the infinite group volume, the geometry still fluctuates locally. Assuming {(SR)}, integrating out
$x$ produces a \emph{local} effective action for $(U,\Phi)$ to which the Symanzik expansion applies,
and the underlying $\SEd$ covariance implies that the induced effective action is globally
$\SOd$-covariant.

A purely global ``geometry'' variable (such as a rigid rotation) has infinite correlation length; if it were coupled nontrivially it would act as a mediator producing nonlocal effective interactions upon elimination, which lies outside the scope of the Symanzik expansion and is exactly what (SR) excludes.
\end{remark}

\noindent
To connect the Symanzik discussion to the usual perturbative intuition, expand the embedding about the regular one,
\begin{equation}\label{eq:eta_expand}
x(n)=a\,n+\eta(n),\qquad
e_\mu(n)=x(n+\hat\mu)-x(n)=a\,\hat\mu+\delta e_\mu(n),
\end{equation}
where $\delta e_\mu(n)=\eta(n+\hat\mu)-\eta(n)$ and $\eta(n)\in\R^d$.
Throughout this and the next subsections, $\nabla_\mu$ denotes the forward lattice derivative
\begin{equation}\label{eq:forward_diff}
(\nabla_\mu f)(n)=\frac{f(n+\hat\mu)-f(n)}{a},
\end{equation}
acting componentwise on $\R^d$--valued fields. In particular,
\begin{equation}\label{eq:diff_eta}
(\nabla_\mu\eta)(n)=\frac{\eta(n+\hat\mu)-\eta(n)}{a}=\frac{\delta e_\mu(n)}{a}\in\R^d.
\end{equation}
Heuristically, if $\eta(n)=u(a n)$ for a smooth interpolant $u$, then $\nabla_\mu\eta(n)=\partial_\mu u(a n)+O(a)$.

\smallskip
Fix components with respect to the regular orthonormal frame $\{\hat\mu\}$:
\begin{equation}\label{eq:eta_components}
\eta_\nu(n)=\hat\nu\cdot\eta(n),\qquad
(\nabla\eta)_{\mu\nu}(n)=\nabla_\mu\eta_\nu(n).
\end{equation}
Thus $\nabla\eta$ is the lattice Jacobian, its trace
\begin{equation}\label{eq:div_eta}
\nabla\!\cdot\eta(n)=\tr(\nabla\eta)(n)=\sum_{\mu}\nabla_\mu\eta_\mu(n),
\end{equation}
is the discrete divergence, and we write $|\nabla\eta|^2=\sum_{\mu,\nu}(\nabla_\mu\eta_\nu)^2$ for the canonical
$\SOd$--scalar quadratic in first differences. When writing schematic vertices such as
$(\nabla\eta)(\nabla\phi)(\nabla\phi)$, indices are understood to be contracted in an $\SOd$--covariant way.

On the periodic torus we use the backward derivative $(\nabla_\mu^* f)(n)=(f(n)-f(n-\hat\mu))/a$ and the associated
backward gradient $(\nabla^*F)_\mu=\nabla_\mu^*F$. Then for any scalar $F$, we have
\begin{equation}\label{eq:sbp}
\sum_n (\nabla\!\cdot\eta)(n)\,F(n) \;=\; -\sum_n \eta(n)\cdot (\nabla^*F)(n),
\end{equation}
with no boundary term on the torus.

All local geometric data entering the field actions are smooth local functions of $\nabla\eta(n)$, 
more precisely of the frame matrix $E(n)$ built from
the $e_\mu(n)$. Expanding around the regular embedding therefore yields schematic local expansions
\begin{equation}\label{eq:gV_expand_eta}
g_{\mu\nu}(n)=a^{2}\Big(\delta_{\mu\nu}+O(\nabla\eta)\Big),\qquad
V(n)=a^{d}\Big(1+O(\nabla\eta)+O(|\nabla\eta|^2)\Big),
\end{equation}
with coefficients bounded uniformly on $\Xadm$ by admissibility (uniform shape regularity).

Let us consider the scalar $\phi^4$ theory for definiteness. Writing the scalar action on a general admissible geometry as
$S_\phi[x,\phi]=S_\phi^{(0)}[\phi]+S^{(1)}_{\phi\eta}[\phi,\eta]+S^{(2)}_{\phi\eta}[\phi,\eta]+\cdots$,
we see that the regular-lattice part $S_\phi^{(0)}$ is the usual hypercubic action, while the geometry-dependent pieces
split into two distinct families already at first order, as follows.

\smallskip
\emph{(i) Metric-induced (kinetic) couplings.}
Expanding the factor $V\,g^{\mu\nu}$ multiplying discrete derivatives produces vertices bilinear in $\phi$,
schematically of the form
\begin{equation}\label{eq:kinetic_vertices}
\begin{split}
S^{(1)}_{\phi\eta} &\supset\ \sum_n a^d\;(\nabla\eta)(n)\,(\nabla\phi)(n)\,(\nabla\phi)(n), \\
S^{(2)}_{\phi\eta} &\supset\ \sum_n a^d\;|\nabla\eta(n)|^2\,(\nabla\phi)^2(n),
\end{split}
\end{equation}
and higher orders. These are the ``derivative bilinear'' couplings.

\smallskip
\emph{(ii) Volume-induced (potential) couplings.}
The volume factor $V(n)$ multiplies the mass and quartic densities.  At linear order
its variation is a scalar made from $\nabla\eta$ (indeed a discrete divergence),
so the expansion produces vertices with \emph{no} derivatives on $\phi$:
\begin{equation}\label{eq:potential_vertices}
\begin{split}
S^{(1)}_{\phi\eta} &\supset\ \sum_n a^d (\nabla\cdot\eta)(n)\Big(\tfrac12 m_0^2\phi(n)^2+\tfrac{\lambda_0}{4!}\phi(n)^4\Big), \\
S^{(2)}_{\phi\eta} &\supset\ \sum_n a^d |\nabla\eta(n)|^2\Big(\phi^2+\phi^4\Big) .
\end{split}
\end{equation}

On the periodic torus (or with suitable decay), the linear volume variation is the trace of the lattice Jacobian,
so schematically
\begin{equation}\label{eq:V_linear_div}
V(n)=a^d\Big(1+\nabla\!\cdot\eta(n)+O(|\nabla\eta|^2)\Big).
\end{equation}
Thus the first-order potential contribution includes, in particular, the quartic term
\begin{equation}\label{eq:div_eta_phi4}
\sum_n a^d\,(\nabla\!\cdot\eta)(n)\,\frac{\lambda_0}{4!}\,\phi(n)^4.
\end{equation}
Applying \eqref{eq:sbp} with $F=\phi^4$ gives
\begin{equation}\label{eq:ibp_phi4}
\sum_n (\nabla\!\cdot\eta)(n)\,\phi(n)^4
\;=\;
-\sum_n \eta(n)\cdot \nabla^*\big(\phi^4\big)(n),
\end{equation}
on the torus. Hence even the non-derivative $\phi^4$ term arising from the volume factor may be rewritten as a coupling
of $\eta$ to a lattice derivative of a local density. Equivalently, since $x$ enters the microscopic action only through
differences $e_\mu(n)$, the fluctuation field has the shift symmetry $\eta\mapsto\eta+c$, so every $\eta$ insertion carries
at least one lattice derivative (and therefore a factor of $\hat p$ in momentum space).

Assuming {(SR)}, connected correlators of local geometric observables built from finitely many
$\nabla\eta$, $g_{\mu\nu}$, $V$, etc.\ decay exponentially in lattice distance, uniformly in $a$.
Equivalently, the geometry sector has a finite correlation length in lattice units, hence a mass scale of
order $1/a$ in physical units.  As a result, integrating out $\eta$ at fixed $a$ produces a \emph{local}
effective action for the field sector $(U,\Phi)$, with a derivative expansion in powers of
$pa$ at fixed physical external momenta $p$.
This is precisely the input needed for the Symanzik description: all induced operators are local and,
by $\SEd$ covariance of the regulator, transform covariantly under $\SOd$, and are gauge invariant when
gauge fields are present.  The geometry sector can and will renormalize the \emph{bare-to-renormalized map}
(shifting masses/couplings and wave-function normalizations), and it can generate higher-dimension
$\SOd$-scalar operators (including higher-point operators such as $\phi^6,\phi^8,\dots$), but under (SR)
it cannot produce new relevant/marginal operators beyond those allowed by the target continuum symmetries.

The following statement is therefore not a theorem of the present paper; it records the
universality scenario expected under (SR), and identifies the precise point at which a
future proof of the geometry phase would enter.

\begin{conjecture}\label{conj:universality}
Assume {(SR)} and fix/quotient the global $\SEd$ zero modes in an $\SEd$-covariant manner (so that local
correlators are affected only by finite-volume effects).  Then, for any choice of local field content
$(U,\Phi)$ and any local $\SEd$-covariant discretization of the corresponding continuum action,
the DLR defines a local effective field theory for $(U,\Phi)$ whose long-distance correlators admit a
Symanzik expansion \eqref{eq:symanzik} with the following properties:

\begin{enumerate}
\item[\textup{(i)}]
The induced Symanzik operators $\mathcal O_i$ may be chosen to be gauge invariant (when gauge fields are present)
and to transform covariantly under $\SOd$; in particular, only $\SOd$-scalar operators contribute to scalar
observables.  Consequently, the set of allowed counterterms and the pattern of operator mixing are constrained by
full $\SOd$, rather than only the hypercubic symmetry.

\item[\textup{(ii)}]
Integrating out the short-range geometry modifies the relation between bare and renormalized parameters: the
critical surface and the tuning relations $g_{\rm bare}(a)$ may shift compared with a rigid hypercubic regulator.
However, no new relevant or marginal directions are generated beyond those already allowed by the target continuum
symmetries.

\item[\textup{(iii)}]
After tuning the usual relevant parameters (e.g.\ the mass in scalar theories, and the gauge coupling in gauge
theories), the continuum limit of DLR lies in the same universality class as the corresponding continuum QFT.
Equivalently, all continuum correlation functions of renormalized local observables agree with those obtained from
any other local regulator preserving the same internal symmetries.
\end{enumerate}
\end{conjecture}

\noindent
In particular, for scalar $\phi^4$ theory in $d=4$, the coefficient of the one-loop $\beta$-function is unchanged
by the geometry sector: geometry fluctuations can renormalize the mapping between $\lambda_0$ and $\lambda_R$, but
do not modify the universal one-loop running of $\lambda_R$ (see \S\ref{subsec:phi4-one-loop}).

\medskip
At a conceptual level, Conjecture~\ref{conj:universality} is an application of EFT to the enlarged system
$(x,U,\Phi)$.  Exact $\SEd$ covariance of the microscopic measure implies that any effective action obtained after
integrating out $x$ is $\SOd$-covariant (and gauge invariant when applicable), yielding \textup{(i)}.  
The substantive input is the short-range geometry hypothesis~\hyperref[hyp:SR]{(SR)}:
exponential clustering in the geometry sector gives a finite correlation length in lattice units, so
integrating out $x$ at fixed cutoff produces a local effective action for $(U,\Phi)$ admitting a derivative expansion
in powers of $pa$.  Consequently the geometry sector can only renormalize the matching of operators already present
in the target theory and generate higher-dimension $\SOd$-covariant operators suppressed by powers of $a$, giving
\textup{(ii)}; standard RG universality then yields \textup{(iii)} after tuning the usual relevant parameters.
Conjecture~\ref{conj:universality} is intended as a concrete target for future work; establishing or falsifying it
beyond one loop and in $4d$ gauge theories would be particularly valuable.

\subsection{Example: one--loop $\beta$--function in scalar $\phi^{4}$ theory}
\label{subsec:phi4-one-loop}

We illustrate Conjecture~\ref{conj:universality} on a test theory where the universal
running is visible already at one loop: a real scalar $\phi^4$ model in $d=4$ with local potential
\begin{equation}
  V(\phi)=\tfrac12 m_0^2\phi^2+\tfrac{\lambda_0}{4!}\phi^4,
\end{equation}
discretised on the dynamical mesh as in Section~\ref{sec:fields-symmetries}.
On the regular embedding $x_{\rm reg}(n)=an$ this reduces to the standard nearest--neighbour
lattice discretisation.
Write
\begin{equation}
  x(n)=an+\eta(n),
\end{equation}
and assume the short-range geometry hypothesis (SR) from \S\ref{subsec:symanzik}:
connected correlators of local geometric observables built from finitely many
$e_\mu(n)$, $g_{\mu\nu}(n)$, $V(n)$ decay exponentially in lattice distance, uniformly in $a$.
Equivalently (and this is how it is used perturbatively), the $\eta$--sector has a finite
correlation length in lattice units. In physical units this means the geometry modes are
\emph{heavy} with characteristic mass scale $\sim 1/a$, so integrating them out produces a
local effective lattice theory for $\phi$ whose induced interactions are short-range and analytic
in the external momenta at scales $\mu\ll 1/a$.

As discussed previously, expanding the scalar action $S_\phi[x,\phi]$ in powers of $\nabla\eta$ produces two conceptually
different classes of geometry--matter couplings.
\begin{itemize}
\item \emph{Kinetic (metric/frame) vertices.}
Expanding $V(n)g^{\mu\nu}(n)\nabla_\mu\phi\,\nabla_\nu\phi$ generates
vertices where geometry couples to \emph{derivative bilinears} in $\phi$, schematically
\begin{equation}
  (\nabla\eta)\,(\nabla\phi)(\nabla\phi),\qquad
  |\nabla\eta|^2\,(\nabla\phi)(\nabla\phi),\ \ldots
\end{equation}
These are the vertices that contribute to wave function renormalisation and to higher-derivative
Symanzik operators.

\item \emph{Volume vertices.}
Expanding the local volume factor $V(n)$ in
\begin{equation}
  \sum_n V(n)\Bigl(\tfrac12 m_0^2\phi(n)^2+\tfrac{\lambda_0}{4!}\phi(n)^4\Bigr)
\end{equation}
produces vertices with no derivatives on the scalar legs (though every $\eta$ insertion carries at least one lattice derivative),
e.g.
\begin{equation}
  (\nabla\!\cdot\eta)\,\phi^2,\qquad
  (\nabla\!\cdot\eta)\,\phi^4,\qquad
  |\nabla\eta|^2\,\phi^2,\qquad
  |\nabla\eta|^2\,\phi^4,\ \ldots
\end{equation}
What matters for universality is that under (SR) these vertices describe couplings to a heavy,
short-range field.
\end{itemize}

\noindent
The one loop renormalisation of the quartic coupling is determined by the coefficient of the
logarithmic dependence on the renormalisation scale $\mu$ in the amputated 1PI four--point function.
This coefficient is universal: it arises from the loop--momentum window
\emph{above} $\mu$ but \emph{below} the cutoff (heuristically $\mu\ll |k|\ll 1/a$), and is
insensitive to finite, scheme--dependent redefinitions of bare parameters.

\smallskip
\noindent\emph{(A) Pure scalar bubble.}
The one--loop 1PI four--point graphs built only from the usual $\phi^4$ vertex are exactly those of the static
hypercubic lattice (the $s/t/u$ bubbles), cf.\ Fig.~\ref{f:pure-phi}.
In each channel there are \emph{two} internal scalar propagators running in the loop, hence the characteristic
integrand $(\hat k^2+m_0^2)^{-2}$. In particular, the standard bubble integral yields
\begin{equation}
  I_{\rm lat}(a,m_0)
  = \int_{B}\frac{d^4k}{(2\pi)^4}\,\frac{1}{(\hat k^2+m_0^2)^2}
  = \frac{1}{16\pi^2}\log\frac{1}{a^2m_0^2}+\text{finite},
\end{equation}
and summing the three channels produces the familiar one--loop coefficient $3/(16\pi^2)$ in the $\phi^4$ channel.
Here $B=[-\pi/a,\pi/a]^4$ is the Brillouin zone and
\begin{equation}
\hat k_\mu = \frac{2}{a}\sin\!\Big(\frac{a k_\mu}{2}\Big),
\qquad
\hat k^{2} = \sum_{\mu=1}^{4}\hat k_\mu^{2},
\end{equation}
so that the free lattice propagator is $(\hat k^{2}+m_0^{2})^{-1}$.

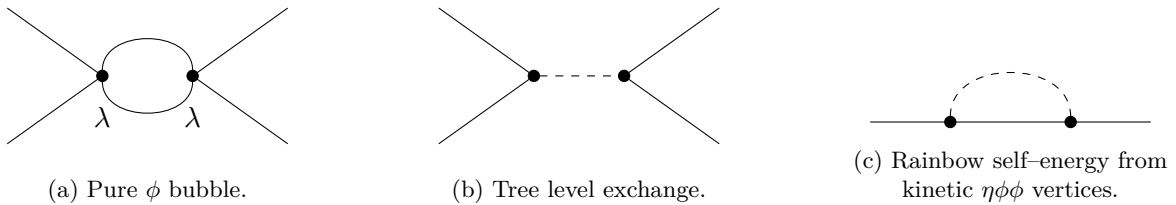
\begin{figure}[ht]
\centering
\begin{subfigure}{0.3\textwidth}\centering
\begin{tikzpicture}
\begin{feynman}
  \vertex (a1) at (-2,  1) {};
  \vertex (a2) at (-2, -1) {};
  \vertex (b1) at ( 2,  1) {};
  \vertex (b2) at ( 2, -1) {};
  \vertex [dot] (vL) at (-0.6, 0) {};
  \vertex [dot] (vR) at ( 0.6, 0) {};
  \diagram*{
    (a1) -- [plain] (vL),
    (a2) -- [plain] (vL),
    (vR) -- [plain] (b1),
    (vR) -- [plain] (b2),
    (vL) -- [plain, half left, looseness=1.2] (vR),
    (vL) -- [plain, half right, looseness=1.2] (vR),
  };
  \node at (-0.6, -0.55) {\(\lambda\)};
  \node at ( 0.6, -0.55) {\(\lambda\)};
\end{feynman}
\end{tikzpicture}
\caption{Pure \(\phi\) bubble.}\label{f:pure-phi}
\end{subfigure}\hfill
\begin{subfigure}{0.3\textwidth}\centering
\begin{tikzpicture}
\begin{feynman}
  \vertex (a1) at (-2,  1) {};
  \vertex (a2) at (-2, -1) {};
  \vertex (b1) at ( 2,  1) {};
  \vertex (b2) at ( 2, -1) {};
  \vertex [dot] (vL) at (-0.6, 0) {};
  \vertex [dot] (vR) at ( 0.6, 0) {};
  \diagram*{
    (a1) -- [plain] (vL),
    (a2) -- [plain] (vL),
    (vR) -- [plain] (b1),
    (vR) -- [plain] (b2),
    (vL) -- [dashed] (vR),
  };
\end{feynman}
\end{tikzpicture}
\caption{Tree level exchange.}\label{f:kinetic-4}
\end{subfigure}\hfill
\begin{subfigure}{0.3\textwidth}\centering
\begin{tikzpicture}
\begin{feynman}
  \vertex (i) at (-2,0) {};
  \vertex (o) at ( 2,0) {};
  \vertex [dot] (v1) at (-0.8,0) {};
  \vertex [dot] (v2) at ( 0.8,0) {};

  \diagram*{
    (i) -- [plain] (v1) -- [plain] (v2) -- [plain] (o),
    (v1) -- [dashed, half left, looseness=1.25] (v2),
  };
\end{feynman}
\end{tikzpicture}
\caption{Rainbow self--energy from kinetic $\eta\phi\phi$ vertices.}\label{f:rainbow}
\end{subfigure}
\caption{Building blocks in the perturbative universality argument. Solid: $\phi$. Dashed: $\eta$.}
\label{f:1}
\end{figure}

\smallskip
\noindent\emph{(B) Diagrams with internal geometry lines.}
At one loop, any 1PI four--point contribution that involves at least one internal $\eta$ propagator is controlled by (SR).
Indeed, (SR) implies the $\eta$ correlator decays exponentially in lattice units, so the geometry is massive in lattice units.
Consequently, at external scales $\mu\ll 1/a$, graphs with internal $\eta$ lines generate \emph{local} contributions
analytic in the external momenta:
\begin{itemize}
\item \emph{Kinetic vertices.}
Kinetic vertices necessarily carry external momenta: already at tree level, two $\eta\phi\phi$
vertices connected by a single internal $\eta$ line produce a momentum--dependent 4--point
amplitude (Fig.~\ref{f:kinetic-4}), whose low--energy expansion generates higher--derivative
operators in the Symanzik action, not a momentum--independent shift of the $\phi^4$ coupling.
The corresponding 2--point correction is the rainbow self--energy in Fig.~\ref{f:rainbow}.

\item \emph{Volume vertices.}
These can and do renormalise the \emph{bare--to--renormalised map}: integrating out short--range volume fluctuations shifts
the coefficients of operators already present in the target theory (mass term, kinetic normalisation, and the quartic coupling)
by finite, scheme--dependent amounts (Fig.~\ref{f:xi-loop}). In addition, exchanging a short--range $\eta$ line between two
local volume insertions produces higher multi--field operators; for example Fig.~\ref{f:xi-exchange} corresponds to an induced
$\phi^8$ interaction, which is irrelevant in $d=4$ and does not affect the running of $\lambda$.
Neither mechanism can generate a new logarithmic $\mu$--dependence in the $\phi^4$ channel at one loop.
\end{itemize}

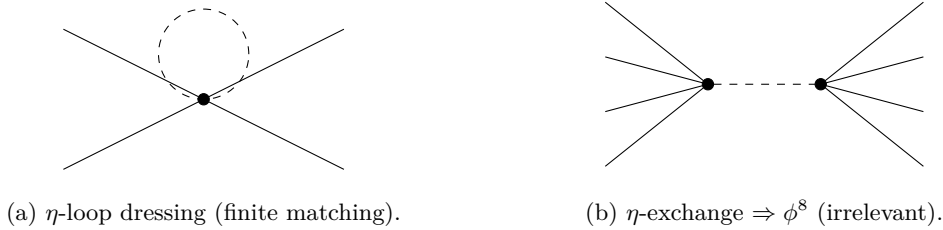
\begin{figure}[t]
\centering
\begin{subfigure}{0.4\textwidth}\centering
\begin{tikzpicture}
\begin{feynman}
  \vertex (a1) at (-2,  1) {};
  \vertex (a2) at (-2, -1) {};
  \vertex (b1) at ( 2,  1) {};
  \vertex (b2) at ( 2, -1) {};
  \vertex [dot] (v) at (0,0) {};

  \diagram*{
    (a1) -- [plain] (v),
    (a2) -- [plain] (v),
    (v)  -- [plain] (b1),
    (v)  -- [plain] (b2),
  };

  \draw[dashed] ($(v)+(0,0.6)$) circle [radius=0.6];
\end{feynman}
\end{tikzpicture}
\caption{\(\eta\)-loop dressing (finite matching).}\label{f:xi-loop}
\end{subfigure}
\qquad
\begin{subfigure}{0.4\textwidth}\centering
\begin{tikzpicture}
\begin{feynman}
  \vertex [dot] (vL) at (-1.5,0) {};
  \vertex (a1) at (-3,  1.2) {};
  \vertex (a2) at (-3,  0.4) {};
  \vertex (a3) at (-3, -0.4) {};
  \vertex (a4) at (-3, -1.2) {};

  \vertex [dot] (vR) at ( 0,0) {};
  \vertex (b1) at ( 1.5,  1.2) {};
  \vertex (b2) at ( 1.5,  0.4) {};
  \vertex (b3) at ( 1.5, -0.4) {};
  \vertex (b4) at ( 1.5, -1.2) {};

  \diagram*{
    (a1) -- [plain] (vL),
    (a2) -- [plain] (vL),
    (a3) -- [plain] (vL),
    (a4) -- [plain] (vL),

    (b1) -- [plain] (vR),
    (b2) -- [plain] (vR),
    (b3) -- [plain] (vR),
    (b4) -- [plain] (vR),

    (vL) -- [dashed] (vR),
  };
\end{feynman}
\end{tikzpicture}
\caption{\(\eta\)-exchange \(\Rightarrow \phi^8\) (irrelevant).}\label{f:xi-exchange}
\end{subfigure}
\caption{Further one--loop graphs with internal geometry lines (line conventions as in Fig.~\ref{f:1}).}
\end{figure}

Assume (SR), so the $\eta$--propagator has a mass scale of order one in lattice units.
At external scales $\mu\ll 1/a$, any 1PI graph with at least one internal $\eta$ line produces a local
contribution analytic in the external momenta. Consequently, kinetic vertices contribute only to
wave--function renormalisation and higher--derivative operators, while volume vertices can shift the local
$\phi^4$ coefficient by a finite (scheme--dependent) amount and can induce irrelevant multi--field operators
(e.g.\ Fig.~\ref{f:xi-exchange}), but cannot generate an additional $\log\mu$ term in the $\phi^4$ channel at one loop.
Therefore the coefficient of the logarithmic divergence in the $\phi^4$ channel is entirely due to the pure scalar bubbles
and equals $3/(16\pi^2)$, as stated in \eqref{eq:beta-cont}; after the usual field renormalisation the renormalised quartic
coupling $\lambda(\mu)$ obeys the same one--loop renormalisation--group equation as in the continuum and in the static--lattice
regularisation,
\begin{equation}\label{eq:beta-cont}
  \mu\,\frac{d\lambda}{d\mu} \;=\; \beta(\lambda) \;=\;
  \frac{3}{16\pi^{2}}\lambda^{2} + O(\lambda^{3}),
\end{equation}
while the dynamical geometry affects only finite matching and higher--dimension operators in the Symanzik expansion. This
supports Conjecture~\ref{conj:universality}.

\section{Numerical experiments}
\label{sec:numerics}

This section describes the Monte Carlo setup and numerical diagnostics used to test the
dynamical--lattice regulator (DLR) in a two--dimensional scalar theory.  The goal is
proof--of--concept rather than high--precision spectroscopy or a systematic exploration of
the geometry phase diagram.  We check (i) algorithmic viability under standard local
updates, (ii) that the geometry sector remains massive near matter criticality, as required
by the short--range geometry hypothesis~\hyperref[hyp:SR]{(SR)}, (iii) suppression of
direction--dependent cutoff artefacts at finite lattice spacing, and (iv) preservation of
long--distance universality in a standard finite--size scaling test.

Accordingly, the numerical role of \hyperref[hyp:SR]{(SR)} in this section is diagnostic:
even when the scalar field is tuned close to criticality, the geometry should not develop
long-range modes.  The simulations below test this requirement in a simple \(d=2\) scalar
setting, together with the local admissibility framework and the Metropolis update scheme.
They should not be read as evidence that the same simple angle-spring geometry action is
sufficient to realize \hyperref[hyp:SR]{(SR)} in \(d=4\), where a separate scan of the
geometry phase diagram would be required.

Throughout, \emph{baseline} refers to the fixed regular square lattice, while \emph{dynamical}
refers to the coupled geometry--scalar system in which the embedding
\(x:\Lambda\to\mathbb{R}^2\) is updated alongside the matter field.

\subsection{Simulation protocol}
\label{subsec:protocol}

We work on a periodic $L\times L$ abstract lattice
$\Lambda=\{0,\dots,L-1\}^2$ with lattice spacing $a=1$ unless stated
otherwise.  The dynamical degrees of freedom are an embedding
$x(n)\in\mathbb{R}^2$ and a real scalar field $\phi(n)\in\mathbb{R}$ at
each $n\in\Lambda$.  The joint Boltzmann weight is
\[
  Z^{-1}\exp\bigl(-S_x[x] - S_\phi[x,\phi]\bigr),
\qquad
  Dx\,D\phi=\prod_{n\in\Lambda}\mathrm{d}^2x(n)\,
  \mathrm{d}\phi(n),
\]
with $x$ restricted to the admissible set $\Xadm$ of \S\ref{sec:local-geometry}.  
The geometry action $S_x$ is the local, $E(2)$--invariant action of \S\ref{sec:geometry-measure}, while the scalar
action $S_\phi$ is taken to be the \emph{prototype} action $S_\phi=S_{\phi}^\circ$ of \S\ref{sec:full-action} (i.e.\ before $\theta$--averaging), evaluated on the same embedded mesh, with the local potential 
\begin{equation}
W_{\mathrm{loc}}(\phi)=\frac12\mu^2\phi^2 + \lambda\phi^4 .
\end{equation}
Note that the theory is parametrized by the bare couplings $(\mu^2,\lambda)$ in lattice units.
Our simulation convention for the quartic term is $\lambda\,\phi^4$ (no $1/4!$ factor).
Thus, relative to the continuum-style normalization
$W_{\mathrm{loc}}(\phi)=\tfrac12 m_0^2\phi^2+\tfrac{\lambda_0}{4!}\phi^4$ used before,
we have
\[
m_0^2=\mu^2,\qquad \lambda_0 = 24\,\lambda.
\]
We set \(\lambda=1\) and vary \(\mu^2\) in the finite-size scaling tests and in the
near-critical correlation-length diagnostics.  The illustrative mesh and one-point
geometry histograms use the auxiliary parameter choice \(\lambda=0.1\), \(\mu^2=1\).

\medskip
\noindent
\emph{Periodic geometry.}
The physical domain is a flat torus with period vectors
$P_0=(La,0)$ and $P_1=(0,La)$.  We store $x(n)$ in a fixed fundamental
domain $[0,La)\times[0,La)$ and use a minimal--image convention for
displacements: whenever a neighbour crosses the boundary we shift by
integer multiples of $P_0,P_1$ to obtain the shortest representative.
All geometric quantities are computed cellwise from the four vertices of a
plaquette, using this periodic reconstruction.

This convention fixes the fundamental cycles of the physical torus in finite volume.  All diagnostics below
are phrased in terms of local twisting/mixing and connected correlators, hence are insensitive to the
particular representative chosen for the global $\mathrm{SE}(2)$ zero modes.

\medskip
\noindent
\emph{Local admissibility.}
Each cell is required to be convex and shape--regular.  Concretely, for the
numerical runs we impose the fixed local bounds
\[
  A \ge A_{\min}=0.10,\qquad
  \ell_\mu \le \ell_{\max}=4.0,\qquad
  \mathrm{aspect}\le R_{\max}=4.0,\qquad
  30^\circ\le\vartheta\le150^\circ,
\]
where $\ell_\mu$ are the two edge lengths, $A$ is the cell area,
$\vartheta$ is the corner angle between the two edges, and
$\mathrm{aspect}=\max(\ell_1,\ell_2)/\min(\ell_1,\ell_2)$.
The specific numerical values are not tuned: they are conservative bounds chosen to exclude
near--degenerate cells while allowing broad fluctuations away from the regular embedding.
Moves that violate any bound on any affected cell are rejected
immediately.  In practice these hard constraints prevent near--degenerate
cells and keep the simulation in the same connected component as the
regular embedding.

\medskip
\noindent
\emph{Updates.}
Both geometry and scalar field sweeps use local Metropolis proposals at uniformly chosen sites.
\begin{itemize}
\item \emph{Geometry move:} pick $n\in\Lambda$ and propose
  $x'(n)=x(n)+\varepsilon_x$ with $\varepsilon_x\sim\mathrm{Unif}
  ([-\Delta_x,\Delta_x]^2)$, wrapped back into the fundamental domain.  The
  move changes only the four cells adjacent to $n$; we re-evaluate their
  admissibility and their local contributions to $S_x$ and $S_\phi$, and
  accept with probability
  $\min\{1,\exp(-\Delta S_x-\Delta S_\phi)\}$.
\item \emph{Scalar move:} pick $n\in\Lambda$ and propose
  $\phi'(n)=\phi(n)+\varepsilon_\phi$ with
  $\varepsilon_\phi\sim\mathrm{Unif}[-\Delta_\phi,\Delta_\phi]$.
  Only the four adjacent cells change; we accept with probability
  $\min\{1,\exp(-\Delta S_\phi)\}$.
\end{itemize}
We distinguish a \emph{geometry sweep} and a \emph{scalar sweep}, each consisting of $L^2$
single--site proposals of the corresponding type.  One Monte Carlo \emph{cycle} comprises one
geometry sweep followed by one scalar sweep.
Proposal scales are tuned to give $O(1)$ acceptance rates; in the runs shown
the geometry acceptance is typically in the few $\times 10\%$ to $\sim 80\%$ range,
depending on $(L,\mu^2)$.

\medskip
\noindent
\emph{Geometry action parameters.}
For the numerical experiments we use a quadratic ``spring'' that penalizes only corner-angle distortions:
\[
  S_x[x]=\sum_{n\in\Lambda} \frac{k}{2}(\vartheta(n)-\vartheta_{\rm ref})^2.
  \qquad (k=2.5) .
\]
Edge lengths, areas, and aspect ratios are controlled primarily by the hard admissibility constraints.

This minimal choice is designed to keep the geometry disordered at the level of local orientations
while the hard constraints control degeneracies; it is \emph{not} intended as an optimized $S_x$ for $d=4$.

\medskip
\noindent
\emph{Run lengths.}
For each parameter point we thermalise for $N_{\rm therm}$ sweeps (typically $5\times10^3$),
then record $N_{\rm prod}$ sweeps (typically $2\times10^4$), storing measurements every $k$ sweeps
(typically $k=10$).  Each experiment is repeated over a few dozen independent random seeds, and
the same protocol is used uniformly across the $(L,\mu^2)$ scan underlying the finite--size scaling plots.

\subsection{Measured observables}
\label{subsec:observables}

We record both ``health'' diagnostics to monitor stability and mixing, and
physics observables to compare baseline vs dynamical and to probe
universality and rotational artefacts.

\medskip
\noindent
\emph{Basic diagnostics.}
On each measurement we record $S_x$, $S_\phi$, the minimum/maximum cell
area, and the maximum aspect ratio over the lattice.  For the scalar we
record the moments $\langle\phi^2\rangle=L^{-2}\sum_n\phi(n)^2$ and
$\langle\phi^4\rangle=L^{-2}\sum_n\phi(n)^4$.

\medskip
\noindent
\emph{Geometry short-range diagnostics.}
A key requirement is that the geometry acts as a regulator rather than an
additional critical field (``SR hypothesis'').  To test this we measure a connected
volume--volume correlator built from the cell areas $A(n)$.
On each stored configuration we form $\delta A(n)=A(n)-\bar A$ with
$\bar A=|\Lambda|^{-1}\sum_m A(m)$, and define the translation--averaged correlator
\[
  C_V(m) \;=\; \frac{1}{|\Lambda|}\sum_{n\in\Lambda}\delta A(n)\,\delta A(n+m),
\]
where addition is modulo $L$ in each direction.  The reported $C_V(m)$ is then
the Monte--Carlo average of this quantity.
We radially bin by the periodic graph ($\ell^1$) distance 
\[
r(m)=\min(i,L-i)+\min(j,L-j)
\] 
for $m=(i,j)$ and report
\[
  C_V(r) \;=\; \frac{1}{N_r}\sum_{m:\,r(m)=r} C_V(m),
  \qquad N_r=\#\{m:\,r(m)=r\}.
\]
In plots we show $|C_V(r)|$ on a log scale and extract an effective decay
length from an exponential fit over short to intermediate $r$.

\medskip
\noindent
\emph{Short--distance rotational probe.}
We measure an angle--binned finite--difference gradient observable:
for a fixed list of short offsets $\Delta$ (nearest neighbours, diagonals,
and a few second neighbours), we compute the minimal--image displacement
$v$, its length $\ell=|v|$, and accumulate $(\Delta\phi/\ell)^2$ into
$N_\theta=18$ bins in $[0,\pi)$ by the orientation of $v$ (with $v\sim-v$).
From the resulting profile $H(\theta)=\langle(\Delta\phi/\ell)^2\rangle(\theta)$
we form the anisotropy measures
\[
  A_{\rm rms}=
  \frac{\sqrt{\langle(H(\theta)-\bar H)^2\rangle_\theta}}{\bar H},
  \qquad
  A_4=\frac{\sqrt{C_4^2+S_4^2}}{\bar H},
\]
with $\bar H=\langle H(\theta)\rangle_\theta$,
$C_4=\langle H(\theta)\cos4\theta\rangle_\theta$,
$S_4=\langle H(\theta)\sin4\theta\rangle_\theta$.

\medskip
\noindent
\emph{Two--point functions and directional splitting.}
We also measure a radial two--point function $G(r)$ by choosing
$N_{\rm src}$ random sources per measurement and correlating with all sites
using minimal--image physical distances $r=|x(n)-x(m)|$, binned into $N_r$
radial bins.  To sharpen rotational artefact detection, we additionally
measure direction--restricted correlators $G_{\rm axis}(r)$ and
$G_{\rm diag}(r)$ from the same source--sink pairs by selecting
displacements within $\pm15^\circ$ of the coordinate axes or diagonals.
We plot both the correlators and their difference
$\Delta G(r)=G_{\rm diag}(r)-G_{\rm axis}(r)$, which isolates the residual
axial/diagonal splitting.

\medskip
\noindent
\emph{Universality observables.}
To locate the critical point and check the universality class we compute:
(i) the Binder cumulant
$U_4=1-\langle M^4\rangle/(3\langle M^2\rangle^2)$ with
$M=\sum_n\phi(n)$, where $\langle\cdot\rangle$ denotes the Monte Carlo average at fixed bare parameters;
and
(ii) the second--moment correlation length $\xi$ from standard structure
factor ratios, plotted as $\xi/L$.
(Implementation details of the $\xi$ estimator are standard and follow the
usual periodic--torus definitions.)

\subsection{Geometry sector diagnostics}
\label{subsec:geom_diagnostics}

Figure~\ref{fig:typ-mesh} shows a representative embedded mesh from the
dynamical ensemble.  The geometry is visibly fluctuating, but the local cells
remain well inside the admissible region.

We first examine one-point diagnostics of the mesh geometry.  Figure~\ref{fig:hist_angle}
shows the distribution of corner angles \(\theta\), Fig.~\ref{fig:hist_edge} the
edge lengths \(\ell\), Fig.~\ref{fig:hist_area} the cell areas \(A\),
Fig.~\ref{fig:hist_aspect} the aspect ratios \(r\), and Fig.~\ref{fig:hist_orient}
the physical edge orientations.  In lattice units we find, for example,
\[
  \langle \ell \rangle = 1.10 \pm 0.48, \qquad
  \langle \theta \rangle = (86.7 \pm 27.1)^\circ, \qquad
  \langle r \rangle \approx 1.78 .
\]
The distributions are broad, so the mesh is not frozen near the square embedding,
but they show no visible pile-up at the admissibility cutoffs.  The orientation
histogram retains peaks near the coordinate axes, indicating residual microscopic
memory of the underlying hypercubic stencil.  This does not by itself rule out a
useful twisting regime: what matters for the short-range geometry hypothesis is
whether local geometric observables decorrelate on \(O(1)\) lattice scales.

We therefore measure connected geometry correlators.  As scalar diagnostics we use
(i) a displacement correlator based on \(u(n)=x(n)-an\), with the configurationwise
mean subtracted to remove the translation zero mode, and (ii) the connected
volume--volume correlator \(C_V(r)\) defined in \S\ref{subsec:observables}, binned
by graph distance.  Both decay rapidly and are consistent with short-range geometry
fluctuations; see Fig.~\ref{fig:dis_corr}--\ref{fig:vol_corr}.

For \(C_V(r)\) we plot \(|C_V(r)|\) on a log scale and extract an effective decay
length \(\ell_{\rm geom}\) from exponential fits over short to intermediate
separations.  These fitted values should be interpreted as effective short-distance
decay lengths, not as finite-size scaling estimates of a precisely defined mass gap.
Using the combined dataset and a uniform fit prescription over the plotted
short-distance range, we find no systematic growth of \(\ell_{\rm geom}\) with
\(L\).  For example, at \(\mu^2=-1.40\),
\[
  \ell_{\rm geom}\approx 1.57,\ 1.07,\ 0.56,\ 0.54
  \qquad (L=16,24,32,48),
\]
while at the nearby value \(\mu^2=-1.36\),
\[
  \ell_{\rm geom}\approx 1.39,\ 1.19,\ 0.95,\ 0.53
  \qquad (L=16,24,32,48).
\]
The individual numbers vary with the fitting window and threshold, as expected
for short-distance fits of noisy connected correlators, but the robust conclusion
is that the geometry correlations remain \(O(1)\) in lattice units over the
volumes tested.

Finally, to test local frame twisting directly, we measure the \(SO(2)\)-connection
introduced in Section~\ref{sec:Ed-invariance}.  Writing the polar factor of the
local frame as \(Q(n)=R_{\alpha(n)}\), define
\[
  R_\mu(n)=Q(n)^{-1}Q(n+\hat\mu),\qquad
  \operatorname{tr}R_\mu(n)
  =2\cos\bigl(\alpha(n+\hat\mu)-\alpha(n)\bigr).
\]
Figure~\ref{fig:twist} shows the connected correlator \(C_R(r)\) of
\(\operatorname{tr}R_\mu\), normalized by \(C_R(0)\), at the representative benchmark point
\(L=48\), \(\mu^2=-1.36\).  The correlator drops below \(10\%\) of its
zero-distance value by graph distance \(r=2\) and is consistent with the noise
floor thereafter.  Thus the orientation peaks in Fig.~\ref{fig:hist_orient} should
not be interpreted as evidence for a rigid or slowly varying global orientation.
They indicate residual local memory of the hypercubic stencil, while the frame
connection itself decorrelates within a few lattice spacings.

\subsection{Thermalisation and stability}
\label{subsec:stability}

Figures~\ref{fig:ts_geom_min_area} and \ref{fig:ts_pair} show time-series
diagnostics for representative baseline and dynamical runs.
After the thermalisation window the actions and low moments fluctuate around
stationary means with no visible long--term drift.
Figure~\ref{fig:ts_geom_min_area} records the evolution of the
minimum cell area and the maximum aspect ratio over the lattice; both remain
stable throughout the production window, with no signs of slow drift or
metastability, indicating that the geometry explores the admissible region
rather than getting trapped near a boundary.

\subsection{Rotational artefacts: short and intermediate scales}
\label{subsec:rot_aniso}

\medskip
\noindent
\emph{Angle--binned gradients.}
On the fixed square lattice $H(\theta)$ exhibits a clear fourfold pattern,
while on the dynamical ensemble the modulation is strongly reduced
(Fig.~\ref{fig:anisotropy_profile}).  Table~\ref{tab:anisotropy-summary}
summarises this via $A_{\rm rms}$ and $A_4$.
Quantitatively, $A_{\rm rms}$ drops from $0.42$ to $0.12$ (a factor $\approx 3.5$), while $A_4$ drops from
$0.17$ to $0.08$ (a factor $\approx 2.2$), indicating substantial suppression of the dominant fourfold
harmonic in this short--distance diagnostic.

\begin{table}[ht]
  \centering
  \begin{tabular}{lcc}
    \hline
    Regulator & $A_{\mathrm{rms}}$ & $A_4$ \\
    \hline
    Baseline  & $0.4221(1)$ & $0.1729(1)$ \\
    Dynamical & $0.1207(3)$ & $0.0787(2)$ \\
    \hline
  \end{tabular}
  \caption{Anisotropy measures from the angle--binned gradient observable ($\mu^2=-1.34$, $L=48$).}
  \label{tab:anisotropy-summary}
\end{table}

\medskip
\noindent
\emph{Directional two--point functions.}
As a longer--distance probe we compare axis--like and diagonal--like
two--point functions.  The difference
$\Delta G(r)=G_{\rm diag}(r)-G_{\rm axis}(r)$ is a compact diagnostic of
residual directional splitting.  Figure~\ref{fig:corr_directional_diff}
shows that the baseline regulator exhibits a modest but systematic
non--zero signal at short to intermediate distances, whereas in the
dynamical ensemble $\Delta G(r)$ is consistent with zero over the range
where $G(r)$ is still clearly above the noise floor.

For completeness we also show the underlying correlators
$G_{\rm axis}(r)$ and $G_{\rm diag}(r)$ in Fig.~\ref{fig:corr_directional};
the difference plot above contains the main rotational--symmetry
information in a more sensitive form.

\subsection{Universality and critical scaling}
\label{subsec:universality}

Finally, we check that coupling the scalar theory to the DLR geometry does not change its long--distance
universality class.  We simulate $L\in\{16,24,32,48\}$ and scan the bare mass parameter $\mu^2$ at fixed
quartic coupling (here $\lambda=1$) using the protocol of Section~\ref{subsec:protocol}.
Figure~\ref{fig:universality} shows the Binder cumulant and $\xi/L$ as functions of $\mu^2$ for the dynamical
ensemble.  Both observables exhibit a common crossing region near $\mu_c^2\simeq -1.40$, as expected for a
second--order transition.  The Binder value at the crossing is consistent with the 2D Ising universality
class; since the baseline square--lattice $\phi^4$ theory is known to lie in the same class, this supports
the conclusion that the coupled geometry does not modify the long--distance universality in this
two--dimensional test.

We estimate integrated autocorrelation times for $M^2$ and observe the expected critical slowing down with
$L$; at $\mu^2\approx -1.40$ we find $\tau_{M^2}\approx (1.4,\,4.0,\,5.2,\,8.5)\times 10^2$ sweeps for
$L=16,24,32,48$.  Accordingly, error bars are obtained with blocking/binning using block sizes
$\gg \tau_{M^2}$ and by averaging over multiple independent seeds.  As a basic check on the geometry sector,
we also monitored autocorrelations of simple global geometry summaries (e.g.\ $\sum_n A(n)$ and
$\sum_n S_x(n)$) and found no anomalously slow mode in the runs shown.

\clearpage
\begin{landscape}
\begin{figure}[p]
\centering
\setlength{\tabcolsep}{10pt}
\renewcommand{\arraystretch}{1.0}

\begin{tabular}{@{}p{0.3\linewidth}p{0.3\linewidth}p{0.3\linewidth}@{}}

\subcaptionbox{Typical embedded mesh.\label{fig:typ-mesh}}{%
  \includegraphics[width=\linewidth,height=0.27\textheight,keepaspectratio]{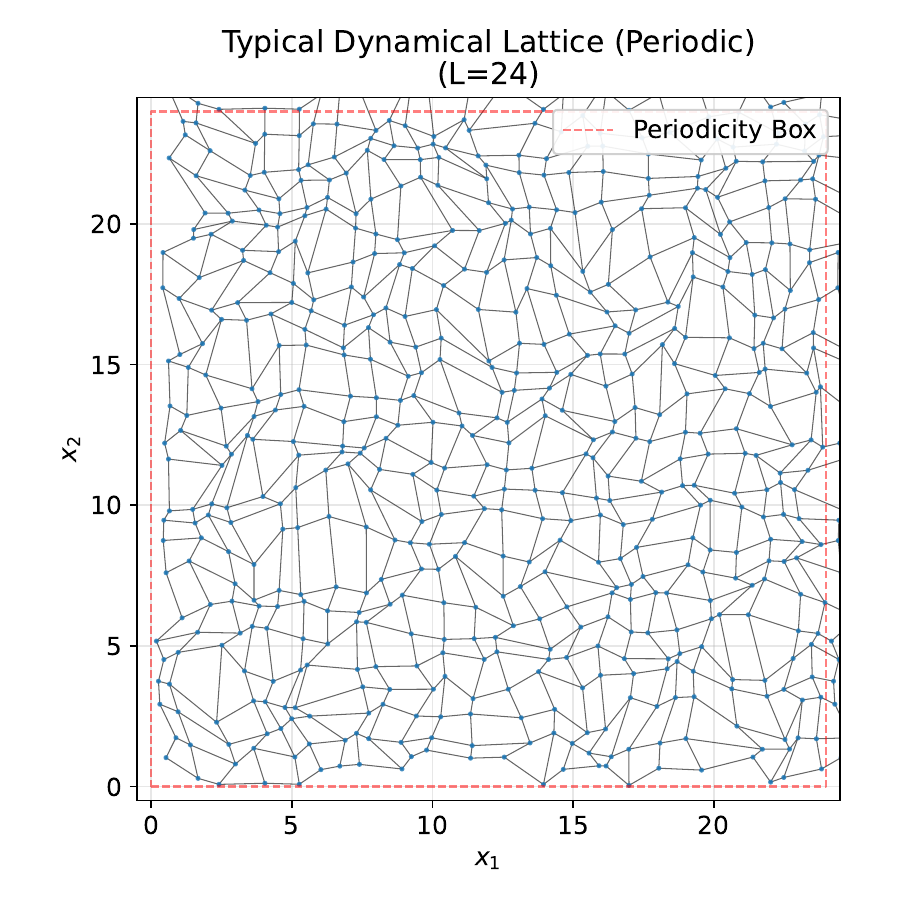}} &
\subcaptionbox{Corner angle distribution.\label{fig:hist_angle}}{%
  \includegraphics[width=\linewidth,height=0.27\textheight,keepaspectratio]{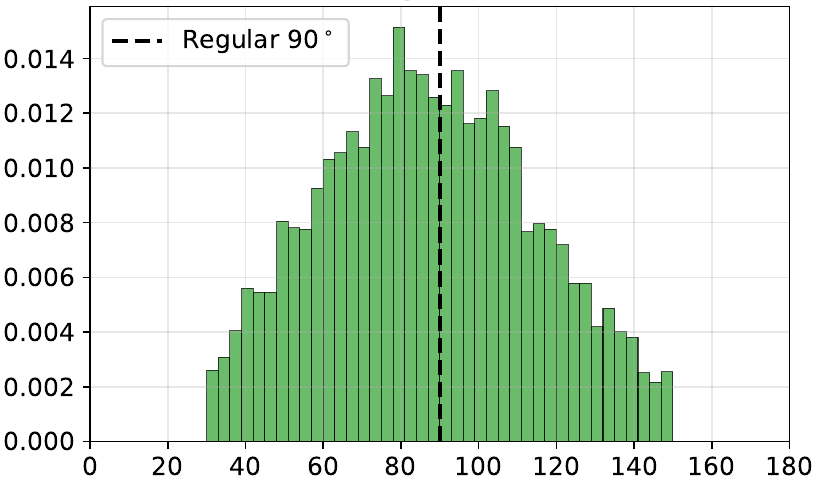}} &
\subcaptionbox{Edge-length distribution.\label{fig:hist_edge}}{%
  \includegraphics[width=\linewidth,height=0.27\textheight,keepaspectratio]{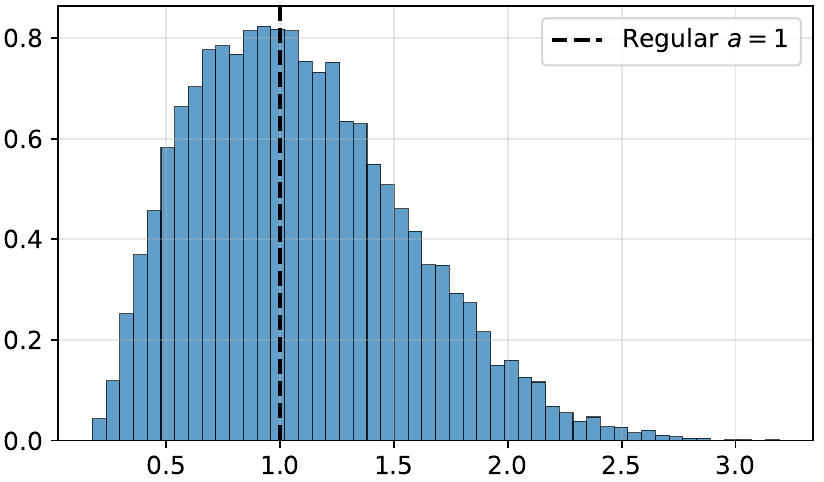}} \\[8mm]

\subcaptionbox{Cell-area distribution.\label{fig:hist_area}}{%
  \includegraphics[width=\linewidth,height=0.27\textheight,keepaspectratio]{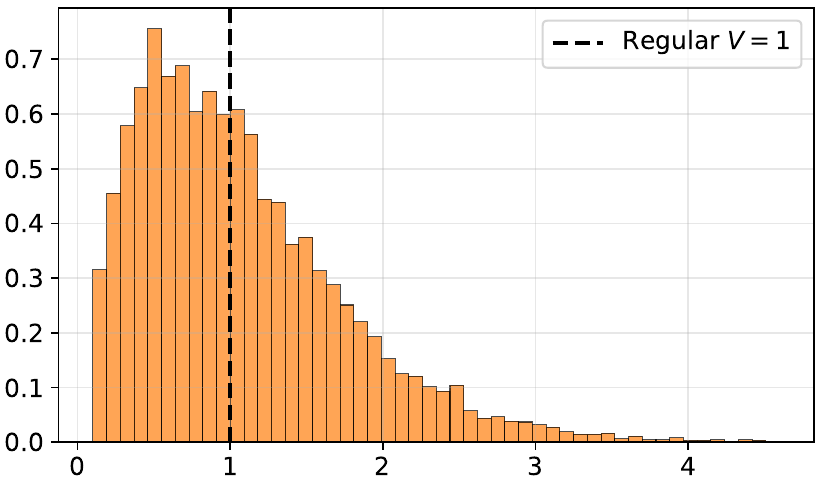}} &
\subcaptionbox{Aspect-ratio distribution.\label{fig:hist_aspect}}{%
  \includegraphics[width=\linewidth,height=0.27\textheight,keepaspectratio]{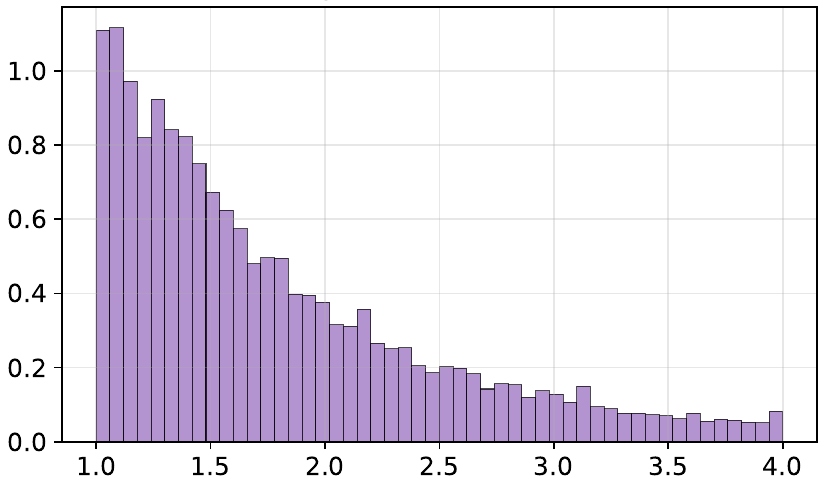}} &
\subcaptionbox{Edge orientations.\label{fig:hist_orient}}{%
  \includegraphics[width=\linewidth,height=0.27\textheight,keepaspectratio]{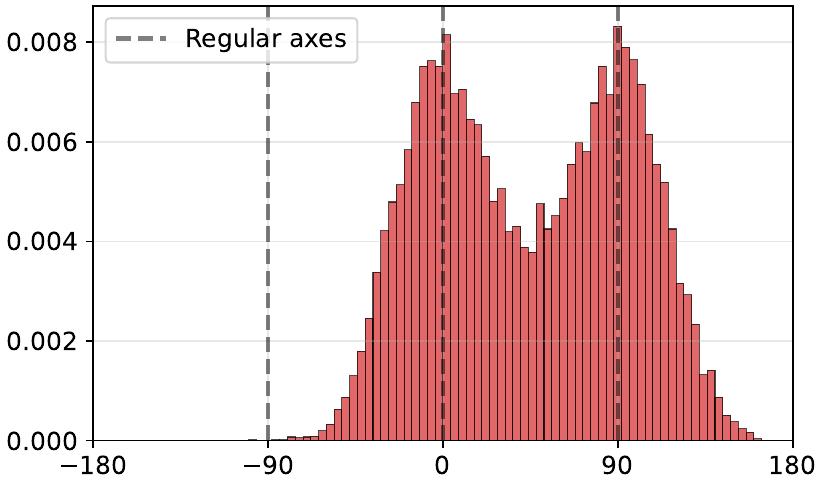}} \\[8mm]

\subcaptionbox{Displacement correlator.\label{fig:dis_corr}}{%
  \includegraphics[width=\linewidth,height=0.27\textheight,keepaspectratio]{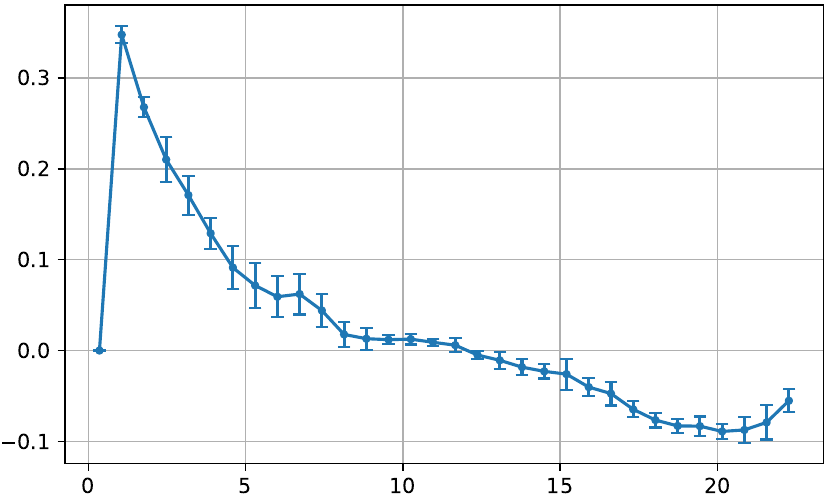}} &
\subcaptionbox{Volume correlator.\label{fig:vol_corr}}{%
  \includegraphics[width=\linewidth,height=0.27\textheight,keepaspectratio]{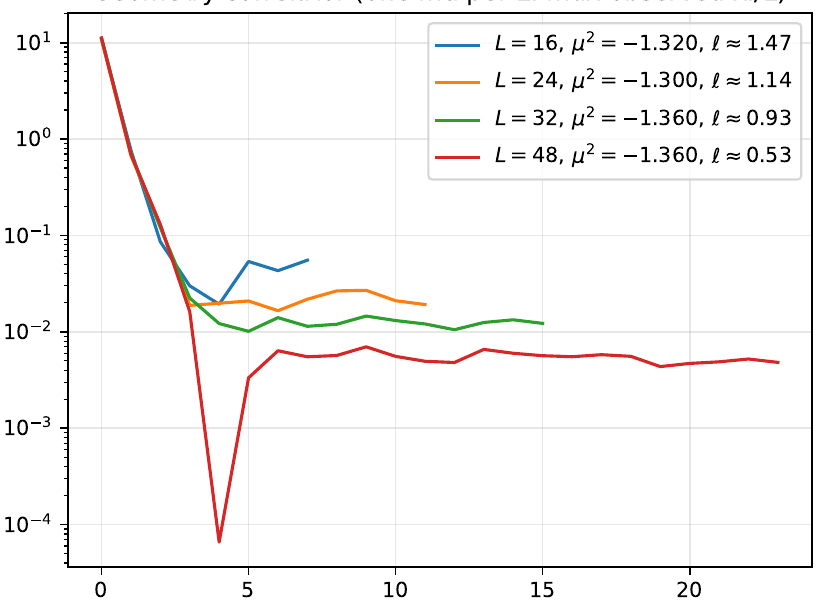}} &
\subcaptionbox{Twisting diagnostic.\label{fig:twist}}{%
  \includegraphics[width=\linewidth,height=0.27\textheight,keepaspectratio]{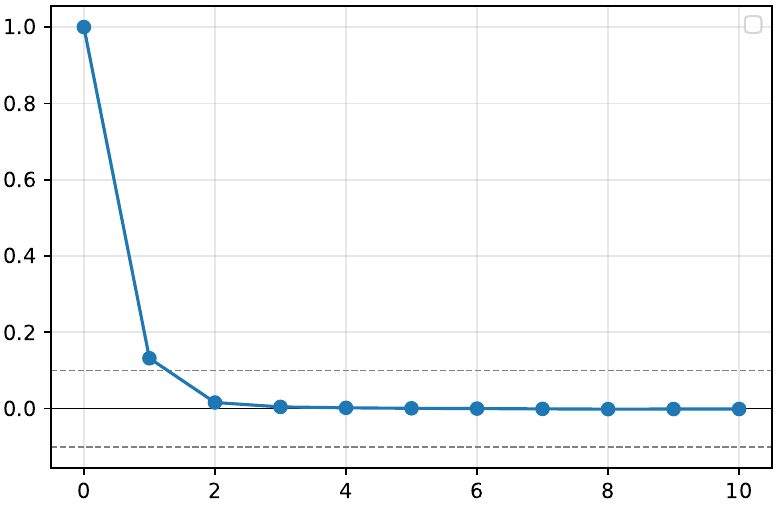}}

\end{tabular}

\caption{Geometry diagnostics for the dynamical ensemble (all panels at identical scale).}
\label{fig:plate_geometry}
\end{figure}
\end{landscape}
\clearpage

\begin{landscape}
\begin{figure}[p]
\centering
\setlength{\tabcolsep}{10pt}
\renewcommand{\arraystretch}{1.0}

\begin{tabular}{@{}p{0.35\linewidth}p{0.64\linewidth}@{}}

\subcaptionbox{Geometry stability: min area / max aspect.\label{fig:ts_geom_min_area}}{%
  \includegraphics[width=\linewidth,height=0.27\textheight,keepaspectratio]{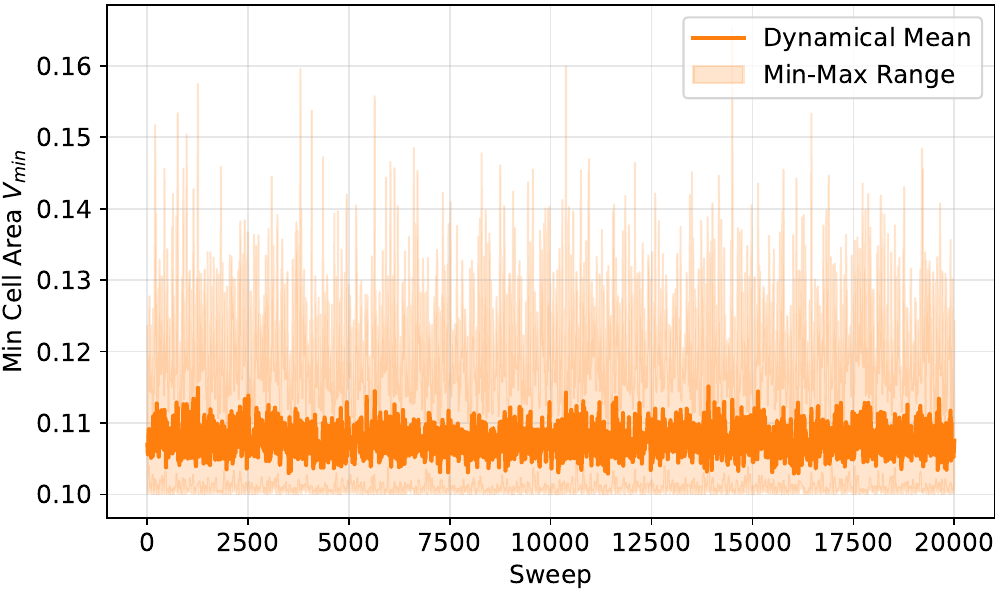}} &
\subcaptionbox{Time series: actions and scalar moments.\label{fig:ts_pair}}{%
  \includegraphics[width=0.49\linewidth,height=0.27\textheight,keepaspectratio]{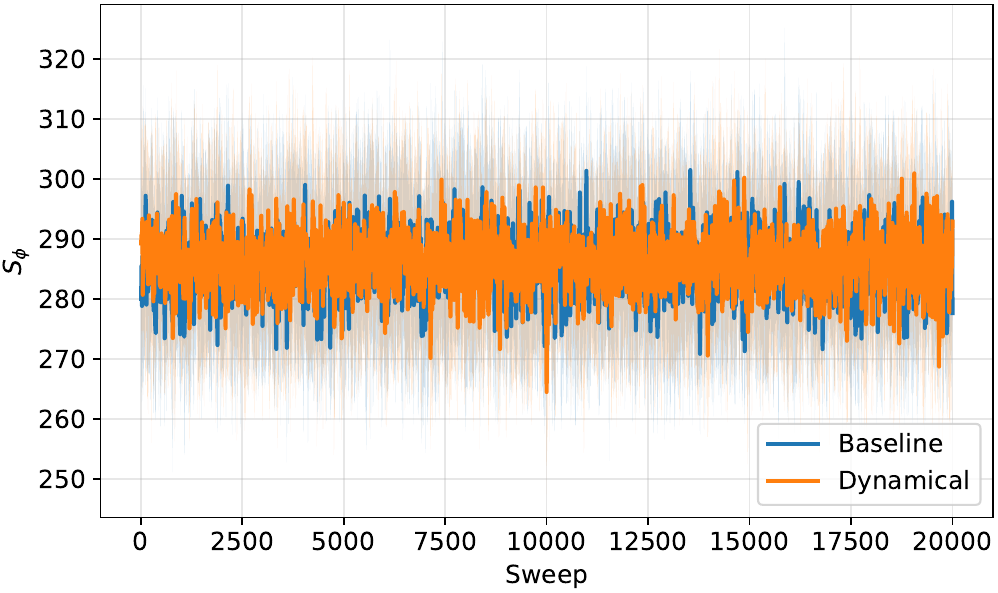}\hfill
  \includegraphics[width=0.49\linewidth,height=0.27\textheight,keepaspectratio]{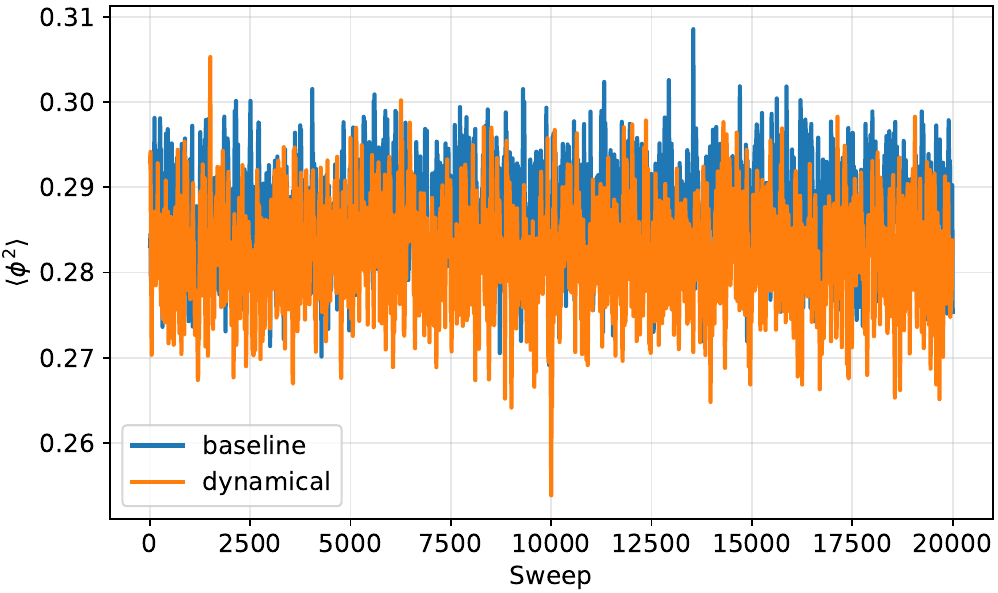}} \\[8mm]

\subcaptionbox{Directional splitting $\Delta G(r)$.\label{fig:corr_directional_diff}}{%
  \includegraphics[width=\linewidth,height=0.27\textheight,keepaspectratio]{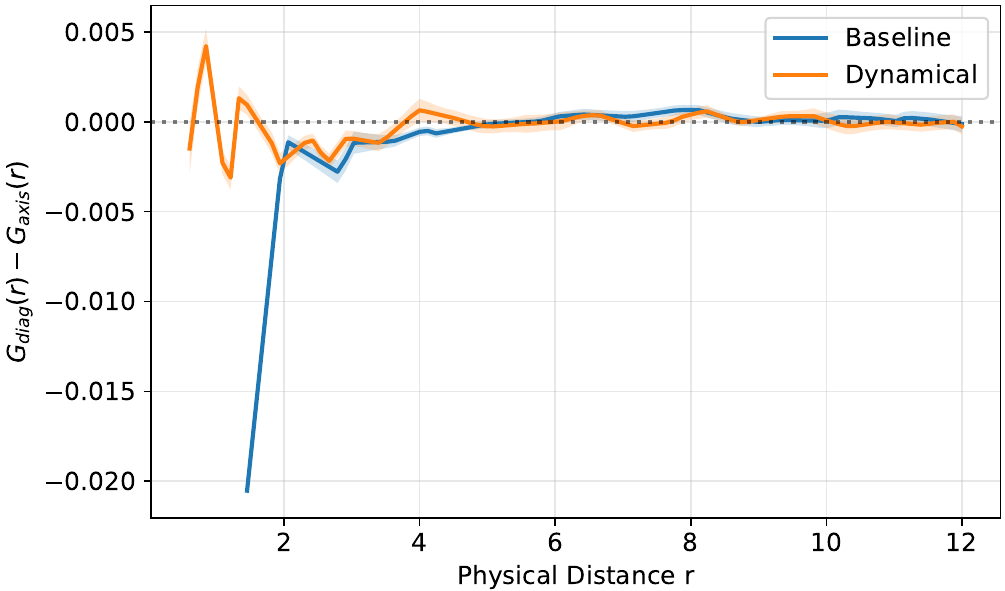}}  &
\subcaptionbox{Directional correlators $G_{\rm axis},G_{\rm diag}$: Baseline (left) and dynamical (right).\label{fig:corr_directional}}{%
  \includegraphics[width=\linewidth,height=0.27\textheight,keepaspectratio]{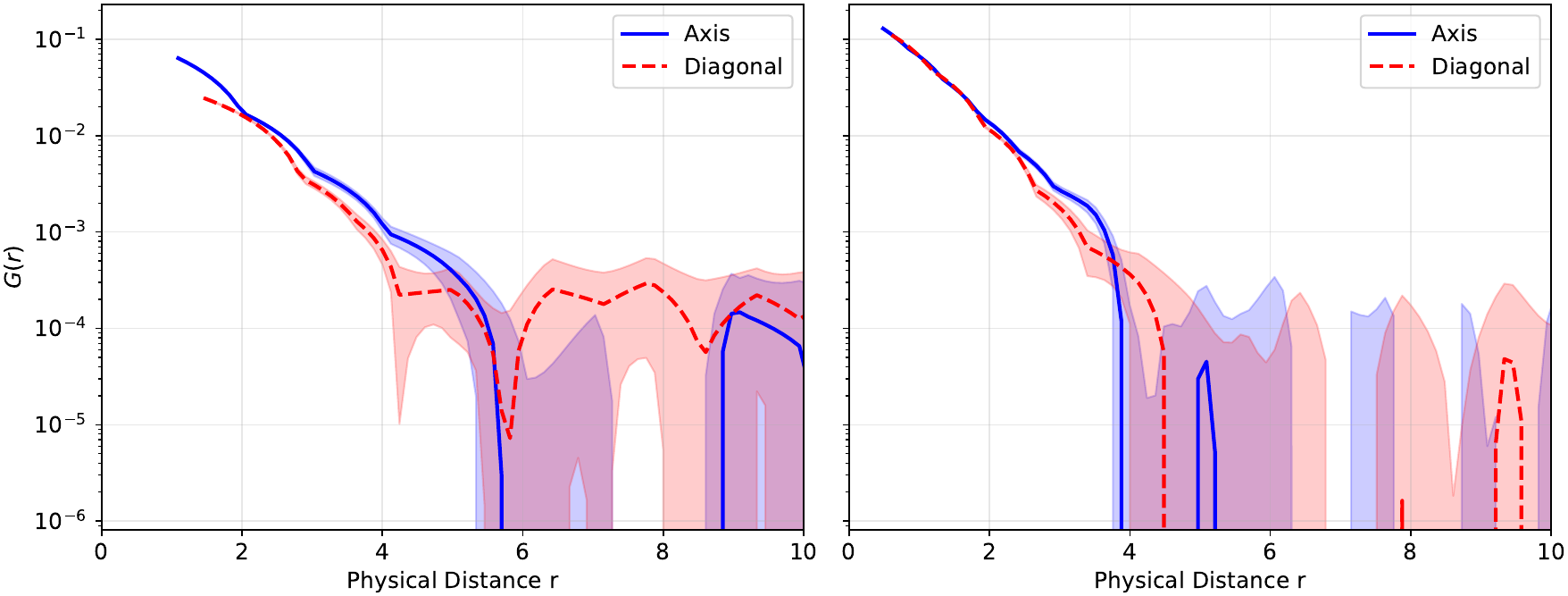}} \\[8mm]

\subcaptionbox{Angle-binned gradient profile $H(\theta)$.\label{fig:anisotropy_profile}}{%
  \includegraphics[width=\linewidth,height=0.27\textheight,keepaspectratio]{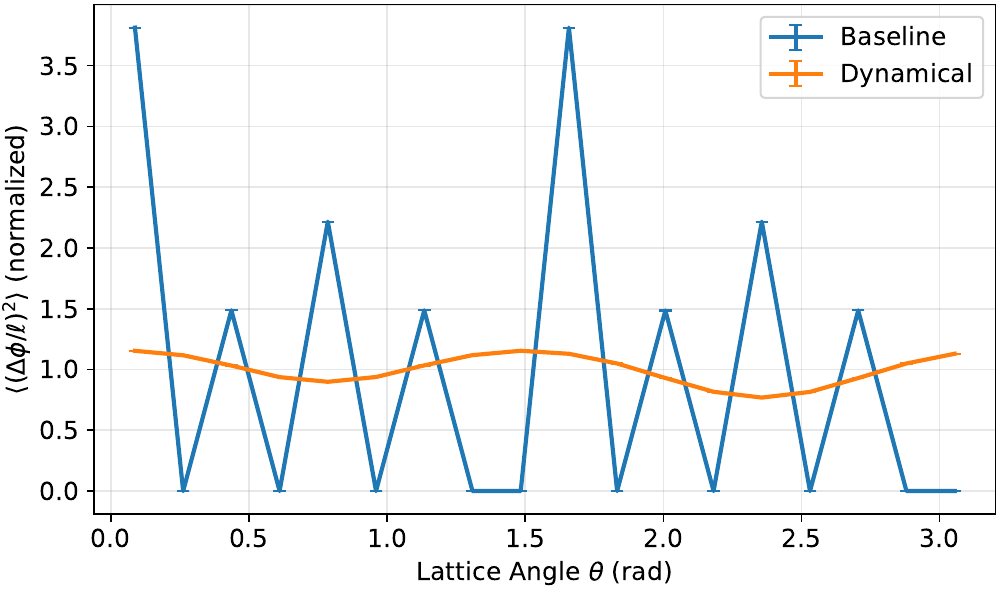}} &
\subcaptionbox{Universality checks: $U_4$ and $\xi/L$.\label{fig:universality}}{%
  \includegraphics[width=0.49\linewidth,height=0.27\textheight,keepaspectratio]{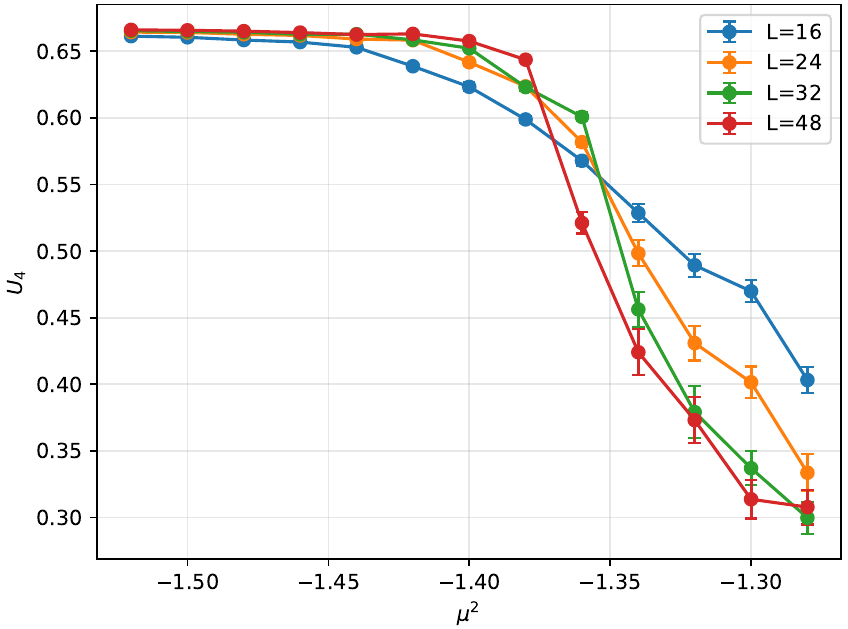}\hfill
  \includegraphics[width=0.49\linewidth,height=0.27\textheight,keepaspectratio]{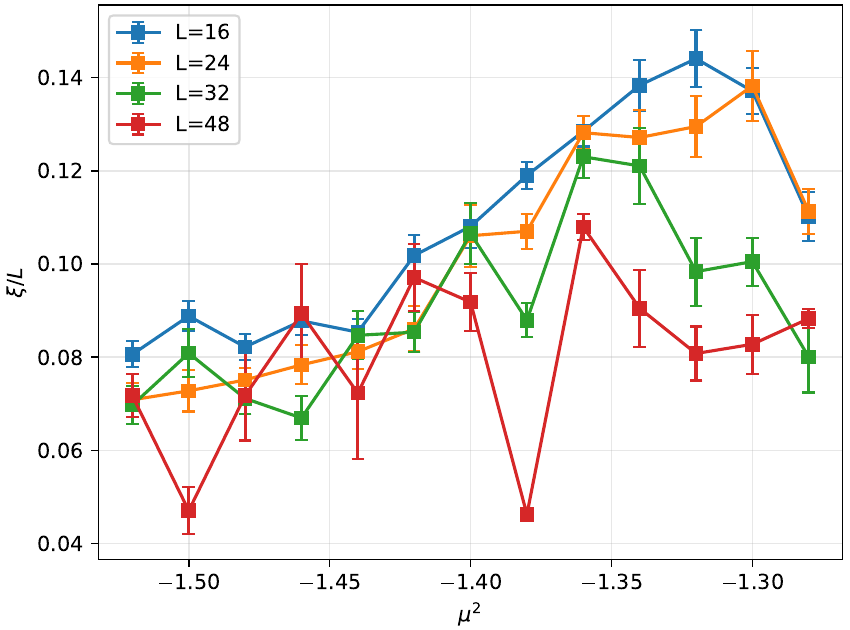}} \\

\end{tabular}

\caption{Matter, rotational-symmetry, and stability diagnostics comparing baseline and dynamical ensembles.}
\label{fig:plate_matter}
\end{figure}
\end{landscape}
\clearpage

\section{Conclusions and outlook}
\label{sec:conclusions}

In this paper we established that the dynamical--lattice regulator (DLR) provides
a local, gauge-invariant cutoff that is exactly \(\SEd\)-covariant at finite lattice
spacing, and admits Osterwalder--Schrader reflection positivity (hence a standard
transfer-matrix reconstruction) for a broad class of coupled geometry--field actions.
Assuming the short-range geometry hypothesis~\hyperref[hyp:SR]{(SR)}, integrating
out \(x\) at fixed \(a\) yields a local Symanzik effective action in which geometry
fluctuations generate only \(\SOd\)-invariant irrelevant operators and finite
renormalisations, a mechanism checked explicitly at one loop in scalar \(\phi^4\).
Our \(d=2\) Monte Carlo tests provide proof-of-concept evidence for stable
\(O(1)\)-scale geometry correlations near matter criticality, short-range decorrelation
of the local frame connection, reduced rotational artefacts at finite cutoff, and
matching critical scaling with the baseline theory.
These numerical tests do not establish the existence of a useful short-range geometry
phase in \(d=4\); that remains a separate geometry-phase-diagram problem.

From the lattice--QCD viewpoint, DLR may be regarded as an \emph{annealed geometric averaging in
coordinate space} that preserves the transfer--matrix/reflection--positivity framework, complementary to
field--space smoothing techniques (smearing, gradient flow) that improve operators/observables but do
not on their own define a new reflection--positive cutoff.  Structurally, DLR also occupies a
conservative middle ground between rigid hypercubic regulators and fluctuating--topology approaches
such as random lattices and dynamical triangulations: the abstract connectivity is fixed (so reflection
positivity and standard lattice algorithms remain accessible), while the embedding fluctuates with a
local action \emph{inside} the same path integral as the gauge and matter fields.

\subsection{Future directions}

There are several natural directions in which the present work can be extended.
\begin{itemize}
  \item \emph{Systematic numerical studies.}
  The exploratory scalar simulations in Section~\ref{sec:numerics} should be developed into a systematic
  study of rotational symmetry restoration, autocorrelations and finite--size effects, and extended to
  gauge theories in $d=3,4$.  In the gauge sector this should include standard benchmarks such as
  plaquette distributions, gradient--flow scales, the static quark potential and string tension
  (including the L\"uscher term), as well as diagnostics of topological objects and topological--charge
  autocorrelations (topology freezing) and their sensitivity to grid artefacts.  A quantitative
  comparison of anisotropy diagnostics and scaling violations between static and dynamical meshes would
  clarify the practical gains of the regulator.

  \item \emph{Geometry phase structure in $4d$ and the short--range hypothesis.}
  The viability of DLR as a regulator hinges on an operating window in which the geometry sector remains
	short-range~\hyperref[hyp:SR]{(SR)} after fixing the global $\SEd$ zero modes: connectivity is fixed and nondegenerate,
	yet local orientations remain disordered. While the
  $d=2$ scalar experiments provide evidence consistent with such a regime even for a minimal geometry
  action, in $4d$ this is not guaranteed: competing entropic regimes (e.g.\ crystallisation/orientational
  ordering or, at the opposite extreme, crumpling--type instabilities familiar from random geometry
  models) may appear unless suitable stiffness terms and admissibility thresholds are chosen.  A
  priority for $4d$ applications is therefore to map the geometry phase
  diagram as a function of these parameters, using geometry correlators and autocorrelation times as
  diagnostics.

  \item \emph{Fermions and chiral symmetry.}
  In this first work we have not addressed the question of fermion doubling in detail.  Since the
  dynamical mesh restores exact $\SEd$ symmetry but does not change the local Dirac operator on the
  abstract lattice, one expects the usual no--go theorems to continue to apply.  Nonetheless, it would
  be worthwhile to investigate systematically how standard fermion formulations (Wilson, staggered,
  overlap) behave in the presence of a fluctuating geometry, and to quantify whether averaging over the
  mesh reduces rotational artefacts such as taste splitting in staggered fermions, effectively
  isotropising the doubler spectrum even if their number is preserved.

  \item \emph{Quenched geometry and ``quantum adaptivity''.}
  Besides the annealed formulation studied here, it would be interesting to consider a \emph{quenched}
  variant in which the abstract hypercubic connectivity is fixed and the embedding $x(n)=a\,n+\eta(n)$ is
  sampled independently from an $\SEd$--invariant short--range law, after which gauge and matter fields
  are simulated conditionally on $x$ and observables are averaged over the geometry ensemble.  This
  quenched perspective removes geometric backreaction and may therefore be more amenable to analytic
  control.  Conversely, in the annealed theory the fields can bias the geometry measure; it is an
  interesting possibility that this backreaction acts as a kind of \emph{quantum adaptivity}, with the
  auxiliary embedding degrees of freedom tending to adjust locally in response to field configurations,
  potentially reducing certain cutoff artefacts.

  \item \emph{Hamiltonian formulation.}
  An important direction is to develop a Kogut--Susskind--type Hamiltonian version of DLR.  Given
  Osterwalder--Schrader reflection positivity and the transfer matrix interpretation of the Euclidean
  theory, one expects a canonical description in which geometry variables appear as additional degrees
  of freedom on time slices and modify the electric and magnetic terms only via local
  geometry--dependent couplings.

  \item \emph{Other topologies and simplicial variants.}
  While we have focused on a hypercubic abstract lattice for constructive and algorithmic reasons, a
  simplicial version of DLR would be natural and would connect more directly to Regge calculus.  One
  could also consider other topologies (e.g.\ manifolds with boundaries, or nontrivial spatial topology)
  as long as a suitable notion of admissible embeddings and OS reflection exists.

  \item \emph{Coupling to dynamical gravity.}
  Finally, one may contemplate turning the auxiliary geometry sector into a genuine gravitational degree
  of freedom by adding a Regge--type curvature term to $S_x[x]$ and allowing the geometry to fluctuate on
  large scales.  This would blur the distinction between ``regulator geometry'' and ``physical
  geometry'', and might provide a technically conservative route towards combining lattice gauge theory
  with discrete approaches to quantum gravity.  Such questions are, however, well beyond the scope of
  the present paper.
\end{itemize}

We view the present work as a starting point, and we invite further numerical and analytical tests of
DLR, in particular in $d=4$ where establishing (SR) and benchmarking rotational artefacts against
standard improvement and flow--based techniques are decisive.

\section*{Acknowledgements}

I am grateful to the Mongolian physicists who took part in an online QFT--GR seminar several years ago.
Our discussions, aimed at revisiting the foundations of quantum field theory and general relativity from complementary viewpoints, encouraged me to think more deeply about quantum field theory and helped shape my understanding of the material underlying this work.
This work was supported by NSERC Discovery Grants Program.


\begin{thebibliography}{99}

\bibitem{Wilson74}
K.~G.~Wilson,
``Confinement of quarks,''
Phys.\ Rev.\ D {\bf 10} (1974), 2445--2459.

\bibitem{Kogut79}
J.~B.~Kogut,
``An introduction to lattice gauge theory and spin systems,''
Rev.\ Mod.\ Phys.\ {\bf 51} (1979), 659--713.

\bibitem{RotheBook}
H.~J.~Rothe,
\emph{Lattice Gauge Theories: An Introduction}, 4th ed.,
World Scientific, Singapore, 2012.

\bibitem{ChristFriedbergLeeRandom}
N.~H.~Christ, R.~Friedberg, and T.~D.~Lee,
``Random lattice field theory: General formulation,''
Nucl.\ Phys.\ B {\bf 202} (1982), 89--125.

\bibitem{ChristFriedbergLeeWeights}
N.~H.~Christ, R.~Friedberg, and T.~D.~Lee,
``Weights of links and plaquettes in a random lattice,''
Nucl.\ Phys.\ B {\bf 210} (1982), 337--346.

\bibitem{ColangeloScrimieriPseudo}
P.~Colangelo and E.~Scrimieri,
``Gauge theories on a pseudorandom lattice,''
Phys.\ Rev.\ D {\bf 35} (1987), 3193.

\bibitem{ColangeloCosmaiScrimieriPseudo}
P.~Colangelo, L.~Cosmai, and E.~Scrimieri,
``Computer simulation of gauge theories on a lattice with improved rotational symmetry,''
Comput.\ Phys.\ Commun.\ {\bf 54} (1989), 235--237.

\bibitem{Regge1961}
T.~Regge,
``General relativity without coordinates,''
Nuovo Cimento {\bf 19} (1961), 558--571.

\bibitem{HamberReview}
H.~W.~Hamber,
``Quantum gravity on the lattice,''
Gen.\ Relativ.\ Gravit.\ {\bf 41} (2009), 817--876.

\bibitem{AmbjornDurhuusJonssonBook}
J.~Ambj{\o}rn, B.~Durhuus, and T.~Jonsson,
\emph{Quantum Geometry: A Statistical Field Theory Approach},
Cambridge University Press, Cambridge, 1997.

\bibitem{AmbjornJurkiewiczLoll}
J.~Ambj{\o}rn, J.~Jurkiewicz, and R.~Loll,
``Dynamically triangulating Lorentzian quantum gravity,''
Nucl.\ Phys.\ B {\bf 610} (2001), 347--382.

\bibitem{AlbaneseAPE}
M.~Albanese \emph{et al.} (APE Collaboration),
``Glueball masses and string tension in lattice QCD,''
Phys.\ Lett.\ B {\bf 192} (1987), 163--169.

\bibitem{MorningstarPeardonStout}
C.~Morningstar and M.~Peardon,
``Analytic smearing of SU(3) link variables in lattice QCD,''
Phys.\ Rev.\ D {\bf 69} (2004), 054501.

\bibitem{HasenfratzKnechtliHYP}
A.~Hasenfratz and F.~Knechtli,
``Flavor symmetry and the static potential with hypercubic blocking,''
Phys.\ Rev.\ D {\bf 64} (2001), 034504.

\bibitem{LuscherFlowJHEP}
M.~L{\"u}scher,
``Properties and uses of the Wilson flow in lattice QCD,''
JHEP {\bf 08} (2010), 071.

\bibitem{SymanzikI}
K.~Symanzik,
``Continuum limit and improved action in lattice theories: (I). Principles and $\phi^4$ theory,''
Nucl.\ Phys.\ B {\bf 226} (1983), 187--204.

\bibitem{SymanzikII}
K.~Symanzik,
``Continuum limit and improved action in lattice theories: (II). O($N$) nonlinear sigma model in perturbation theory,''
Nucl.\ Phys.\ B {\bf 226} (1983), 205--227.

\bibitem{BuddHuangRussellActa}
C.~J.~Budd, W.~Huang, and R.~D.~Russell,
``Adaptivity with moving grids,''
Acta Numer.\ {\bf 18} (2009), 111--241.

\bibitem{HuangRussellBook}
W.~Huang and R.~D.~Russell,
\emph{Adaptive Moving Mesh Methods},
Applied Mathematical Sciences Vol.~174,
Springer, New York, 2011.

\bibitem{OS1}
K.~Osterwalder and R.~Schrader,
``Axioms for Euclidean Green's functions,''
Commun.\ Math.\ Phys.\ {\bf 31} (1973), 83--112.

\bibitem{OS2}
K.~Osterwalder and R.~Schrader,
``Axioms for Euclidean Green's functions II,''
Commun.\ Math.\ Phys.\ {\bf 42} (1975), 281--305.

\bibitem{LuscherTransfer}
M.~L{\"u}scher,
``Construction of a selfadjoint, strictly positive transfer matrix
for Euclidean lattice gauge theories,''
Commun.\ Math.\ Phys.\ {\bf 54} (1977), 283--292.

\bibitem{OsterwalderSeiler}
K.~Osterwalder and E.~Seiler,
``Gauge field theories on the lattice,''
Ann.\ Phys.\ {\bf 110} (1978), 440--471.

\bibitem{NielsenNinomiyaNoGo}
H.~B.~Nielsen and M.~Ninomiya,
``No-go theorem for regularizing chiral fermions,''
Phys.\ Lett.\ B {\bf 105} (1981), 219--223.

\bibitem{WilsonFermions}
K.~G.~Wilson,
``Quarks and strings on a lattice,''
in \emph{New Phenomena in Subnuclear Physics},
ed.\ A.~Zichichi, Plenum, New York, 1977, pp.\ 69--142.



\end{thebibliography}
\end{document}